\documentclass[11pt,letterpaper]{article}
\usepackage[affil-it]{authblk}

\usepackage{chihao}

\usepackage[margin=1in]{geometry}

\usepackage{graphicx}
\usepackage{threeparttable}
\usepackage{multirow}
\usepackage{ragged2e}
\usepackage{array}
\usepackage{makecell}
\usepackage{physics}
\usepackage{xifthen}
\usepackage{booktabs}
\usepackage{array}
\usepackage{makecell}
\usepackage{tikz}
\usepackage{rotating}
\usetikzlibrary{math}
\usetikzlibrary{positioning}
\usepackage[labelsep=period]{caption}

\usepackage{bbm}
\usepackage{subfig}
\usepackage{afterpage}
\usepackage{pifont}

\usepackage{enumerate}
\usepackage{comment}

\newcommand{\distKL}{\!{KL}}
\newcommand{\distTV}{\!{TV}}
\newcommand{\Bern}{\!{Bern}\xspace}

\newcommand{\GainsFromTrade}{\textsf{Gains from Trade}\xspace}
\newcommand{\TotalGainsFromTrade}{\textsf{Total Gains from Trade}\xspace}
\newcommand{\GFT}{\textnormal{\textsf{GFT}}\xspace}
\newcommand{\SocialWelfare}{\textsf{Social Welfare}\xspace}
\newcommand{\SW}{\textnormal{\textsf{SW}}}
\newcommand{\Profit}{\textnormal{\textsf{Profit}}}
\newcommand{\Regret}{\textnormal{\textsf{Regret}}}

\newcommand{\EstTrade}[1]{\Hat{\textnormal{\textsf{Z}}}_{#1}}

\newcommand{\Horizonal}[1]{\textnormal{\textsf{H}}(#1)}
\newcommand{\ApxHorizonal}[1]{\Tilde{\textnormal{\textsf{H}}}(#1)}
\newcommand{\EstHorizonal}[1]{\Hat{\textnormal{\textsf{H}}}(#1)}
\newcommand{\horizonal}{\textnormal{\textsf{h}}}
\newcommand{\Vertical}[1]{\textnormal{\textsf{V}}(#1)}
\newcommand{\ApxVertical}[1]{\Tilde{\textnormal{\textsf{V}}}(#1)}
\newcommand{\EstVertical}[1]{\Hat{\textnormal{\textsf{V}}}(#1)}
\newcommand{\vertical}{\textnormal{\textsf{v}}}

\newcommand{\EstGFT}[2][]{\Hat{\GFT}_{#1}[#2]}

\newcommand{\Mech}{\+M}
\newcommand{\FirstBest}{\textsf{First-Best}}
\newcommand{\SecondBest}{\textsf{Second-Best}}

\newcommand{\StrongBudgetBalance}{\textsf{Strong Budget Balance}\xspace}
\newcommand{\SBB}{\textsf{SBB}\xspace}
\newcommand{\WeakBudgetBalance}{\textsf{Weak Budget Balance}}
\newcommand{\WBB}{\textsf{WBB}\xspace}
\newcommand{\LocalBudgetBalance}{\textsf{Local Budget Balance}}
\newcommand{\LBB}{\textsf{LBB}\xspace}
\newcommand{\GlobalBudgetBalance}{\textsf{Global Budget Balance}\xspace}
\newcommand{\GBB}{\textsf{GBB}\xspace}

\newcommand{\val}{v}

\newcommand{\Val}[1]{(S^{#1}, B^{#1})}
\newcommand{\SVal}[1]{S^{#1}}
\newcommand{\BVal}[1]{B^{#1}}

\newcommand{\Price}[1]{(P^{#1}, Q^{#1})}
\newcommand{\SPrice}[1]{P^{#1}}
\newcommand{\BPrice}[1]{Q^{#1}}

\newcommand{\Feedback}[1]{(X^{#1}, Y^{#1})}
\newcommand{\SFeedback}[1]{X^{#1}}
\newcommand{\BFeedback}[1]{Y^{#1}}

\newcommand{\Trade}[2][]{Z_{#1}^{#2}}

\newcommand{\Z}{\textnormal{\textsf{Z}}}

\newcommand{\CANDIDATE}{\+C}
\newcommand{\GBBOneBit}{\textnormal{\textsc{GBB-OneBit}}}
\newcommand{\FractalElimination}{\textnormal{\textsc{FractalElimination}}}

\newcommand{\ProfitMax}{\textnormal{\textsc{ProfitMax}}}

\newcommand{\VAL}{\+V}
\newcommand{\VALupperleft}{\+V_{\mathrm{UL}}}
\newcommand{\VALlowerright}{\+V_{\mathrm{LR}}}
\newcommand{\VALcorner}{\+V_{\mathrm{cor}}}
\newcommand{\valmajority}{v_{\mathrm{maj}}}

\newcommand{\VALhorleft}{\+V_{\mathrm{HL}}}
\newcommand{\VALhorright}{\+V_{\mathrm{HR}}}
\newcommand{\VALmajority}{\+V_{\mathrm{maj}}}
\newcommand{\VALver}{\+V_{\mathrm{ver}}}

\newcommand{\ACTION}{\+A}

\renewcommand{\Tilde}{\widetilde}
\renewcommand{\Bar}{\overline}
\renewcommand{\Hat}{\widehat}

\newcommand{\ignore}[1]{}
\renewcommand{\setminus}{\,\backslash\,}
\newcommand{\colorref}[2]{{\hypersetup{linkcolor=#1}\ref{#2}}}
\newcommand{\blackref}[1]{{\hypersetup{linkcolor=black}\ref{#1}}}

\newcommand{\blackeqref}[1]{{\hypersetup{linkcolor=black}\eqref{#1}}}

\newcommand{\tO}{\Tilde{\+O}}
\newcommand{\tOmega}{\Tilde{\Omega}}
\newcommand{\tTheta}{\Tilde{\Theta}}

\algrenewcommand\algorithmicrequire{\textbf{Input:}}
\algrenewcommand\algorithmicensure{\textbf{Output:}}

\makeatletter
\def\term{\@ifnextchar[\term@optarg\term@noarg}%]
\def\term@optarg[#1]#2{%
  \textup{#1}%
  \def\@currentlabel{#1}%
  \def\cref@currentlabel{[][2147483647][]#1}%
  \cref@label[term]{#2}}
\def\term@noarg#1{%
  \refstepcounter{termcounter}%
  \textup{\thetermcounter}%
  \cref@label[term]{#1}}
\makeatother

\def\colorful{1}
\ifnum\colorful=1

\fi
\ifnum\colorful=0

\fi

\usepackage{pict2e}

\makeatletter
\DeclareRobustCommand{\Circle}{%
  \mathbin{\mathpalette\on@ntimes\relax}%
}
\newcommand{\on@ntimes}[2]{%
  \vcenter{\hbox{%
    \sbox0{\m@th$#1\otimes$}%
    \setlength\unitlength{\wd0}%
    \begin{picture}(1,1)
    \linethickness{0.5pt}
    \put(.5,.5){\circle{.8}}
    \end{picture}%
  }}%
}
\makeatother

\def\version{Full}

\title{Tight Regret Bounds for Fixed-Price Bilateral Trade}

\ifthenelse{\equal{\version}{Full}}
{\author{
Houshuang Chen\thanks{Shanghai Jiao Tong University. Email: {\tt chenhoushuang@sjtu.edu.cn}}
\and
Yaonan Jin\thanks{Huawei's Taylor Lab. Email: {\tt jinyaonan@huawei.com}}
\and
Pinyan Lu\thanks{Shanghai University of Finance and Economics, Laboratory of Interdisciplinary Research of Computation and Economics (SUFE), \& Huawei's Taylor Lab. Email: {\tt lu.pinyan@mail.shufe.edu.cn}}
\and
Chihao Zhang\thanks{Shanghai Jiao Tong University. Email: {\tt chihao@sjtu.edu.cn}}
}}
{\ifthenelse{\equal{\version}{Anonymous}}
{\author{
Anonymous Submission
}}
{}}

\date{}

\begin{document}
	
\maketitle
\begin{abstract}
We examine fixed-price mechanisms in bilateral trade through the lens of regret minimization. Our main results are twofold.
(i)~For independent values, a near-optimal $\widetilde{\Theta}(T^{2/3})$ tight bound for {\sf Global Budget Balance} fixed-price mechanisms with two-bit/one-bit feedback.
(ii)~For correlated/adversarial values, a near-optimal $\Omega(T^{3/4})$ lower bound for {\sf Global Budget Balance} fixed-price mechanisms with two-bit/one-bit feedback, which improves the best known $\Omega(T^{5/7})$ lower bound obtained in the work \cite{BCCF24} and, up to polylogarithmic factors, matches the $\widetilde{\mathcal{O}}(T^{3 / 4})$ upper bound obtained in the same work.
Our work in combination with the previous works \cite{CCCFL24mor, CCCFL24jmlr, AFF24, BCCF24} (essentially) gives a thorough understanding of regret minimization for fixed-price bilateral trade.

En route, we have developed two technical ingredients that might be of independent interest:
(i)~A novel algorithmic paradigm, called {\em fractal elimination}, to address one-bit feedback and independent values.
(ii)~A new {\em lower-bound construction} with novel proof techniques, to address the {\sf Global Budget Balance} constraint and correlated values.
\end{abstract}

\thispagestyle{empty}
\newpage
\setcounter{tocdepth}{2}
{\hypersetup{linkcolor=black}\tableofcontents}
\thispagestyle{empty}
\newpage
\setcounter{page}{1}

\section{Introduction}
\label{sec:intro}

We address a classic problem in Mechanism Design, maximizing \textit{economic efficiency} in repeated bilateral trade:
In each round $t \in [T]$, a (new) seller and a (new) buyer seek to trade an indivisible item, which has value $\SVal{t}$ to the seller and value $\BVal{t}$ to the buyer.
There are two standard metrics of economic efficiency:\\
1.\ {\GainsFromTrade}, defined as $\GFT = \sum_{t \in [T]} \GFT^{t} = \sum_{t \in [T]} (\BVal{t} - \SVal{t}) \cdot \Trade[]{t}$.\\
2.\ {\SocialWelfare}, defined as $\SW = \sum_{t \in [T]} \SW^{t} = \sum_{t \in [T]} (\BVal{t} \cdot \Trade[]{t} + \SVal{t} \cdot (1 - \Trade[]{t}))$.\\
Here, $\Trade[]{t} = \Trade[]{t}(\SVal{t}, \BVal{t}) \in \{0, 1\}$ denotes the trade outcome---either a success or a failure.

As is standard in Mechanism Design, there are three models for generating values $\Val{t}_{t \in [T]}$, listed below from most to least general.
\begin{flushleft}
\begin{itemize}
    \item \textbf{Adversarial Values:}
    An (oblivious) adversary determines an (arbitrary) $2T$-dimensional $[0, 1]^{2T}$- supported joint distribution $\+D$; then, the values $\Val{t}_{t \in [T]}$ across all rounds are drawn from it.\footnote{\label{footnote:adversarial}This definition emphasizes the generality of the ``adversarial values'' model over the ``correlated/independent values'' models. From the perspective of (additive) regret minimization, however, the worst-case distribution $\+D_{*}$ can be assumed, without loss of generality, to place all its mass on a finite set of $2T$ discrete points $(\SVal{t}_{*}, \BVal{t}_{*}) \in [0, 1]^{2T}$.}
    
    \item \textbf{Correlated Values:}
    An adversary determines an (arbitrary) two-dimensional $[0, 1]^{2}$-supported joint distribution $\+D$; then, the values $\Val{t}$ in each round $t \in [T]$ are drawn i.i.d.\ from it.
    
    \item \textbf{Independent Values:}
    This is identical to ``correlated values'', except that $\+D$ is further required to be a product distribution $\+D_{S} \bigotimes \+D_{B}$, making all the $2T$ values $\Val{t}_{t \in [T]}$ mutually independent.
\end{itemize}
% \ctodo{define $\+D_{S}$ and $\+D_{B}$.\\
% \yj{\orange{Marked in orange (P6 \& P20)}}}
\end{flushleft}
For reference, the ``independent values'' model is arguably the most canonical model, dating back to seminal works \cite{V61, MS83}, while the other two models are more general and have also received growing attention recently \cite{BSZ06, CCCFL24mor, CCCFL24jmlr, AFF24, BCCF24, DS24}.

An ideal mechanism should trade the item whenever economic efficiency can improve, i.e., $\SVal{t} \le \BVal{t}$, achieving \textit{ex-post efficiency}.
However, the celebrated Myerson-Satterthwaite theorem \cite{MS83} asserts that no \textit{economically viable} mechanism can ensure ex-post efficiency, even in highly restricted scenarios; see \Cref{sec:intro:related_work} for a more rigorous assertion.
Here, economic viability requires:\\
1.\ Individual Rationality (IR)---each agent receives a nonnegative utility by reporting his/her true value.\\
2.\ Incentive Compatibility (IC)---each agent has no incentive to misreport his/her true value.\\
3.\ Budget Balance (BB)---a mechanism cannot subsidize either agent, i.e., it cannot run a deficit.\footnote{Compared with, e.g., revenue maximization in one-sided market \cite{M81}, the BB constraint is (more) crucial for efficiency maximization in bilateral trade. E.g., the VCG mechanism \cite{V61} (which trades the item from seller to buyer at prices $\SVal{}$ and $\BVal{}$, respectively, whenever $\SVal{} \le \BVal{}$) ensures ex-post efficiency and satisfies both EIR and DSIC, but violates BB.}\\
The impossibility result in \cite{MS83} has motivated significant  work in mechanism design.
In this work, we evaluate mechanisms from the perspective of \textit{(additive) regret minimization}, by which {\SocialWelfare} and {\GainsFromTrade} are interchangeable since their gap $\SW - \GFT = \sum_{t \in [T]} \SVal{t}$ is mechanism-independent.\footnote{From the perspective of \textit{(multiplicative) efficiency approximation}, however, we must distinguish between {\SocialWelfare} and {\GainsFromTrade}; see \Cref{sec:intro:related_work} for further discussions.}
Without loss of generality, we adopt {\GainsFromTrade} for our presentation.

Regarding the IR and IC constraints, we adopt the strongest notions, \textit{Ex-Post Individual Rationality (EIR)} and \textit{Dominant-Strategy
Incentive Compatibility (DSIC)}.\footnote{\label{footnote:EIR-DSIC}EIR ensures nonnegative ex-post utility for each agent, while DSIC ensures that truthful reporting is a dominant strategy regardless of the other agent's behavior.}
It turns out that EIR, DSIC (and BB) mechanisms are characterized by the family of \textit{fixed-price mechanisms} \cite{HR87, CKLT16, DDFS25}:\footnote{\label{footnote:fixed-price}For more details about this assertion, we encourage the interested reader to reference \cite{HR87, CKLT16, DDFS25}. Basically:
(i)~The family of EIR, DSIC, and \textit{\StrongBudgetBalance} mechanisms is \textit{identical} to the family of fixed-price mechanisms.
(ii)~The family of EIR, DSIC, and \textit{\WeakBudgetBalance} mechanisms is \textit{strictly broader}---but for efficiency maximization, every such mechanism is dominated by a fixed-price mechanism. So for our purpose, we can focus on fixed-price mechanisms.}
Such a mechanism posts two possibly randomized prices $\Price{t}$ in each round $t \in [T]$, one for the seller $\SPrice{t}$ and one for the buyer $\BPrice{t}$. A trade occurs if both agents accept their respective prices: $\Trade[]{t} = {\bb 1}[\SVal{t} \le \SPrice{t} \land \BPrice{t} \le \BVal{t}] \in \{0, 1\}$.
This yields the {\GainsFromTrade} $\GFT^{t} = (\BVal{t} - \SVal{t}) \cdot \Trade[]{t}$ and the \textit{profit} $\Profit^{t} = (\BPrice{t} - \SPrice{t}) \cdot \Trade[]{t}$.
Throughout this work, we use the term \textit{mechanisms} to refer to \textit{fixed-price mechanisms} when no ambiguity arises.

Regarding the BB constraint, prior works have examined two notions, a stricter \textit{{\LocalBudgetBalance} ({\LBB})} constraint \cite{MS83} and a looser \textit{{\GlobalBudgetBalance} ({\GBB})} constraint \cite{BCCF24}.\footnote{\label{footnote:single-round}In single-round bilateral trade $T = 1$, the {\GBB} and {\WBB} constraints are essentially identical, but they are still weaker than the {\SBB} constraint.}
\begin{flushleft}
\begin{itemize}
    \item {\LocalBudgetBalance}:
    This can be further differentiated into a stricter \textit{{\StrongBudgetBalance} ({\SBB})} constraint and a looser \textit{{\WeakBudgetBalance} ({\WBB})} constraint.
    \begin{itemize}
        \item {\StrongBudgetBalance}: 
        $\Profit^{t} = 0$ (or essentially $\SPrice{t} = \BPrice{t}$) ex-post, $\forall t \in [T]$.\\
        This requires \textit{zero profit} $\Profit^{t} = 0$ in each round, almost surely.
        
        \item {\WeakBudgetBalance}: 
        $\Profit^{t} \ge 0$ (or essentially $\SPrice{t} \le \BPrice{t}$) ex-post, $\forall t \in [T]$.\\
        This requires \textit{nonnegative profit} $\Profit^{t} \ge 0$ in each round, almost surely.
    \end{itemize}
    
    \item {\GlobalBudgetBalance}:
    $\sum_{t \in [T]} \Profit^{t} \ge 0$ ex-post.\\
    This requires \textit{nonnegative total profit} $\sum_{t \in [T]} \Profit^{t} \ge 0$, almost surely.
\end{itemize}
\end{flushleft}

\begin{remark}[Budget Balance]
\label{rmk:BB}
The {\LBB} constraint imposes ``local'' restrictions on individual rounds $t \in [T]$. While prior works \cite{MS83, CCCFL24mor, CCCFL24jmlr, AFF24, BCCF24, DS24} studied {\SBB} and {\WBB} separately, they yield the same tight bound in every specific context, up to polylogarithmic factors. (E.g., if the best possible upper bound for {\SBB} mechanisms is $\tO(T^{2 / 3})$, then the best possible lower bound for {\WBB} mechanisms would be $\Omega(T^{2 / 3})$.) To reflect this unity, we unify them under {\LBB} in our presentation.

In contrast, {\GBB} relaxes such ``local'' restrictions to a ``global'' one over all rounds $t \in [T]$.
\end{remark}

\begin{remark}[Benchmarks]
\label{rmk:benchmark}
Following prior works \cite{CCCFL24mor, CCCFL24jmlr, AFF24, BCCF24}, we adopt the \textit{optimal (Budget Balance) fixed prices in hindsight} $(p^{*}, q^{*})$ as the benchmark for our mechanisms.\footnote{For completeness, there is one exception in the literature \cite[Section~6]{BCCF24} that studies a stronger benchmark.}
This choice of benchmark is justified because the optimal prices under the looser {\GBB} constraint $(p_{\GBB}^{*}, q_{\GBB}^{*})$ still satisfy the stricter {\SBB} constraint $p_{\GBB}^{*} = q_{\GBB}^{*}$; see \Cref{def:benchmark} for a more rigorous assertion.

In contrast, both the design and analysis of a mechanism critically depend on whether the {\LBB} or {\GBB} constraint is imposed.
\end{remark}

As is usual in Online Optimization, the design and analysis of a mechanism depend on the underlying \textit{feedback model}---the trade information revealed at the end of each round $t \in [T]$.
Prior works \cite{CCCFL24mor, CCCFL24jmlr, AFF24, BCCF24} have studied two feedback models, the more informative \textit{full feedback} and the less informative \textit{partial feedback}.
\begin{flushleft}
\begin{itemize}
    \item \textbf{Full Feedback:}
    This reveals both agents' values $\Val{t} \in [0, 1]^{2}$.
    
    \item \textbf{Partial Feedback:}
    This can be further differentiated into the more informative \textit{two-bit feedback} (i.e., both agents' intentions to trade) and the less informative \textit{one-bit feedback} (i.e., the trade's success/failure outcome).
    \begin{itemize}
        \item \textbf{Two-Bit Feedback:}
        $\Feedback{t} = ({\bb 1}[\SVal{t} \le \SPrice{t}],\ {\bb 1}[\BPrice{t} \le \BVal{t}]) \in \{0, 1\}^{2}$.
        
        \item \textbf{One-Bit Feedback:}
        $\Trade[]{t} = {\bb 1}[\SVal{t} \le \SPrice{t} \land \BPrice{t} \le \BVal{t}] \equiv \SFeedback{t} \land \BFeedback{t} \in \{0, 1\}$.
    \end{itemize}
\end{itemize}
\end{flushleft}

\begin{remark}[Feedback Models]
Regarding partial feedback, while prior works \cite{CCCFL24mor, CCCFL24jmlr, AFF24, BCCF24} studied two-bit feedback and one-bit feedback separately, these two models yield the same tight bound in every specific context, up to polylogarithmic factors. (E.g., if the best possible upper bound for one-bit feedback is $\tO(T^{2 / 3})$, then the best possible lower bound for two-bit feedback would be $\Omega(T^{2 / 3})$.) To reflect this unity, we unify them under \textit{partial feedback} in our presentation.

Additionally, in \Cref{sec:prelim,sec:GBB-independent:LB}, we will introduce and study a new feedback model called \textit{semi feedback}, which serves as an intermediate between full feedback and partial feedback. Notably, this model was independently introduced in the concurrent work \cite{LCM25a}, where it is termed \textit{asymmetric feedback}.
\end{remark}

\noindent
\textbf{Roadmap.}
In total, there are $3 \times 2 \times 2 = 12$ well-defined contexts. \Cref{sec:intro:full,sec:intro:partial} provide a comprehensive survey, but \textbf{we concentrate on $\mathbf{3 \times 1 \times 1 = 3}$ contexts}: ``{\GBB} partial-feedback mechanisms for independent/correlated/adversarial values''. These $3$ contexts stand out for the following reasons.\\
1.\ As argued in all prior works \cite{CCCFL24mor, CCCFL24jmlr, AFF24, BCCF24}, partial feedback is (much) more realistic and technically (much) more challenging than full feedback.\\
2.\ Unfortunately, {\LBB} partial-feedback mechanisms cannot guarantee sublinear regret, even in the most restricted but canonical ``independent values'' setting (\Cref{tbl:partial-SBB-WBB}). In stark contrast, {\GBB} partial-feedback mechanisms can guarantee sublinear regret (\Cref{tbl:partial-GBB}).\\
3.\ Prior works \cite{CCCFL24mor, CCCFL24jmlr, AFF24, BCCF24} have cultivated a thorough understanding of all the other $9$ contexts (\Cref{tbl:full,tbl:partial-SBB-WBB}), leaving exactly these $3$ contexts unresolved (\Cref{tbl:partial-GBB}).

Despite these technical challenges, we establish tight regret bounds (up to polylogarithmic factors) for all three of these standout contexts.
Our results rely on two technical ingredients that might be of independent interest, detailed in \Cref{sec:intro:technique}.
Finally, we conclude with a discussion of related works and relevant issues in \Cref{sec:intro:related_work}.

% \red{1.\ Both constraints \textit{\StrongBudgetBalance} versus \textit{\WeakBudgetBalance} are always indifferent in respect to tight bounds, so we shall unify them into \textit{\LocalBudgetBalance}.\\
% 2.\ Both feedback models \textit{two-bit feedback} versus \textit{one-bit feedback} are always indifferent in respect to tight bounds, so we shall unify them into \textit{partial feedback}.}

\subsection{Full-Feedback Mechanisms}
\label{sec:intro:full}

\begin{table}[t]
    \centering
    \begin{tabular}{|c||>{\centering}p{5.5cm}|>{\centering\arraybackslash}p{5.5cm}|}
        \hline
        \rule{0pt}{12pt} & {\LocalBudgetBalance} & {\GlobalBudgetBalance} \\ [1pt]
        \hline
        \hline
        \rule{0pt}{12pt}Independent & \multicolumn{2}{c|}{\multirow{2}*{\rule{0pt}{16pt}$T^{1 / 2}$ \cite{CCCFL24mor,BCCF24} [Thm~\ref{thm:appendix:GBB}]}} \\ [1pt]
        \cline{1-1}
        \rule{0pt}{12pt}Correlated & \multicolumn{2}{c|}{} \\ [1pt]
        \cline{1-2}
        \rule{0pt}{12pt}Adversarial & $T$ \cite{CCCFL24mor,AFF24} &  \\ [1pt]
        \hline
    \end{tabular}
    \caption{\label{tbl:full}Regret bounds of \textit{{\LBB}/{\GBB} full-feedback} fixed-price mechanisms, up to polylogarithmic factors.}
\end{table}

To begin with, we survey the regret bounds of \textit{full-feedback} mechanisms.
Prior work has largely completed the puzzle (see \Cref{tbl:full} for a summary), so we provide only a brief review here. Expert readers may proceed directly to \Cref{sec:intro:partial}.

In the ``independent/correlated values'' settings, the previous work \cite{CCCFL24mor} (the conference version appeared in EC'21) has made substantial progress.
They designed an $\+O(T^{1 / 2})$ regret {\LBB} mechanism for correlated values \cite[Theorem~1]{CCCFL24mor} and established a matching $\Omega(T^{1 / 2})$ lower bound even for independent values \cite[Theorem~2]{CCCFL24mor}.\footnote{\label{footnote:CCCFL24mor-LB}More precisely, while \cite[Theorems~2 and 6]{CCCFL24mor} stated the $\Omega(T^{1 / 2})$ and $\Omega(T)$ lower bounds for {\SBB} mechanisms, the proofs hold more generally for {\WBB} mechanisms; see \Cref{thm:appendix:GBB,thm:appendix:WBB} in \Cref{sec:appendix:GBB,sec:appendix} for details.}
Indeed, this $\Omega(T^{1 / 2})$ lower bound can be easily extended to the most technically challenging context under consideration, {\GBB} mechanisms for independent values (see \Cref{thm:appendix:GBB} in \Cref{sec:appendix:GBB} for details).\footnote{Note that \Cref{thm:appendix:GBB} essentially reuses the lower-bound construction from \cite[Theorems~2 and 4]{CCCFL24mor}, with a symmetric proof.
We supplement this by showing that under this construction, the {\GBB} constraint reduces to the {\WBB} constraint.}
Thus, the tight regret bound is $\Theta(T^{1 / 2})$ for all four contexts: ``{\LBB}/{\GBB} mechanisms'' $\times$ ``independent/correlated values''.

In the most general ``adversarial values'' setting, however, {\LBB} mechanisms do not admit sublinear regret; a linear $\Omega(T)$ lower bound was established in \cite[Theorem~7]{CCCFL24mor} and also in \cite[Theorem~1]{AFF24} (the conference version appeared in NeurIPS'22).\footnote{\label{footnote:linear-lower-bound}Specifically, \cite{CCCFL24mor} first proved a linear $\Omega(T)$ lower bound for {\SBB} mechanisms, and \cite{AFF24} later extended it to {\WBB} mechanisms. Here, we focus on the \textit{orders} of these bounds rather than their precise constants. For a detailed \textit{quantitative analysis}, we refer the reader to \cite{AFF24}.}

The recent work \cite{BCCF24} circumvented this impossibility result by relaxing the ``local'' {\LBB} constraint to the ``global'' {\GBB} constraint, and designed an $\tO(T^{1 / 2})$ regret {\GBB} mechanism \cite[Theorem~4.2]{BCCF24}. Moreover, a nearly matching $\Omega(T^{1 / 2})$ lower bound was established in \cite[Theorem~4.4]{BCCF24}.\footnote{For {\GBB} full-feedback mechanisms, \cite[Theorem~4.4]{BCCF24} established an $\Omega(T^{1 / 2})$ lower bound for ``correlated values''. In contrast, \Cref{thm:appendix:GBB} will extend this $\Omega(T^{1 / 2})$ lower bound to ``independent values''.}

\vspace{.1in}
\noindent
\textbf{Summary of \Cref{sec:intro:full}.}
\Cref{tbl:full} summarizes the tight regret bounds of \textit{full-feedback} mechanisms:
(i)~For ``independent/correlated values'', we need not distinguish between {\LBB} and {\GBB} mechanisms, as both achieve the tight regret bound $\Theta(T^{1 / 2})$.
(ii)~For ``adversarial values'', only {\GBB} mechanisms admit sublinear regret, specifically $\tTheta(T^{1 / 2})$.

\subsection{Partial-Feedback Mechanisms}
\label{sec:intro:partial}

We now examine the tight regret bounds of \textit{partial-feedback} mechanisms.
In fact, {\LBB} and {\GBB} mechanisms exhibit significantly different regret minimization capabilities (see \Cref{tbl:partial-SBB-WBB,tbl:partial-GBB} for comparison). We thus discuss them separately.

\subsection*{{\LBB} Mechanisms are Intractable}

{\LBB} partial-feedback mechanisms \textit{cannot} achieve sublinear regret (\Cref{tbl:partial-SBB-WBB}):
\cite[Theorem~6]{CCCFL24mor} established a linear $\Omega(T)$ lower bound even for the most restricted ``independent values'' setting,\textsuperscript{\ref{footnote:CCCFL24mor-LB}}
and \cite[Theorem~6]{AFF24} improved the hidden constant in this bound for the most general ``\textit{adversarial values}'' setting.

These strong negative results pose significant theoretical and practical challenges for {\LBB} mechanisms. To address this, \cite{BCCF24} introduced the {\GBB} constraint, which relaxes ``local'' restrictions on individual rounds $t \in [T]$ to a ``global'' restriction over all rounds.

\begin{table}[t]
    \centering
    \begin{tabular}{|c||>{\centering\arraybackslash}p{5.5cm}|}
        \hline
        \rule{0pt}{12pt} & {\LocalBudgetBalance} \\ [1pt]
        \hline
        \hline
        \rule{0pt}{12pt}Independent &  \\ [1pt]
        \cline{1-1}
        \rule{0pt}{12pt}Correlated & $T$ \cite{CCCFL24mor,AFF24} [Thm~\ref{thm:appendix:WBB}] \\ [1pt]
        \cline{1-1}
        \rule{0pt}{12pt}Adversarial &  \\ [1pt]
        \hline
    \end{tabular}
    \caption{\label{tbl:partial-SBB-WBB}Linear regret lower bounds of \textit{{\LBB} partial-feedback} fixed-price mechanisms.}
\end{table}

\subsection*{{\GBB} Mechanisms are Tractable}

{\GBB} partial-feedback mechanisms achieve sublinear regret (\Cref{tbl:partial-GBB}).
We study the ``independent values'' setting and the ``correlated/adversarial values'' settings separately, as they lead to different regret bounds.

\vspace{.1in}
\noindent
\textbf{Independent Values.}
Although the ``independent values'' setting is arguably the most canonical model for value generation, a significant gap remains between the best known regret bounds in prior work:\footnote{More precisely, \cite[Theorem~5.4]{BCCF24} applies to ``{\GBB} one-bit-feedback mechanisms for \textit{adversarial values}'', and \Cref{thm:appendix:GBB} applies to ``{\GBB} \textit{full-feedback} mechanisms for independent values''.}\\
The best known upper bound $\tO(T^{3 / 4})$ follows from \cite[Theorem~5.4]{BCCF24} for \textit{``adversarial values''}.\\
The best known lower bound $\Omega(T^{1 / 2})$ follows from \Cref{thm:appendix:GBB} for \textit{full feedback}.

We make two main contributions.
On the upper-bound side, we design an $\tO(T^{2 / 3})$ regret {\GBB} partial-feedback mechanism (\Cref{thm:GBB-independent:one,alg:GBB-independent:one}).
At a high level, our mechanism builds on the \textit{two-phase meta mechanism} framework from \cite[Section~3]{BCCF24}. Our technical contribution lies in the design of Phase~\ref{alg:GBB-independent:one:exploration}, which seeks a good enough approximation to the optimal action $(p^{*}, q^{*})$ using a novel algorithmic paradigm called \textit{fractal elimination}.\\
\Comment{A detailed overview of this algorithmic paradigm can be found in the ``Ingredient~1'' part of \Cref{sec:intro:technique}.}

\noindent
(i)~Phase~\ref{alg:GBB-independent:one:profit} aims to accumulate sufficient profit, i.e., $\tOmega(T^{2 / 3})$ (at the cost of $\tO(T^{2 / 3})$ regret), thereby providing flexibility in Phase~\ref{alg:GBB-independent:one:exploration} with respect to the {\GBB} constraint.
To this end, Phase~\ref{alg:GBB-independent:one:profit} can simply black-box invoke the {\ProfitMax} subroutine from \cite[Algorithm~1 and Section~5.1]{BCCF24}.

\noindent
(ii)~Phase~\ref{alg:GBB-independent:one:exploration} invokes the {\FractalElimination} subroutine (\Cref{alg:GBB-independent:fractal}) to find a good enough approximation to the optimal action $(p^{*}, q^{*})$.
Standard discretization of the action space yields $|\CANDIDATE_{0}| = \tTheta(T^{1 / 3})$ many \textit{candidates} for such an approximation.
{\FractalElimination} then proceeds in $L \approx \log(|\CANDIDATE_{0}|)$ many stages; each stage $\ell \in [L]$ refines the candidate set $\CANDIDATE_{\ell} \subseteq \CANDIDATE_{\ell - 1}$ by playing not only the current candidates $\CANDIDATE_{\ell}$ but also other actions.
Ultimately, all remaining candidates in $\CANDIDATE_{L}$ are good enough relative to the optimal candidate and/or the optimal action $(p^{*}, q^{*})$.
By design, {\FractalElimination} operates within the profit budget accumulated in Phase~\ref{alg:GBB-independent:one:profit}, thus ensuring the {\GBB} constraint over all rounds.

On the lower-bound side, we show a nearly matching $\Omega(T^{2 / 3})$ lower bound (\Cref{thm:GBB-independent:LB}).\footnote{Indeed, this $\Omega(T^{2 / 3})$ lower bound holds even for the more informative \textit{semi feedback} (\Cref{sec:prelim}) and under the \textit{density-boundedness} assumption (\Cref{asm:density}).}
Together, these results resolve the gap left by prior work.

\begin{remark}
In fact, achieving $\tO(T^{2 / 3})$ upper bound is straightforward with \textit{two-bit feedback}---we can naturally adapt \cite[Algorithm~3]{CCCFL24mor} within the two-phase meta mechanism framework from \cite[Section~3]{BCCF24}. However, achieving $\tO(T^{2 / 3})$ upper bound with \textit{one-bit feedback} is significantly more challenging---this is our main contribution, based on the \textit{fractal elimination} paradigm.

For the lower-bound side, we essentially reuse the construction from \cite[Theorem~4]{CCCFL24mor}, which established an $\Omega(T^{2 / 3})$ lower bound for ``{\LBB} partial-feedback mechanisms for independent values''.
We supplement this by showing that, under this construction, no {\GBB} mechanism can sacrifice a certain amount of {\GainsFromTrade} in earlier rounds to recoup the same amount (or more) in later rounds. Consequently, the {\GBB} constraint reduces to the {\LBB} constraint.
\end{remark}

\begin{table}[t]
    \centering
    \begin{tabular}{|c||>{\centering}p{5.5cm}|>{\centering\arraybackslash}p{5.5cm}|}
        \hline
        \rule{0pt}{12pt} & \multicolumn{2}{c|}{{\GlobalBudgetBalance}} \\ [1pt]
        \cline{2-3}
        \rule{0pt}{12pt} & Previous Lower/Upper Bounds & Current Tight Bounds \\ [1pt]
        \hline
        \hline
        \rule{0pt}{12pt}Independent & $[T^{1 / 2}, T^{3 / 4}]$ [Thm~\ref{thm:appendix:GBB}] \cite{BCCF24} & $T^{2 / 3}$ [Thms~\ref{thm:GBB-independent:LB} \& \ref{thm:GBB-independent:one}] \\ [1pt]
        \hline
        \rule{0pt}{12pt}Correlated & \multirow{2}*{\rule{0pt}{16pt}$[T^{5 / 7}, T^{3 / 4}]$ \cite{BCCF24}} & \multirow{2}*{\rule{0pt}{16pt}$T^{3 / 4}$ [Thm~\ref{thm:GBB-correlated:LB}] \cite{BCCF24}} \\ [1pt]
        \cline{1-1}
        \rule{0pt}{12pt}Adversarial &  &  \\ [1pt]
        \hline
    \end{tabular}
    \caption{\label{tbl:partial-GBB}Regret bounds of \textit{{\GBB} partial-feedback} fixed-price mechanisms, up to polylogarithmic factors.}
\end{table}

\vspace{.05in}
\noindent
\textbf{Correlated/Adversarial Values.}
Previous work \cite{BCCF24} has made considerable progress, including an $\tO(T^{3 / 4})$ upper bound for ``adversarial values'' \cite[Theorem~5.4]{BCCF24} and an $\Omega(T^{5 / 7})$ lower bound even for ``correlated values'' \cite[Theorem~5.5]{BCCF24}. However, this progress remains incomplete.

Here, our contribution is an $\Omega(T^{3 / 4})$ lower bound (again) for ``correlated values'' (\Cref{thm:GBB-correlated:LB}).\footnote{Indeed, this $\Omega(T^{3 / 4})$ lower bound holds even if we impose the \textit{density-boundedness} assumption (\Cref{asm:density}).}
This matches the aforementioned $\tO(T^{3 / 4})$ upper bound up to polylogarithmic factors, thus resolving the main open problem posed by \cite{BCCF24}.
The technical challenge stems largely from the {\GBB} constraint, which introduces dependencies among different rounds; our analysis addresses this with a novel approach.\\
\Comment{A detailed description of our lower-bound proof can be found in the ``Ingredient~2'' part of \Cref{sec:intro:technique}.}

\vspace{.1in}
\noindent
\textbf{Summary of \Cref{sec:intro:partial}.}
\Cref{tbl:partial-SBB-WBB,tbl:partial-GBB} summarize the tight regret bounds of \textit{partial-feedback} mechanisms:
(i)~{\LBB} mechanisms do not admit sublinear regret.
(ii)~{\GBB} mechanisms always admit sublinear regret, specifically $\tTheta(T^{2 / 3})$ for ``independent values'' and $\tTheta(T^{3 / 4})$ for ``correlated/adversarial values''.

\subsection{Technical Overview}
\label{sec:intro:technique}

To achieve the results presented in \Cref{sec:intro:partial}, we have developed several technical ingredients that may be of independent interest.
We highlight the two most important contributions:
\begin{itemize}
    \item \textbf{Ingredient 1:} A novel algorithmic paradigm called \textit{fractal elimination} for devising one-bit-feedback mechanisms for ``independent values''. It is inspired by the notion of \textit{fractals}, or more concretely, the famous \textit{Sierpi\'{n}ski triangle} (\Cref{{fig:intro-fractal-1}}); this connection will be clear from the later materials.
    
    \item \textbf{Ingredient 2:} A new \textit{lower-bound construction} for ``correlated values''.
\end{itemize}
For readability, we omit many minor technical details.

\subsection*{Ingredient~1: Fractal Elimination for One-Bit Feedback and Independent Values}

We begin by explaining why one-bit feedback greatly complicates the design of low-regret mechanisms for ``independent values $(S, B) \sim \+D_{S} \bigotimes \+D_{B}$''. The algorithmic framework proposed in \cite{BCCF24} allows us to first accumulate sufficient profit and then discretize the action space near the diagonal.
For an action $(p, q) \in [0, 1]^{2}$ with $p \approx q$, we can express the expected {\GainsFromTrade} as follows \cite[Lemma~1]{CCCFL24mor}:
\begin{align*}
    \textstyle
    \GFT(p, q)
    ~=~ \int_{0}^{p} \+D_{S}(x) \d x \cdot (1 - \+D_{B}(q)) + \+D_{S}(p) \cdot \int_{q}^{1} (1 - \+D_{B}(y)) \d y + \mbox{[minor terms]}.
\end{align*}
Here we abuse notation: $\+D_{S}(p) \defeq {\bb P}_{S \sim \+D_{S}}[S \le p]$ and $\+D_{B}(q) \defeq 1 - {\bb P}_{B \sim \+D_{B}}[B \ge q] = {\bb P}_{B \sim \+D_{B}}[B < q]$ denote the corresponding \textit{cumulative distribution functions (CDF's)}.

To adapt a classic \textit{Multi-Armed Bandit (MAB)} algorithm into a low-regret mechanism, we require ``good enough'' estimates of the $\GFT(p,q)$ values.
The query complexity for this task depends critically on the available feedback.
\begin{itemize}
    \item \textbf{Two-Bit Feedback $\Feedback{t}$:}
    As observed in \cite{CCCFL24mor}, with access to the seller's intention to trade $\SFeedback{t}$, only $\tO(\eps^{-2})$ queries/rounds are needed to estimate the integral $\int_{0}^{p} \+D_{S}(x) \d x$ within error $\eps > 0$, \textit{pointwise over the entire interval $p \in [0, 1]$}; similarly for $\int_{q}^{1} (1 - \+D_{B}(y)) \d y$.
    Thus, estimating $\GFT(p,q)$ reduces to estimating $\+D_{S}(p)$ and $1 - \+D_{B}(q)$.
    Crucially, the two bits $\SFeedback{t}$ and $\BFeedback{t}$ themselves serve as unbiased estimators for these quantities.
    
    Feeding $\SFeedback{t}$ and $\BFeedback{t}$ into a standard MAB algorithm directly yields an $\+O(T^{2 / 3})$ regret {\GBB} mechanism (building on the two-phase meta-mechanism framework from \cite[Section~3]{BCCF24}).
    
    \item \textbf{One-Bit Feedback $\Trade[]{t}$:}
    In contrast, with access only to the trade outcome $\Trade[]{t}$, we lose the tailored unbiased estimators for $\+D_{S}(p)$, $1 - \+D_{B}(q)$ etc., which constitutes our main technical challenge.
    Prior work \cite[Algorithm~4]{CCCFL24mor} demonstrated how to construct unbiased estimators of $\GFT(p,q)$ from the less informative $\Trade[]{t}$ queries.
    They showed that $\tO(\eps^{-2})$ queries still suffice for $\eps$-approximations.
    Although their estimators are optimal in query complexity, the resulting $\tO(T^{3 / 4})$ regret mechanism is suboptimal in regret,\footnote{More precisely, \cite[Algorithm~4]{CCCFL24mor} is a \textit{{\WBB} one-bit-feedback} mechanism for ``independent values''. Here we consider a \textit{{\GBB} one-bit-feedback} mechanism adapted straightforwardly from \cite[Algorithm~4]{CCCFL24mor} using the two-phase meta mechanism framework from \cite[Section~3]{BCCF24}.}
    because their estimators inherently incur \textit{constant regret} per query/round, rather than \textit{diminishing regret}.
\end{itemize}
This leads to the core challenge for designing a low-regret {\GBB} one-bit-feedback mechanism:

\vspace{-.075in}
\begin{quote}
    \textit{How can we \textbf{regret-optimally} estimate the $\GFT(p,q)$ values?}
\end{quote}

\vspace{-.075in}
\noindent
Our main contribution is an \textit{elimination-based multi-stage} algorithm, {\FractalElimination}, that addresses this challenge and yields an $\tO(T^{2 / 3})$ regret mechanism.
The key idea is that when estimating $\+D_{S}(p)$, we decompose it into a product of ratios and estimate these ratios recursively.
Although the initial actions incur constant regret, they converge exponentially fast to the diagonal $\{(p, q) \;|\; 0 \le p = q \le 1\}$ as the recursion depth increases (i.e., as the algorithm proceeds stage by stage), resulting in diminishing regret.
Overall, the total regret is at most $\tO(T^{2 / 3})$.

\vspace{.1in}
\noindent
\textbf{Elimination Algorithms for the Standard MAB Problem.}
For readability, it is helpful to review the standard MAB problem and the elimination algorithms: there are $K$ arms whose (independent) random rewards follow the Bernoulli distributions $\Bern(\rho_{i})$ with \textit{mean rewards} $\rho_{i} \in [0, 1]$, $\forall i \in [K]$.

An elimination algorithm operates in $L + 1 = \log(\frac{1}{\eps}) + 1$ stages, for some parameter $\eps = \eps(T, K) > 0$ to be determined. In stage $\ell \in [0 : L]$, we pull each arm $2^{\ell} / \eps$ times; by standard concentration inequalities, with high probability we can identify and eliminate those $\Omega(\sqrt{\eps / 2^{\ell}})$-suboptimal mean rewards.
After all stages, the rest of the entire time horizon $[T]$ (if existential) can arbitrarily exploit the non-eliminated mean rewards--every such reward is $\+O(\sqrt{\eps / 2^{L}}) = \+O(\eps)$-approximately optimal.

This elimination algorithm has two key features in each stage $\ell \in [0 : L]$, which we call the \textit{accuracy guarantee} \blackref{MABAccuracyGuarantee} and the \textit{regret guarantee} \blackref{MABRegretGuarantee}.
\begin{itemize}
    \item \term[(\textbf{MAB-AG})]{MABAccuracyGuarantee} After stage $\ell$, every non-eliminated mean reward is $\+O(\sqrt{\eps / 2^{\ell}})$-approximately optimal.
    
    \item \term[(\textbf{MAB-RG})]{MABRegretGuarantee} Stage $\ell$ accumulates at most $\+O(\sqrt{\eps / 2^{\ell - 1}}) \cdot 2^{\ell} / \eps \cdot K = \+O(K / \eps)$ regret.\footnote{More rigorously, the initial stage $\ell = 0$ requires a different deduction to the same regret bound $1 \cdot 2^{0} / \eps \cdot K = \+O(K / \eps)$.}
\end{itemize}
In combination, the total regret of this elimination algorithm is at most
\[
    \textstyle
    \sum_{\ell \in [0 : L]} \+O(K / \eps) + T \cdot \+O(\sqrt{\eps / 2^{L}})
    ~=~ \tO(K / \eps + T \cdot \eps).
\]
By choosing $\eps = \tTheta(\sqrt{K / T})$, we achieve (nearly) optimal $\tO(\sqrt{KT})$ regret.

\vspace{.1in}
\noindent
\textbf{Elimination Algorithms for the Bilateral Trade Model.}
In the bilateral trade model, it turns out that (\Cref{lem:GBB-independent:exploration}) by considering the \textit{$\frac{1}{K}$-net} of the entire action space $\{a_{i, j} \defeq (\frac{i}{K}, \frac{j - 1}{K})\}_{i, j \in [K]} \subseteq [0, 1]^{2}$, at least one of the near-diagonal actions $\{a_{k, k}\}_{k \in [K]}$ is a \textit{$\frac{1}{K}$-approximation} to the benchmark action $(p^{*}, q^{*})$;\footnote{As mentioned (\Cref{rmk:benchmark,def:benchmark}), without loss of generality $(p^{*}, q^{*})$ is exactly on the diagonal $p^{*} = q^{*}$.}
we thus designate $\{a_{k, k}\}_{k \in [K]}$ as the \textit{(initial) candidates} and index them by $\CANDIDATE_{0} \defeq [1 : K]$.

If we could mechanically implement the above elimination algorithm for the standard MAB problem, then the total regret would be at most
\begin{align*}
    \tO(K / \eps + \eps T)
    ~=~ \tO(T^{2/3}),
\end{align*}
by choosing $K = \tTheta(\eps^{-1})$ and $\eps = \tTheta(T^{-1 / 3})$, which gives the desired regret bound.

\afterpage{
\begin{figure}
\centering
    \subfloat[\label{fig:intro-fractal-3}The vanilla estimation scheme]
    {\tikzset{every picture/.style={line width = 0.375pt}} %set default line width to 0.75pt

\begin{tikzpicture}[x = 2pt, y = 2pt, scale = 1]
% \fill[BrickRed, fill opacity = 0.2] (10, 5+5) -- (90, 5+5) -- (90, 85+5) -- (75, 70+5) -- (75, 55+5) -- (60, 55+5) -- (45, 40+5) -- (45, 15+5) -- (20, 15+5) -- cycle;
% \fill[SkyBlue, fill opacity = 0.2] (20, 15+5) -- (45, 15+5) -- (45, 40+5) -- cycle;
% \fill[YellowGreen, fill opacity = 0.2] (60, 55+5) -- (75, 55+5) -- (75, 70+5) -- cycle;
% \draw[densely dotted] (20, 15+5) -- (90, 15+5);
% \draw[densely dotted] (45, 40+5) -- (45, 5+5);
% \draw[densely dotted] (60, 55+5) -- (90, 55+5);
% \draw[densely dotted] (75, 70+5) -- (75, 5+5);
% \draw[densely dotted] (30,30) -- (45,30);
% \draw[densely dotted] (45,45) -- (90,45);
% \draw[densely dotted] (30,30) -- (30,20);

% \draw[densely dotted,color=BrickRed] (90,90) -- (90,100);

\draw (0,0) -- (0,100) -- (100,100) -- (100,0) -- cycle;
\draw (50, -3) node [below][inner sep=0.75pt] {seller};
\draw (-3, 50) node [above][inner sep = 0.75pt,rotate=90] {buyer};

\foreach \x in {5, 10, 15, 20, 25, 30, 35, 40, 45, 50, 55, 60, 65, 70, 75, 80, 85, 90, 95}
% \foreach \x in {30}
{\draw[line width=0.75pt,color=red,dotted] (0, \x) -- (\x, \x) -- (\x, 100);}
% \draw[black, fill = red] (30, 30) circle (2pt);

% candidates
\foreach \x in {0, 5, 10, 15, 20, 25, 30, 35, 40, 45, 50, 55, 60, 65, 70, 75, 80, 85, 90, 95, 100}
{\draw[black, fill = white] (\x, \x) circle (2pt);}

\draw[red, line width = 3pt] (0, 30) -- (30, 30) -- (30, 100);
\draw[black, fill = red] (30, 30) circle (2pt);
\node at (30, 30) [anchor = 135] {$a_{k, k}$};

\end{tikzpicture}}
    \hfill
    \subfloat[\label{fig:intro-fractal-1}A Sierpi\'{n}ski triangle]
    {\tikzset{every picture/.style={line width = 0.75pt}} %set default line width to 0.75pt        

\begin{tikzpicture}[x = 2pt, y = 2pt, scale = 1]
\draw (50, -3) node [below][inner sep=0.75pt] {\phantom{seller}};

\foreach \point in
{(0, 0), (6.25, 0), (6.25, 6.25), (0+12.5, 0), (6.25+12.5, 0), (6.25+12.5, 6.25), (0+12.5, 0+12.5), (6.25+12.5, 0+12.5), (6.25+12.5, 6.25+12.5)}
{\draw \point -- +(6.25, 0) -- +(6.25, 6.25) -- cycle;
\draw \point+(25, 0) -- +(31.25, 0) -- +(31.25, 6.25) -- cycle;
\draw \point+(25, 25) -- +(31.25, 25) -- +(31.25, 31.25) -- cycle;
\draw \point+(50, 0) -- +(56.25, 0) -- +(56.25, 6.25) -- cycle;
\draw \point+(75, 0) -- +(81.25, 0) -- +(81.25, 6.25) -- cycle;
\draw \point+(75, 25) -- +(81.25, 25) -- +(81.25, 31.25) -- cycle;
\draw \point+(50, 50) -- +(56.25, 50) -- +(56.25, 56.25) -- cycle;
\draw \point+(75, 50) -- +(81.25, 50) -- +(81.25, 56.25) -- cycle;
\draw \point+(75, 75) -- +(81.25, 75) -- +(81.25, 81.25) -- cycle;}

% level 
% \draw (0,0) -- (100,100) -- (100,0) -- cycle;

% \draw (50,0) -- (50,50) -- (100,50) -- cycle;

% \foreach \point in {(0, 0), (50, 0), (50, 50)} {
%     \draw \point+(25, 0) -- +(25, 25) -- +(50, 25) -- cycle;
% }

% \foreach \point in {(0, 0), (25, 0), (25, 25), (50, 0), (75, 0), (75, 25), (50, 50), (75, 50), (75, 75)} {
%     \draw \point+(12.5, 0) -- +(12.5, 12.5) -- +(25, 12.5) -- cycle;
% }

% \foreach \point in {(0, 0), (12.5, 0), (25, 0), (37.5, 0), (50, 0), (62.5, 0), (75, 0), (87.5, 0), (12.5, 12.5), (37.5, 12.5), (62.5, 12.5), (87.5, 12.5), (25, 25), (37.5, 25), (75, 25), (87.5, 25), (37.5, 37.5), (87.5, 37.5), (50, 50), (62.5, 50), (75, 50), (87.5, 50), (62.5, 62.5), (87.5, 62.5), (75, 75), (87.5, 75), (87.5, 87.5)} {
%     \draw \point+(6.25, 0) -- +(6.25, 6.25) -- +(12.5, 6.25) -- cycle;
% }
\end{tikzpicture}}
    
    \vspace{.2in}
    \subfloat[\label{fig:intro-fractal-2}A single stage of {\FractalElimination}]
    {\tikzset{every picture/.style={line width = 0.375pt}} %set default line width to 0.75pt

\begin{tikzpicture}[x = 2pt, y = 2pt, scale = 1]
\fill[BrickRed, fill opacity = 0.2] (10, 5+5) -- (90, 5+5) -- (90, 85+5) -- (75, 70+5) -- (75, 55+5) -- (60, 55+5) -- (45, 40+5) -- (45, 15+5) -- (20, 15+5) -- cycle;
\fill[SkyBlue, fill opacity = 0.2] (20, 15+5) -- (45, 15+5) -- (45, 40+5) -- cycle;
\fill[SkyBlue, fill opacity = 0.2] (60, 55+5) -- (75, 55+5) -- (75, 70+5) -- cycle;
% \draw[densely dotted] (20, 15+5) -- (90, 15+5);
% \draw[densely dotted] (45, 40+5) -- (45, 5+5);
% \draw[densely dotted] (60, 55+5) -- (90, 55+5);
% \draw[densely dotted] (75, 70+5) -- (75, 5+5);

% 
\draw (0,0) -- (0,100) -- (100,100) -- (100,0) -- cycle;
\draw (50, -3) node [below][inner sep=0.75pt] {seller};
\draw (-3, 50) node [above][inner sep = 0.75pt,rotate=90] {buyer};

% diagonal
% \draw (0, 0) -- (100, 100);

% boxed line
% \foreach \y in {5, 10, 15, 20, 25, 30, 35, 40, 45, 50, 55, 60, 65, 70, 75, 80, 85, 90, 95} {
%     \draw[dotted, draw opacity = 0.5]  (\y, 0) -- (\y, 100);
%     \draw [dotted, draw opacity = 0.5]  (0, \y) -- (100, \y);
% }

% candidates
\foreach \x in {0, 5, 15, 25, 30, 35, 40, 50, 55, 65, 70, 80, 85, 95, 100}
{\draw[black, fill = white] (\x, \x) circle (2pt);}

% \draw[BrickRed, ultra thick, dotted] (10, 5+5) -- (90, 5+5) -- (90, 85+5);
\draw[black, fill = BrickRed] (10, 5+5) circle (2pt);
\draw[black, fill = BrickRed] (90, 85+5) circle (2pt);
% \draw (10, 5) node[anchor = 90] {$a_{\sigma, \sigma}$};
% \draw (90, 85) node[anchor = 180] {$a_{\tau, \tau}$};
% \draw (90, 5) node[anchor = 135] {$a_{\tau, \sigma}$};
% \fill[SkyBlue] (90+1, 5+5+1) -- (90-1, 5+5+1) -- (90-1, 5+5-1) -- cycle;
% \fill[SkyBlue] (90+1, 5+5+1) -- (90-1, 5+5-1) -- (90+1, 5+5-1) -- cycle;
% \draw[color = black] (90+1, 5+5+1) -- (90-1, 5+5+1) -- (90-1, 5+5-1) -- (90+1, 5+5-1) -- cycle;
% \draw[color = black] (90+1, 5+5+1) -- (90-1, 5+5-1);

\draw[SkyBlue, ultra thick, dotted] (20, 15+5) -- (45, 15+5) -- (45, 40+5);
\draw[black, fill = SkyBlue] (20, 15+5) circle (2pt);
\draw[black, fill = SkyBlue] (45, 40+5) circle (2pt);
% \draw (20, 15) node[anchor = 90] {$a_{\sigma', \sigma'}$};
% \draw (45, 40) node[anchor = 180] {$a_{\tau', \tau'}$};
% \draw[color = black, fill = SkyBlue] (45+1, 15+5+1) -- (45-1, 15+5+1) -- (45-1, 15+5-1) -- (45+1, 15+5-1) -- cycle;
% \draw (45, 15) node[anchor = 135] {$a_{\tau', \sigma'}$};
% \draw[color = black, fill = SkyBlue] (45+1, 5+5+1) -- (45-1, 5+5+1) -- (45-1, 5+5-1) -- (45+1, 5+5-1) -- cycle;
% \draw (45, 5) node[anchor = 90] {$a_{\tau', \sigma}$};
% \draw[color = black, fill = SkyBlue] (90+1, 15+5+1) -- (90-1, 15+5+1) -- (90-1, 15+5-1) -- (90+1, 15+5-1) -- cycle;
% \draw (90, 15) node[anchor = 180] {$a_{\tau, \sigma'}$};

\draw[SkyBlue, ultra thick, dotted] (60, 55+5) -- (75, 55+5) -- (75, 70+5);
\draw[black, fill = SkyBlue] (60, 55+5) circle (2pt);
\draw[black, fill = SkyBlue] (75, 70+5) circle (2pt);
% \draw (60, 55) node[anchor = 90] {$a_{\sigma'', \sigma''}$};
% \draw (75, 70) node[anchor = 180] {$a_{\tau'', \tau''}$};
% \draw[color = black, fill = SkyBlue] (75+1, 55+5+1) -- (75-1, 55+5+1) -- (75-1, 55+5-1) -- (75+1, 55+5-1) -- cycle;
% \draw (75, 55) node[anchor = 135] {$a_{\tau'', \sigma''}$};
% \draw[color = black, fill = SkyBlue] (75+1, 5+5+1) -- (75-1, 5+5+1) -- (75-1, 5+5-1) -- (75+1, 5+5-1) -- cycle;
% \draw (75, 5) node[anchor = 90] {$a_{\tau'', \sigma}$};
% \draw[color = black, fill = SkyBlue] (90+1, 55+5+1) -- (90-1, 55+5+1) -- (90-1, 55+5-1) -- (90+1, 55+5-1) -- cycle;
% \draw (90, 55) node[anchor = 180] {$a_{\tau, \sigma''}$};

% \draw (15-0.7071, 15-0.7071) -- (15+0.7071, 15+0.7071);
% \draw (15-0.7071, 15+0.7071) -- (15+0.7071, 15-0.7071);

% \draw (50-0.7071, 50-0.7071) -- (50+0.7071, 50+0.7071);
% \draw (50-0.7071, 50+0.7071) -- (50+0.7071, 50-0.7071);

% \draw (55-0.7071, 55-0.7071) -- (55+0.7071, 55+0.7071);
% \draw (55-0.7071, 55+0.7071) -- (55+0.7071, 55-0.7071);

% \draw (80-0.7071, 80-0.7071) -- (80+0.7071, 80+0.7071);
% \draw (80-0.7071, 80+0.7071) -- (80+0.7071, 80-0.7071);

% \draw (85-0.7071, 85-0.7071) -- (85+0.7071, 85+0.7071);
% \draw (85-0.7071, 85+0.7071) -- (85+0.7071, 85-0.7071);

\foreach \x in {10, 15, 50, 55, 80, 85, 90}
{\draw (\x-1, \x) -- (\x+1, \x);}

% \draw (30-0.7071, 30-0.7071) -- (30+0.7071, 30+0.7071);
% \draw (30-0.7071, 30+0.7071) -- (30+0.7071, 30-0.7071);

% \draw (35-0.7071, 35-0.7071) -- (35+0.7071, 35+0.7071);
% \draw (35-0.7071, 35+0.7071) -- (35+0.7071, 35-0.7071);

\node at (10+1, 10-1) [anchor = -45] {$\sigma$};
\node at (90+1, 90-1) [anchor = -45] {$\tau$};
\node at (20+1, 20-1) [anchor = -45] {$\sigma'$};
\node at (45+1, 45-1) [anchor = -45] {$\tau'$};
\node at (60+1, 60-1) [anchor = -45] {$\sigma''$};
\node at (75+1, 75-1) [anchor = -45] {$\tau''$};

%legends
% \draw[black, fill = white] (110, 75) circle (2pt) node[right] {$\CANDIDATE = \{a_{k, k}\}_{k \in [1 \colon K]}$};
% \draw[black, fill = BrickRed] (110, 65) circle (2pt) node[right] {$\{a_{\sigma, \sigma},\ a_{\tau, \tau}\}$};
% \draw[black, fill = SkyBlue] (110, 55) circle (2pt) node[right] {$\{a_{\sigma', \sigma'},\ a_{\tau', \tau'}\}$};
% \draw[color = black, fill = SkyBlue] (110+1, 45+1) -- (110-1, 45+1) -- (110-1, 45-1) -- (110+1, 45-1) -- cycle;
% \node at (110, 45) [right] {$\{a_{\tau', \sigma'},\ a_{\tau', \sigma},\ a_{\tau, \sigma'}\}$};
% \draw[black, fill = SkyBlue] (110, 35) circle (2pt) node[right] {$\{a_{\sigma'', \sigma''},\ a_{\tau'', \tau''}\}$};
% \draw[color = black, fill = SkyBlue] (110+1, 25+1) -- (110-1, 25+1) -- (110-1, 25-1) -- (110+1, 25-1) -- cycle;
% \node at (110, 25) [right] {$\{a_{\tau'', \sigma''},\ a_{\tau'', \sigma},\ a_{\tau, \sigma''}\}$};
\end{tikzpicture}}
    \hfill
    \subfloat[\label{fig:intro-fractal-4}The novel estimation scheme]
    {\tikzset{every picture/.style={line width = 0.375pt}} %set default line width to 0.75pt

\begin{tikzpicture}[x = 2pt, y = 2pt, scale = 1]
\fill[BrickRed, fill opacity = 0.2] (10, 5+5) -- (90, 5+5) -- (90, 85+5) -- (75, 70+5) -- (75, 55+5) -- (60, 55+5) -- (45, 40+5) -- (45, 15+5) -- (20, 15+5) -- cycle;
\fill[SkyBlue, fill opacity = 0.2] (20, 15+5) -- (45, 15+5) -- (45, 40+5) -- cycle;
\fill[SkyBlue, fill opacity = 0.2] (60, 55+5) -- (75, 55+5) -- (75, 70+5) -- cycle;
% \draw[densely dotted] (20, 15+5) -- (90, 15+5);
% \draw[densely dotted] (45, 40+5) -- (45, 5+5);
% \draw[densely dotted] (60, 55+5) -- (90, 55+5);
% \draw[densely dotted] (75, 70+5) -- (75, 5+5);
\draw[densely dotted] (30,30) -- (45,30);
\draw[densely dotted] (45,45) -- (90,45);
\draw[densely dotted] (30,30) -- (30,20);
\draw[densely dotted] (45,20) -- (45,10);
\draw[densely dotted] (75,60) -- (75,10);
\draw[densely dotted] (90,45) -- (90,10);

% \draw[densely dotted,color=BrickRed] (90,90) -- (90,100);

\draw (0,0) -- (0,100) -- (100,100) -- (100,0) -- cycle;
\draw (50, -3) node [below][inner sep=0.75pt] {seller};
\draw (-3, 50) node [above][inner sep = 0.75pt,rotate=90] {buyer};

% candidates
\foreach \x in {0, 5, 15, 25, 30, 35, 40, 50, 55, 65, 70, 80, 85, 95, 100}
{\draw[black, fill = white] (\x, \x) circle (2pt);}

\draw[ultra thick,line width=3pt,color=Sepia,opacity=0.8] (90,90) -- (90,100);
\draw[ultra thick,line width=3pt,color=BlueViolet,opacity=0.8] (90,45) -- (90,90);

\draw[ultra thick,line width=3pt,color=YellowGreen] (45,45) -- (45,100);
\draw[ultra thick,line width=3pt,color=purple,opacity=0.8] (45,45) -- (45,30);

\draw[SkyBlue, ultra thick,dotted] (20, 15+5) -- (45, 15+5) -- (45, 40+5);
\draw[black, fill = SkyBlue] (20, 15+5) circle (2pt);
\draw[black, fill = SkyBlue] (45, 40+5) circle (2pt);

% \draw[BrickRed, ultra thick,dotted] (10, 5+5) -- (90, 5+5) -- (90, 85+5);
\draw[black, fill = BrickRed] (10, 5+5) circle (2pt);
\draw[black, fill = BrickRed] (90, 85+5) circle (2pt);

\draw[SkyBlue, ultra thick,dotted] (60, 55+5) -- (75, 55+5) -- (75, 70+5);
\draw[black, fill = SkyBlue] (60, 55+5) circle (2pt);
\draw[black, fill = SkyBlue] (75, 70+5) circle (2pt);

% \draw (25-0.7071, 25-0.7071) -- (25+0.7071, 25+0.7071);
% \draw (25-0.7071, 25+0.7071) -- (25+0.7071, 25-0.7071);

% \draw (35-0.7071, 35-0.7071) -- (35+0.7071, 35+0.7071);
% \draw (35-0.7071, 35+0.7071) -- (35+0.7071, 35-0.7071);

% \draw (15-0.7071, 15-0.7071) -- (15+0.7071, 15+0.7071);
% \draw (15-0.7071, 15+0.7071) -- (15+0.7071, 15-0.7071);

% \draw (50-0.7071, 50-0.7071) -- (50+0.7071, 50+0.7071);
% \draw (50-0.7071, 50+0.7071) -- (50+0.7071, 50-0.7071);

% \draw (55-0.7071, 55-0.7071) -- (55+0.7071, 55+0.7071);
% \draw (55-0.7071, 55+0.7071) -- (55+0.7071, 55-0.7071);

% \draw (80-0.7071, 80-0.7071) -- (80+0.7071, 80+0.7071);
% \draw (80-0.7071, 80+0.7071) -- (80+0.7071, 80-0.7071);

% \draw (85-0.7071, 85-0.7071) -- (85+0.7071, 85+0.7071);
% \draw (85-0.7071, 85+0.7071) -- (85+0.7071, 85-0.7071);

\node at (10+1, 10-1) [anchor = -45] {$\sigma$};
\node at (90+1, 90-1) [anchor = -45] {$\tau$};
\node at (20+1, 20-1) [anchor = -45] {$\sigma'$};
\node at (45+1, 45-1) [anchor = -45] {$\tau'$};
\node at (60+1, 60-1) [anchor = -45] {$\sigma''$};
\node at (75+1, 75-1) [anchor = -45] {$\tau''$};

\node at (30, 30) [anchor = -20] {$(p,q)$};

\draw[black, fill = BrickRed] (90, 90) circle (2pt);

% \draw[black, fill = SkyBlue] (45, 45) circle (2pt);

% \draw[black, fill = orange] (45, 10) circle (2pt);
% \draw[black, fill = orange] (75, 10) circle (2pt);
% \draw[black, fill = orange] (90, 10) circle (2pt);
\draw[black, fill = orange] (44, 9) rectangle (46, 11);
\draw[black, fill = orange] (74, 9) rectangle (76, 11);
\draw[black, fill = orange] (89, 9) rectangle (91, 11);

\draw[ultra thick,line width=3pt,color=red] (30,30) -- (30,100);
\draw[black, fill = red] (30, 30) circle (2pt);

% \draw[black, fill = orange] (30, 20) circle (2pt);
% \draw[black, fill = orange] (45, 20) circle (2pt);
\draw[black, fill = SkyBlue] (29, 19) rectangle (31, 21);
\draw[black, fill = SkyBlue] (44, 19) rectangle (46, 21);

\foreach \x in {10, 15, 50, 55, 80, 85, 90}
{\draw (\x-1, \x) -- (\x+1, \x);}

\end{tikzpicture}}
\caption{Diagrams for the {\FractalElimination} subroutine.\\
(\Cref{fig:intro-fractal-3}) Diagram of the vanilla estimation scheme for the {\GainsFromTrade} $\GFT(p, q)$.\\
(\Cref{fig:intro-fractal-1}) The motivating primitive: The \textit{Sierpi\'{n}ski triangle}, a famous example of a \textit{fractal}.\\
(\Cref{fig:intro-fractal-2}) Schematic diagram of a single stage $\ell$ of the {\FractalElimination} subroutine.\\
(\Cref{fig:intro-fractal-4}) Schematic diagram of the novel estimation scheme for the {\GainsFromTrade} $\GFT(p, q)$.}
\label{fig:intro-fractal}
\end{figure}
\par}

However, the standard MAB problem is an ``over-simplified'' model, as the feedback by pulling an arm by itself is an unbiased estimator for its mean reward. In the bilateral trade model, instead, it is challenging to develop such an estimator using one-bit feedback, especially under the counterpart \textit{accuracy guarantee} \blackref{AccuracyGuarantee} and \textit{regret guarantee} \blackref{RegretGuarantee} requirements in each stage $\ell \in [0 : L]$.
\begin{itemize}
    \item \term[(\textbf{AG})]{AccuracyGuarantee}
    After stage $\ell$, every non-eliminated candidate $a_{k, k}$ is $\tO(\sqrt{\eps / 2^{\ell}})$-approximately optimal.
    
    \item \term[(\textbf{RG})]{RegretGuarantee}
    Stage $\ell$ accumulates at most $\tO(K / \eps)$ regret.
\end{itemize}
The essence of our {\FractalElimination} subroutine is exactly a new estimation scheme that succeeds in fulfilling both requirements---we achieve this by leveraging \textit{Bayes' theorem} and the \textit{independence} (between both agents' values $\Val{} \sim \+D_{S} \bigotimes \+D_{B}$) in a novel way.

Clearly, based on the analysis above, the satisfaction of \blackref{AccuracyGuarantee} and \blackref{RegretGuarantee} together with the choices of parameters $K = \tTheta(\eps^{-1})$ and $\eps = \tTheta(T^{-1 / 3})$ immediately gives the desired regret bound $\tO(T^{2/3})$.

\vspace{.1in}
\noindent
\textbf{Notation.}
For readability, below we assume $a_{i, j} = (\frac{i}{K}, \frac{j}{K})$, $\forall i, j \in [0 : K]$ and $\CANDIDATE_{0} = [0 : K]$---this shifts the candidates $\{a_{k, k}\}_{k \in [0 : K]}$ exactly to the diagonal and can hide many minor technical details.

We introduce some requisite notation. An action $(p, q) \in [0, 1]^{2}$, in expectation over the randomness of values $(S, B) \sim \+D_{S} \bigotimes \+D_{B}$, yields the {\GainsFromTrade} $\GFT(p, q)$.
\begin{align*}
    \GFT(p, q)
    & ~\defeq~ {\bb E}_{(S, B) \sim \+D_{S} \bigotimes \+D_{B}}[\GFT(S, B, p, q)],
    && \forall (p, q) \in [0, 1]^{2}.
\end{align*}
% As mentioned (\Cref{rmk:benchmark,def:benchmark}), the benchmark action $(p^{*}, q^{*})$ can be any one of the diagonal actions that maximize this expression.
The independence $(S, B) \sim \+D_{S} \bigotimes \+D_{B}$ enables (\Cref{lem:GFT-independent}) useful decomposition of $\GFT(p, q)$; specifically, in case of a candidate $(p, q) = a_{k, k}$, we have
\begin{align}
    \GFT(a_{k, k})
    &\textstyle ~=~ \Horizonal{[0, p], q} + \Vertical{p, [q, 1]}.
    \tag{\textbf{D1}}\label{eqn:gft-H-V}
\end{align}
Here the ``horizonal'' function $\Horizonal{[p', p], q} \defeq \int_{p'}^{p} \+D_{S}(x) \d x \cdot (1 - \+D_{B}(q))$, $\forall (p, q) \in [0, 1]^{2}$, $\forall p' \in [0, p]$,
and the ``vertical'' function $\Vertical{p, [q, q']} \defeq \+D_{S}(p) \cdot \int_{q}^{q'} (1 - \+D_{B}(y)) \d y$, $\forall (p, q) \in [0, 1]^{2}$, $\forall q' \in [q, 1]$.

Provided one-bit feedback, while regret-optimal estimation of $\GFT(a_{k, k})$ is a challenging task, observations \blackref{Estimation1} and \blackref{Estimation2} are easy to verify (\Cref{lem:GBB-independent:estimate}) and serve as the bases of our further investigation.
\begin{itemize}
    \item \term[(\textbf{O1})]{Estimation1}
    Given $\delta > 0$ and $(p, q) \in [0, 1]^{2}$, querying $(p, q)$ a number of $\tTheta(\delta^{-2})$ times suffices to estimate $\+D_{S}(p) \cdot (1 - \+D_{B}(q))$ within precision $\delta > 0$.
    
    \item \term[(\textbf{O2})]{Estimation2}
    Given $\delta > 0$ and $(p, q) \in [0, 1]^{2}$, querying actions lying on the ``vertical segment $(p, q)$--$(p, 1)$'' a total number of $\tTheta(\delta^{-2})$ times suffices to \textit{simultaneously} estimate $\Vertical{p, [q, q']}$, $\forall q' \in [q, 1]$ within precision $\delta > 0$. Similarly for $\Horizonal{[p', p], q}$, $\forall p' \in [0, p]$.
\end{itemize}

\noindent
\textbf{Fractal Elimination.}
Now we elaborate the {\FractalElimination} subroutine (\Cref{alg:GBB-independent:fractal}), providing \Cref{fig:intro-fractal} for diagrams. This is essentially an elimination algorithm over $L + 1 = \log(\frac{1}{\eps}) + 1$ stages.
In each stage $\ell \in [0 : L]$, to fulfill \blackref{AccuracyGuarantee}, we need to estimate every non-eliminated candidate $a_{k, k}$, $\forall k \in \CANDIDATE_{\ell}$ (initially $\CANDIDATE_{0} = [0 : K]$).

For comparison purposes, it is instructive to first consider (\Cref{rmk:vanilla-estimation}) a vanilla estimation scheme.

\begin{remark}[A Vanilla Estimation Scheme]
\label{rmk:vanilla-estimation}
To estimate a candidate $(p, q) = a_{k, k}$, a straightforward attempt is to estimate $\Horizonal{[0, p], q}$ and $\Vertical{p, [q, 1]}$ directly using \blackref{Estimation2}; see \Cref{fig:intro-fractal-3} for a diagram.

To fulfill \blackref{AccuracyGuarantee}, we need to make $\tTheta(2^{\ell} / \eps) + \tTheta(2^{\ell} / \eps) = \tTheta(2^{\ell} / \eps)$ many queries on actions lying on the polyline $(0, q)$--$(p, q)$--$(p, 1)$.
But for different candidates $a_{k, k}$, $\forall k \in \CANDIDATE_{\ell}$, such polylines $(0, q)$--$(p, q)$--$(p, 1)$ have little overlaps, which leads to $\tTheta(|\CANDIDATE_{\ell}| \cdot 2^{\ell} / \eps)$ many queries in total.
However, it is easy to construct instances such that $|\CANDIDATE_{\ell}| = \Omega(K)$ and each query incurs constant regret, which yields $\tOmega(2^{\ell} \cdot K / \eps)$ total regret and thus violates \blackref{RegretGuarantee}.
\end{remark}

The main disadvantage of this vanilla estimation scheme is the \textit{lack of information reuse} across candidates. In stark contrast, we develop a novel estimation scheme that enables \textit{information reuse} (\Cref{rmk:information-reuse}); see \Cref{fig:intro-fractal-2,fig:intro-fractal-4} for diagrams.
Without loss of generality---in a sense of mathematical induction---we consider on a specific stage $\ell \in [0 : L]$.\footnote{More rigorously, we consider a specific \textit{non-initial} stage $\ell \in [1 : L]$; the \textit{initial} stage $\ell = 0$ needs slightly different treatment.}

To better reflect the core algorithmic idea, we consider a simplified case under the following assumptions, which are also demonstrated in \Cref{fig:intro-fractal-2,fig:intro-fractal-4}.
\begin{itemize}
    \item Before the preceding stage $(\ell - 1)$, the candidates $\CANDIDATE_{\ell - 1}$ are those lying on a diagonal segment $(\sigma, \sigma)$--$(\tau, \tau)$; these two parameters satisfy $\tau - \sigma \le 2^{-(\ell - 1)}$.\footnote{\label{footnote:general-case}In general, in each stage $\ell \in [0 : L]$ we require up to $2^{\ell}$ many diagonal segments of the format $(\sigma, \sigma)$--$(\tau, \tau)$ to enclose the candidates $\CANDIDATE_{\ell}$; every such diagonal segment satisfies $\tau - \sigma \le 2^{-\ell}$. This can be guaranteed by induction:\\
    In the \textit{initial} stage $\ell = 0$, we can enclose all candidates $\CANDIDATE_{0} = [0 : K]$ using the entire diagonal $(0, 0)$--$(1, 1)$.\\
    In a specific \textit{non-initial} stage $\ell \in [1 : L]$,  (\Cref{fig:intro-fractal-2}) regarding a specific diagonal segment $(\sigma, \sigma)$--$(\tau, \tau)$ with $\tau - \sigma \le 2^{-\ell}$, we can eliminate the ``bad'' candidates, halve the non-eliminated ``good'' candidates, and then enclose them using two shorter diagonal segments $(\sigma',\sigma')$--$(\tau',\tau')$ and $(\sigma'',\sigma'')$--$(\tau'',\tau'')$ with $\tau' - \sigma',\ \tau'' - \sigma'' \le \frac{\tau - \sigma}{2} \le 2^{-\ell}$.}
    
    \item Stage $(\ell - 1)$ eliminates the ``type-$\ominus$'' candidates, thus shrinking $\CANDIDATE_{\ell - 1}$ to $\CANDIDATE_{\ell}$.
    
    \item So before the current stage $\ell$, the candidates $\CANDIDATE_{\ell}$ are those lying on two diagonal segments $(\sigma', \sigma')$--$(\tau', \tau')$ and $(\sigma'', \sigma'')$--$(\tau'', \tau'')$; these four parameters satisfy $\tau' - \sigma',\ \tau'' - \sigma'' \le \frac{\tau-\sigma}{2} \le 2^{-\ell}$.\textsuperscript{\ref{footnote:general-case}}
\end{itemize}

Without loss of generality, we consider a specific candidate $(p, q) = a_{k, k}$, $\forall k \in \CANDIDATE_{\ell}$.
In light of Decomposition~\blackeqref{eqn:gft-H-V} that $\GFT(a_{k, k}) = \Horizonal{[0, p], q} + \Vertical{p, [q, 1]}$, we only elaborate how to estimate $\Vertical{p, [q, 1]}$;
everything extends to $\Horizonal{[0, p], q}$ symmetrically.
Instead of the vanilla estimation scheme (\Cref{rmk:vanilla-estimation}), we estimate $\Vertical{p, [q, 1]}$ based on the following Decomposition~\blackeqref{eqn:decomp}; see \Cref{fig:intro-fractal-4} for a diagram.\footnote{\label{footnote:decomp}As we merely consider two stages $(\ell - 1)$ and $\ell$ here for simplification, Decomposition~\blackeqref{eqn:decomp} has a \textit{two-level structure}. However, as we can actually consider all stages thus far $0, 1, \dots, \ell$, the actual decomposition has an \textit{$(\ell + 1)$-level recursive structure}.}
\begin{align}
    &\textstyle \textcolor{red}{\Vertical{p, [q, 1]}}
    \notag\\
    \mr{Bayes' theorem \& independence}
    &\textstyle ~=~ \Vertical{\tau', [q, 1]} \cdot \textcolor{SkyBlue}{\underbrace{\frac{\+D_{S}(p) \cdot (1 - \+D_{B}(q))}{\+D_{S}(\tau') \cdot (1 - \+D_{B}(q))}}_{(\dagger)}}
    \notag\\
    \mr{additivity of integrals}
    &\textstyle ~=~ \Big(\textcolor{purple}{\Vertical{\tau', [q, \tau']}} + \textcolor{YellowGreen}{\Vertical{\tau', [\tau', 1]}}\Big) \cdot \textcolor{SkyBlue}{(\dagger)}
    \notag\\
    \mr{Bayes' theorem \& independence}
    &\textstyle ~=~ \Big(\textcolor{purple}{\Vertical{\tau', [q, \tau']}} + \Vertical{\tau, [\tau', 1]} \cdot \textcolor{orange}{\underbrace{\frac{\+D_{S}(\tau') \cdot (1 - \+D_{B}(\tau'))}{\+D_{S}(\tau) \cdot (1 - \+D_{B}(\tau'))}}_{(\ddagger)}}\Big) \cdot \textcolor{SkyBlue}{(\dagger)}
    \notag\\
    \mr{additivity of integrals}
    &\textstyle ~=~ \Big(\textcolor{purple}{\Vertical{\tau', [q, \tau']}} + \big(\textcolor{BlueViolet}{\Vertical{\tau, [\tau', \tau]}} + \textcolor{Sepia}{\Vertical{\tau, [\tau, 1]}}\big) \cdot \textcolor{orange}{(\ddagger)}\Big) \cdot \textcolor{SkyBlue}{(\dagger)}.
    \tag{\textbf{D2}}\label{eqn:decomp}
\end{align}
To proceed, we make two induction hypotheses \blackref{InductionHypothesis1} and \blackref{InductionHypothesis2} for the current stage $\ell$:
\begin{itemize}
    \item \term[(\textbf{IH1})]{InductionHypothesis1}
    All previous stages $0, 1, \dots, (\ell - 1)$ fulfill \blackref{AccuracyGuarantee}.
    
    \item \term[(\textbf{IH2})]{InductionHypothesis2}
    $\textcolor{Sepia}{\Vertical{\tau, [\tau, 1]}}$ has already been estimated within precision $\tTheta(\eps)$.
\end{itemize}
Later, we will verify the counterparts of \blackref{InductionHypothesis1} and \blackref{InductionHypothesis2} for the next stage $(\ell + 1)$.

\vspace{.1in}
\noindent
\textbf{The Queries in the Current Stage $\ell$.}
The current stage $\ell$ makes queries of two types \colorref{SkyBlue}{Query1} and \colorref{orange}{Query2}:
\begin{itemize}
    \item \term[(\textbf{Q1})]{Query1}
    For each of \textcolor{SkyBlue}{the $|\CANDIDATE_{\ell}|$ many discretized actions $a_{i, j}$ on the two vertical segments $(\tau', \sigma')$--$(\tau', \tau')$ and $(\tau'', \sigma'')$--$(\tau'', \tau'')$}, query $\tTheta(2^{\ell} / \eps)$ times.\footnote{In general (\Cref{footnote:general-case}), in each stage $\ell \in [0 : L]$ we need to address up to $2^{\ell}$ vertical segments of the format $(\tau, \sigma)$--$(\tau, \tau)$. Nonetheless, we always address exactly $|\CANDIDATE_{\ell}|$ many discretized actions $a_{i, j}$---querying each of them $\tTheta(2^{\ell} / \eps)$ times.}\\
    \Comment{The preceding stage $(\ell - 1)$ operates on the single vertical segment $(\tau, \sigma)$--$(\tau, \tau)$ in a similar way, but the regret incurred should be counted to stage $(\ell - 1)$ rather than stage $\ell$.}
    
    \item \term[(\textbf{Q2})]{Query2}
    For each of \textcolor{orange}{three specific actions $(\tau', \sigma)$, $(\tau'', \sigma)$, and $(\tau, \sigma)$}, query $\tTheta(\eps^{-2})$ times.\footnote{\label{footnote:Query2}In general (\Cref{footnote:general-case}), in each stage $\ell \in [0 : L]$ we need to address up to $\+O(2^{\ell})$ many such special actions---querying each of them $\tTheta(\eps^{-2})$ times.}
\end{itemize}

\noindent
\textbf{Fulfilling the Accuracy Guarantee \blackref{AccuracyGuarantee}.}
We reason about terms in Decomposition~\blackeqref{eqn:decomp} one by one.
\begin{itemize}
    \item $\textcolor{purple}{\Vertical{\tau', [q, \tau']}}$ can be estimated within precision $\tTheta(\eps)$.
    \hfill
    \blackref{Estimation2} and \colorref{SkyBlue}{Query1}
    
    \item $\textcolor{BlueViolet}{\Vertical{\tau, [\tau', \tau]}}$ has already been estimated within precision $\tTheta(\eps)$ in stage $(\ell - 1)$.
    \hfill
    \blackref{Estimation2} and \colorref{SkyBlue}{Query1}
    
    \item $\textcolor{Sepia}{\Vertical{\tau, [\tau, 1]}}$ has already been estimated within precision $\tTheta(\eps)$ provided \blackref{InductionHypothesis2}.
    
    \item $\textcolor{orange}{(\ddagger)}$ can be estimated within precision $\tTheta(\eps)$.
    \hfill
    \blackref{Estimation1} and \colorref{orange}{Query2}
    
    \item $\textcolor{SkyBlue}{(\dagger)}$ can be estimated within precision $\tTheta(\sqrt{\eps / 2^{\ell}})$.
    \hfill
    \blackref{Estimation1} and \colorref{SkyBlue}{Query1}
\end{itemize}
Altogether and since $\ell \le L = \log(\frac{1}{\eps})$, we can estimate $\textcolor{red}{\Vertical{p, [q, 1]}}$ within precision $\textcolor{red}{\tTheta(\sqrt{\eps / 2^{\ell}})}$.\footnote{\label{footnote:recursion-error}While this claim is correct, we omit many technical details here. For example (\Cref{footnote:decomp}), as the actual decomposition has an \textit{$(\ell + 1)$-level recursive structure}, the estimation error blows up by a factor of $\+O(\ell) = \+O(L) = \+O\big(\log(\frac{1}{\eps})\big) = \tO(1)$.}
\begin{align*}
    \textcolor{purple}{\tTheta(\eps)}
    + \textcolor{BlueViolet}{\tTheta(\eps)}
    + \textcolor{Sepia}{\tTheta(\eps)}
    + \textcolor{orange}{\tTheta(\eps)}
    + \textcolor{SkyBlue}{\tTheta(\sqrt{\eps / 2^{\ell}})}
    ~=~ \textcolor{red}{\tTheta(\sqrt{\eps / 2^{\ell}})}.
\end{align*}
After treating $\Horizonal{[0, p], q}$ symmetrically, we succeed in estimating $(p, q) = a_{k, k}$, $\forall k \in \CANDIDATE_{\ell}$ within precision $\tTheta(\sqrt{\eps / 2^{\ell}})$.
Eliminating those $\tOmega(\sqrt{\eps / 2^{\ell}})$-suboptimal candidates immediately fulfills \blackref{AccuracyGuarantee}.

\begin{remark}[Information Reuse]
\label{rmk:information-reuse}
In our estimation scheme, different candidates share both \textcolor{SkyBlue}{type-\colorref{SkyBlue}{Query1} query actions} and \textcolor{orange}{type-\colorref{orange}{Query2} query actions} and can be estimated in a batch. This information reuse (across candidates) is exactly its strength relative to the vanilla estimation scheme (\Cref{rmk:vanilla-estimation,fig:intro-fractal-3}).
\end{remark}

\vspace{.05in}
\noindent
\textbf{Fulfilling the Regret Guarantee \blackref{RegretGuarantee}.}
We reason about the two types of queries separately.
Firstly, a \textcolor{SkyBlue}{type-\colorref{SkyBlue}{Query1} query action} has a distance of at most $2^{-\ell}$ to an $\tO(\sqrt{\eps / 2^{\ell - 1}})$-approximately optimal candidate $(\sigma', \sigma')$, $(\tau', \tau')$, $(\sigma'', \sigma'')$, or $(\tau'', \tau'')$, given \blackref{InductionHypothesis1} and \blackref{AccuracyGuarantee}.
Essentially due to the triangle inequality, a single \textcolor{SkyBlue}{type-\colorref{SkyBlue}{Query1} query} incurs regret of at most $2^{-\ell} + \tO(\sqrt{\eps / 2^{\ell - 1}})$.
Since there are $|\CANDIDATE_{\ell}|$ many such actions and each is queried $\tTheta(2^{\ell} / \eps)$ many times, \textcolor{SkyBlue}{type-\colorref{SkyBlue}{Query1} queries} in stage $\ell$ together incur regret of at most
\begin{align*}
    \big(2^{-\ell} + \tO(\sqrt{\eps / 2^{\ell - 1}})\big) \cdot |\CANDIDATE_{\ell}| \cdot \tTheta(2^{\ell} / \eps)
    ~=~ \tO(K / \eps),
\end{align*}
given that $\ell \le L = \log(\frac{1}{\eps})$ and $|\CANDIDATE_{\ell}| \le |[0 : K]| = K + 1$.

Secondly, for the same reasons, a single \textcolor{orange}{type-\colorref{orange}{Query2} query} incurs regret of at most $2^{-(\ell - 1)} + \tO(\sqrt{\eps / 2^{\ell - 2}})$.
Since there are $\+O(2^{\ell - 1})$ many such special actions (\Cref{footnote:Query2}) and each is queried $\tO(\eps^{-2})$ many times, given that $K = \tTheta(\eps^{-1})$, \textcolor{orange}{type-\colorref{orange}{Query2} queries} in stage $\ell$ together incur regret of at most
\begin{align*}
    \big(2^{-(\ell - 1)} + \tO(\sqrt{\eps / 2^{\ell - 2}})\big) \cdot \+O(2^{\ell - 1}) \cdot \tTheta(\eps^{-2})
    ~=~ \tO(K / \eps).
\end{align*}

Altogether, the estimation of $\Vertical{p, [q, 1]}$ incurs regret of at most $\tO(K / \eps)$; so it is for $\Horizonal{[0, p], q}$ by symmetry. This fulfills \blackref{RegretGuarantee}.

\vspace{.1in}
\noindent
\textbf{Fulfilling the Induction Hypotheses \blackref{InductionHypothesis1} and \blackref{InductionHypothesis2}.}
The satisfaction of \blackref{AccuracyGuarantee} for the current stage $\ell$ immediately fulfills \blackref{InductionHypothesis1} for the next stage $(\ell + 1)$.
It is also easy to show \blackref{InductionHypothesis2} for the next stage $(\ell + 1)$; given that $\textcolor{YellowGreen}{\Vertical{\tau', [\tau', 1]}} = (\textcolor{BlueViolet}{\Vertical{\tau, [\tau', \tau]}} + \textcolor{Sepia}{\Vertical{\tau, [\tau, 1]}}) \cdot \textcolor{orange}{(\ddagger)}$, we are already able to estimate $\textcolor{YellowGreen}{\Vertical{\tau', [\tau', 1]}}$ within precision $\textcolor{BlueViolet}{\tTheta(\eps)} + \textcolor{Sepia}{\tTheta(\eps)} + \textcolor{orange}{\tTheta(\eps)} = \textcolor{YellowGreen}{\tTheta(\eps)}$;\textsuperscript{\ref{footnote:recursion-error}} similarly for $\Vertical{\tau'', [\tau'', 1]}$.

\vspace{.1in}
\noindent
\textbf{Recursion.}
In the above discussions, we concentrate on two stages $(\ell - 1)$ and $\ell$ for simplification. Rather, the actual implementation of {\FractalElimination} requires a divide-and-conquer recursion over all stages $\ell = 0, 1, \dots, L$. Overall, the queried actions (\Cref{fig:intro-fractal-3,fig:intro-fractal-4}) in spirit constitute the edges of a Sierpi\'{n}ski triangle (\Cref{fig:intro-fractal-2}), which accounts for the name \textit{fractal elimination}.

\subsection*{Ingredient~2: A New Lower-Bound Construction, for Correlated Values}

Our second main contribution is an $\Omega(T^{3 / 4})$ lower bound for {\GBB} two-bit-feedback mechanisms in the ``correlated values'' setting, which improves the $\Omega(T^{5 / 7})$ lower bound obtained in \cite[Theorem~5.5]{BCCF24} for the same context and matches the $\tO(T^{3 / 4})$ upper bound obtained in \cite[Theorem~5.4]{BCCF24} even for more general contexts.

At a high level, the technical challenge is due in large part to the {\GBB} constraint---it introduces relevance among different rounds---and the crux of our proof is a new remedy for it.
For readability, let us first omit the {\GBB} constraint and sketch a general lower-bound approach.
Then, we will carefully compare the previous {\GBB} remedy \cite{BCCF24} and our new {\GBB} remedy.

\vspace{.1in}
\noindent
\textbf{A General Lower-Bound Approach.}
Basically, a regret lower bound requires constructing a family of \textit{hard-to-distinguish} instances:
When facing some instance from this family, a mechanism must determine its identity in the online learning process, namely ``finding a needle in a haystack''.

% such that each mechanism, when faced with an input from the family, must resolve the task of \textit{finding a needle in the haystack}---identifying the exact instance.

To address our problem using this approach, we shall construct one \textit{base instance} $\+D_{0}$ and $K \ge 1$ \textit{hard instances} $\{\+D_{k}\}_{k \in [K]}$---recall that an instance is a $[0, 1]^{2}$-supported joint distribution.
Each hard instance $\+D_{k}$ shall differ from the base instance $\+D_{0}$ by some $\delta > 0$ in the total variation distance.
As such, $\+D_{k}$ can simply perturb $\+D_{0}$ by total probability mass of $\Theta(\delta)$, distributed across constant number of actions, which forms the ``needle''. The construction shall follow two criteria:
\begin{itemize}
    \item Information-Regret Dilemma:
    Each hard instance $\+D_{k}$ has a set of \textit{informative actions}.
    Only those actions can provide information (on playing) that helps distinguish this hard instance $\+D_{k}$, but each of them will incur \textit{constant regret} $\Omega(1)$.
    
    \item Disjointness:
    All hard instances $\{\+D_{k}\}_{k \in [K]}$ shall have \textit{disjoint} sets of informative actions, thus no information sharing on individual plays of informative actions of different $\+D_{k}$'s.
\end{itemize}
Given such a construction (if possible), we can informally reason about the total regret of a mechanism as follows:
% one usually has the following argument for the regret lower bound: 
\begin{itemize}
    \item If all individual hard instances $\{\+D_{k}\}_{k \in [K]}$ are distinguishable from the base instance $\+D_{0}$, given the total variation distances of $\delta$, this necessitates $\Omega(\delta^{-2})$ number of plays of a single $\+D_{k}$'s informative actions and thus $\Omega(K \delta^{-2})$ number of such plays altogether (Disjointness).
    However, then, we will suffer from $\Omega(K \delta^{-2})$ total regret (Information-Regret Dilemma).
    
    \item Otherwise, some hard instances $\+D_{k}$ are indistinguishable, and (roughly speaking) we will suffer from $\Omega(\delta T)$ total regret when facing one of them.
\end{itemize}
By choosing $\delta = K^{1 / 3} T^{-1 / 3}$, any mechanism will incur total regret
\begin{align}
    \min\set{\Omega(K \delta^{-2}),\ \Omega(\delta T)}
    ~=~ \Omega(K^{1 / 3} T^{2 / 3}).
    \label{eqn:intro:bound-for-LB}
\end{align}
Consequently, proving an optimal lower bound reduces to the task of seeking a construction that has the largest possible $K \ge 1$ and, simultaneously, retains the above criteria.

It remains to determine the bottleneck of $K \ge 1$.
Since each hard instance $\+D_{k}$ for $k \in [K]$ needs to perturb the base instance $\+D_{0}$ by total probability mass of $\Theta(\delta)$ (planting a ``mass-$\delta$ needle''), and these perturbations are ``disjoint'', we naturally require $K = \+O(\delta^{-1})$.
% \textsuperscript{\ref{footnote:intro:LB}}
So if a construction can satisfy $K = \Theta(\delta^{-1})$, plugging this into \Cref{eqn:intro:bound-for-LB} directly gives an $\Omega(T^{3 / 4})$ lower bound, as desired.

To implement the above general approach in practice, however, the technical challenge is due in large part to the {\GBB} constraint. As we quote from \cite{BCCF24}:

\vspace{-.075in}
\begin{quote}
    \textit{``This (the {\GBB} constraint) considerably complicates the construction of the hard instances, as any algorithm could sacrifice some profit temporarily by posting prices with $\SPrice{t} > \BPrice{t}$ to extract a large {\GainsFromTrade}.''}
\end{quote}

\vspace{-.075in}
\noindent
The previous work \cite{BCCF24} and our work will adopt very different remedies for the {\GBB} constraint.

\vspace{.1in}
\noindent
\textbf{The Previous {\GBB} Remedy.}
The previous work \cite{BCCF24} circumvents the {\GBB} constraint in an ingenious manner. Their base instance $\+D_{0}$ (see \Cref{fig:intro-LB-1} for a diagram) involves, \textit{just below the diagonal}, value points ``so bad'' that even a single action play on their \textit{lower right side} will incur intolerable regret.\\
\Comment{These value points refer to grey points $\+V^{4}$ in \Cref{fig:intro-LB-1}. Further, ``so bad'' means the {\GFT}-decrease due to a single such value point even dominates the total {\GFT}-increase due to all other value points vertically above it.}\\
This automatically forces a regret-optimal mechanism to satisfy the {\GBB} constraint.
In other words, central to their lower-bound construction is such a more ``qualitative'' principle:

\vspace{-.075in}
\begin{quote}
    \textit{Sacrifice of profit \textbf{cannot} produce extra {\GainsFromTrade}.}
\end{quote}

\vspace{-.075in}
\noindent
As it turns out, there are as many ``so bad'' value points as the hard instances ($K \ge 1$), and they each have probability mass $\Omega(K \delta)$.
Then, it can be shown that $K = \Theta(\delta^{-1 / 2})$. Plugging this in \Cref{eqn:intro:bound-for-LB} gives an $\Omega(T^{5 / 7})$ lower bound \cite[Theorem~5.5]{BCCF24}. 

% \Comment{Every hard instance $\+D_{k}$ perturbs the $k$-th and $(k+1)$-th purple/blue points $\+V^{1}$ and $\+V^{2}$ from left.}

\begin{figure}[t]
\centering
    \subfloat[\label{fig:intro-LB-1}The previous construction]
    {\tikzset{every picture/.style={line width = 0.375pt}} %set default line width to 0.75pt        

\begin{tikzpicture}[x = 2pt, y = 2pt, scale = 0.8]
\draw (0,0) -- (0,100) -- (100,100) -- (100,0) -- cycle;
\draw (50, -3) node [below][inner sep=0.75pt] {seller};
\draw (-3, 50) node [above][inner sep = 0.75pt,rotate=90] {buyer};

% diagonal
\draw (0, 0) -- (100, 100);

% good action
\fill[green, fill opacity = 0.2] (40, 40) -- (60, 60) -- (40, 60) -- cycle;
\draw (40, 40) -- (60, 60) -- (40, 60) -- cycle;

% bad action
\fill[red, fill opacity = 0.2] (40, 80) -- (60, 80) -- (60, 90) -- (40, 90) -- cycle;
\draw (40, 80) -- (60, 80) -- (60, 90) -- (40, 90) -- cycle;

% value'
\foreach \x in {40, 45, 50, 55, 60} {
    \draw[black, fill = violet] (\x, 90) circle (2pt);
    \draw[black, fill = blue] (\x, 80) circle (2pt);
    \draw[black, fill = green] (0, \x - 3) circle (2pt);
    \draw[black, fill = gray] (\x, \x - 3) circle (2pt);
}
\draw[black, fill = YellowOrange] (0, 60) circle (2pt);

% value''
% \foreach \x in {0, 5, 10, 15, 20, 25, 30, 35, 40, 45, 50, 55, 60} {
%     \draw[black, fill = black] (\x, \x + 40) circle (2pt);
% }

% value corner
\draw[black, fill = BrickRed] (0, 0) circle (2pt);
\draw[black, fill = BrickRed] (0, 100) circle (2pt);
\draw[black, fill = BrickRed] (100, 0) circle (2pt);
\draw[black, fill = BrickRed] (100, 100) circle (2pt);

%legends
\draw[black, fill = violet] (110, 75) circle (2pt) node[right] {$\VAL_{1}$};
\draw[black, fill = blue] (110, 65) circle (2pt) node[right] {$\VAL_{2}$};
\draw[black, fill = green] (110, 55) circle (2pt) node[right] {$\VAL_{3}$};
\draw[black, fill = gray] (110, 45) circle (2pt) node[right] {$\VAL_{4}$};
\draw[black, fill = YellowOrange] (110, 35) circle (2pt) node[right] {$\val_{5}$};
\draw[black, fill = BrickRed] (110, 25) circle (2pt) node[right] {$\VALcorner$};
\end{tikzpicture}}
    \hfill
    \subfloat[\label{fig:intro-LB-2}Our new construction]
    {\tikzset{every picture/.style={line width = 0.375pt}} %set default line width to 0.75pt        

\begin{tikzpicture}[x = 2pt, y = 2pt, scale = 0.8]
\draw (0,0) -- (0,100) -- (100,100) -- (100,0) -- cycle;
\draw (50, -3) node [below][inner sep=0.75pt] {seller};
\draw (-3, 50) node [above][inner sep=0.75pt, rotate=90] {buyer};

% diagonal
\draw (0, 0) -- (100, 100);

% good action
\fill[green, fill opacity = 0.2] (40, 40) -- (60, 40) -- (60, 60) -- (40, 60) -- cycle;
\draw (40, 40) -- (60, 40) -- (60, 60) -- (40, 60) -- cycle;

% bad action
\fill[red, fill opacity = 0.2] (0, 40) -- (40, 40) -- (40, 100) -- (0, 100) -- cycle;
\draw (0, 40) -- (40, 40) -- (40, 100) -- (0, 100) -- cycle;

% value'
\foreach \x in {20, 25, 30, 35, 45, 50, 55, 60} {
    \draw[black, fill = white] (\x, \x + 20) circle (2pt);
}
\draw[black, fill = white] (40 + 0.7071, 60 + 0.7071) arc (45:225:2pt);

% value''
\foreach \x in {0, 5, 10, 15, 20, 25, 30, 35, 40, 45, 50, 55, 60} {
    \draw[black, fill = black] (\x, \x + 40) circle (2pt);
}

% value majority
\draw[black, fill = SkyBlue] (40 - 0.7071, 60 - 0.7071) arc (-135:45:2pt);
\draw[black] (40 - 0.7071, 60 - 0.7071) -- (40 + 0.7071, 60 + 0.7071);

% value corner
\draw[black, fill = BrickRed] (0, 0) circle (2pt);
\draw[black, fill = BrickRed] (0, 100) circle (2pt);
\draw[black, fill = BrickRed] (100, 0) circle (2pt);
\draw[black, fill = BrickRed] (100, 100) circle (2pt);

%legends
\draw[black, fill = white] (110, 65) circle (2pt) node[right] {$\VALlowerright$};
\draw[black, fill = black] (110, 55) circle (2pt) node[right] {$\VALupperleft$};
\draw[black, fill = BrickRed] (110, 45) circle (2pt) node[right] {$\VALcorner$};
\draw[black, fill = SkyBlue] (110, 35) circle (2pt) node[right] {$\valmajority$};
\end{tikzpicture}}
\caption{Diagrams of the previous $\Omega(T^{5 / 7})$ lower-bound construction \cite[Theorem~5.5]{BCCF24} and our new $\Omega(T^{3 / 4})$ lower-bound construction (\Cref{thm:GBB-correlated:LB}).}
\label{fig:intro-LB}
\end{figure}

% \bigskip
\vspace{.1in}
\noindent
\textbf{Our New {\GBB} Remedy.}
Here, our contribution is a more careful treatment of the {\GBB} constraint instead of simply circumventing it like \cite{BCCF24}. Specifically, instead of incorporating ``so bad'' value points into the base instance $\+D_{0}$, our lower-bound construction derives from a more ``quantitative'' principle:

\vspace{-.075in}
\begin{quote}
    \textit{Sacrifice of profit (``investment'') \textbf{can} produce extra {\GainsFromTrade} (``return'').\\
    Just, the ``return on investment'' is \textbf{not worthy enough} for regret minimization.}
\end{quote}

\vspace{-.075in}
\noindent
This means a regret-optimal mechanism must abandon ``investment'' and, thus, restrict its action plays to the \textit{upper left side} of the diagonal.
In this manner, our more ``quantitative'' principle reaches the same goal---the {\GBB} constraint---as the more ``qualitative'' principle by \cite{BCCF24}.
Technically, our principle is also more flexible for lower-bound construction, which enables us to construct $K = \Theta(\delta^{-1})$ many hard instances (while retaining the mentioned criteria) and thus show an $\Omega(T^{3 / 4})$ lower bound.
% Compared with the original constraint, our new 

% threefold:

% (i)~It focuses solely on the prices $\Price{t}$, omitting ???

% (ii)~It 

% In comparison, our new constraint is technically easier to manipulate, as 

% Technically, this new constraint partly

% This means weaker relevance among different rounds.

% thus relatively easier to manipulate.

% it is more flexible to manipulate.

% This notion might be of independent interest in the study of bilateral trading.

%Our base instance $\+D_{0}$ (see \Cref{fig:intro-LB-2} for a diagram)  and we are able to construct $K = \Theta(\delta^{-1})$ hard instances by perturbing it, which is optimal.

We have developed new techniques to implement the above discussions into a formal proof. In particular, we introduce a new constraint called \textit{\textsf{Global Price Balance}}, given by
\begin{align*}
   \textstyle
   {\bb E}\big[\sum_{t \in [T]} (\BPrice{t} - \SPrice{t})\big]
   ~\ge~ 0.
\end{align*}
Under our lower-bound construction, this new constraint turns out to be a consequence/relaxation of the original {\GBB} constraint---every {\GBB} mechanism must satisfy it---and is relatively easier to manipulate. Indeed, our $\Omega(T^{3 / 4})$ lower bound holds ``more generally'' for any mechanism that satisfies this new constraint; see \Cref{sec:LB-GBB-Partial-Correlated} for details.

\begin{table}[t]
    \centering
    \begin{tabular}{|c||>{\centering}p{1.6cm}|>{\centering}p{1.6cm}||>{\centering}p{1.6cm}|>{\centering\arraybackslash}p{1.6cm}|}
        \hline
        \rule{0pt}{12pt} & \multicolumn{2}{c||}{\LocalBudgetBalance} & \multicolumn{2}{c|}{\GlobalBudgetBalance} \\ [1pt]
        \cline{2-5}
        \rule{0pt}{12pt} & Full & Partial & Full & Partial \\ [1pt]
        \hline
        \hline
        \rule{0pt}{12pt}Independent & \multirow{2}*{\rule{0pt}{16pt}$T^{1 / 2}$} &  &  & $T^{2 / 3}$ \\ [1pt]
        \cline{1-1}\cline{5-5}
        \rule{0pt}{12pt}Correlated &  &  & $T^{1 / 2}$ & \multirow{2}*{\rule{0pt}{16pt}$T^{3 / 4}$} \\ [1pt]
        \cline{1-2}
        \rule{0pt}{12pt}Adversarial & \multicolumn{1}{c}{} & $T$ &  & \\ [1pt]
        \hline
    \end{tabular}
    \caption{\label{tbl:summary}Regret bounds of fixed-price mechanisms, up to polylogarithmic factors.}
\end{table}

\iffalse
\begin{table}[t]
    \centering
    \begin{tabular}{|c||>{\centering}p{1.6cm}|>{\centering}p{1.6cm}|>{\centering}p{1.6cm}||>{\centering}p{1.6cm}|>{\centering}p{1.6cm}|>{\centering\arraybackslash}p{1.6cm}|}
        \hline
        \rule{0pt}{12pt} & \multicolumn{3}{c||}{\LocalBudgetBalance} & \multicolumn{3}{c|}{\GlobalBudgetBalance} \\ [1pt]
        \cline{2-7}
        \rule{0pt}{12pt} & Full & Semi & Partial & Full & Semi & Partial \\ [1pt]
        \hline
        \hline
        \rule{0pt}{12pt}Independent & \multirow{2}*{\rule{0pt}{16pt}$T^{1 / 2}$} & \multirow{2}*{\rule{0pt}{16pt}$T^{2 / 3}$} &  &  & \multicolumn{1}{c}{$T^{2 / 3}$} &  \\ [1pt]
        \cline{1-1}\cline{7-7}
        \rule{0pt}{12pt}Correlated &  &  &  & $T^{1 / 2}$ & \multirow{2}*{} & \multirow{2}*{\rule{0pt}{16pt}$T^{3 / 4}$} \\ [1pt]
        \cline{1-3}
        \rule{0pt}{12pt}Adversarial & \multicolumn{2}{c}{} & $T$ &  &  & \\ [1pt]
        \hline
    \end{tabular}
    \caption{\label{tbl:summary}Regret bounds of fixed-price mechanisms, up to polylogarithmic factors.}
\end{table}
\fi

\subsection{Conclusions and Further Discussions}
\label{sec:intro:related_work}

This work establishes the (nearly) tight regret bounds of fixed-price mechanisms in various bilateral trade scenarios.
Our results together with prior work \cite{CCCFL24mor, CCCFL24jmlr, AFF24, BCCF24} complete the puzzle; see \Cref{tbl:summary} for a summary.

In doing so, we have developed two technical ingredients:
(i)~a novel algorithmic paradigm, \textit{fractal elimination}, and
(ii)~a new \textit{lower-bound construction} and its novel proof techniques.
We are optimistic about their potential for future applications and developments, as all aspects of our problem---the bilateral trade model itself, pricing-based mechanisms, and the regret minimization perspective---are fundamental to Mechanism Design and Online Optimization \cite{CL06,NRTV2007}.

Beyond the scenarios mentioned, the same problem has also been studied under certain \textit{distributional assumptions} \cite{CCCFL24mor, CCCFL24jmlr}.
In addition to (additive) regret minimization, the economic efficiency of bilateral trade has been widely studied from the perspective of \textit{(multiplicative) efficiency approximation}.
Furthermore, since the seminal work \cite{MS83}, the bilateral trade model has been extensively generalized over decades.
Below, we briefly discuss these three directions.

\vspace{.1in}
\noindent
\textbf{Assumptions on Valuations.}
Earlier work in this line of research \cite{CCCFL24mor, AFF24, CCCFL24jmlr} focused on {\LBB} mechanisms (prior to the introduction of the {\GBB} constraint \cite{BCCF24}).
However, {\LBB} mechanisms cannot achieve sublinear regret, except in the very restricted ``full feedback and independent/cor-related values'' scenarios.
To address this, \cite{CCCFL24mor} imposed the \textit{$M$-density-boundedness} assumption on ``independent/correlated values'', while \cite{CCCFL24jmlr} (the conference version appeared in COLT'23) imposed the \textit{$\sigma$-smoothness} assumption on ``adversarial values'',\footnote{More rigorously, the model for value generation in \cite{CCCFL24jmlr} is neither stronger nor weak than the ``adversarial values'' model---their model makes the $\sigma$-smoothness assumption and allows an more powerful \textit{adaptive adversary}, whereas the ``adversarial values'' model is assumption-free but only allows a less powerful \textit{oblivious adversary}.}
introducing a secondary variable $M \in [1, +\infty)$ or $\sigma \in (0, 1]$ for regret minimization.
For details, we refer the reader to these works.
(As a sanity check, their positive results all degenerate to linear regret $\Omega(T)$ when $M \to +\infty$ or $\sigma \to 0$.)

\vspace{.1in}
\noindent
\textbf{Efficiency Approximation in Bayesian Bilateral Trade.}
This work has explored the regret bounds of (repeated) fixed-price mechanisms.
Instead, traditional works in Mechanism Design have focused more on the single-round Bayesian scenario ($T = 1$),\textsuperscript{\ref{footnote:single-round}}
i.e., ``independent values'' with \textit{complete prior information} $\+D = \+D_{S} \bigotimes \+D_{B}$.
The Myerson-Satterthwaite theorem \cite{MS83} asserts that no \textit{Interim Individually Rational (IIR)}, \textit{Bayesian Incentive Compatible (BIC)}, and \textit{Budget Balanced (BB)} can guarantee ex-post efficiency (aka \textit{{\FirstBest}}); instead, they designed the efficiency-optimal such mechanism (aka \textit{{\SecondBest}}).

An interesting question is to what extent simple, well-structured mechanisms can multiplicatively approximate {\FirstBest} or {\SecondBest}.
(Here we must distinguish between approximating {\SocialWelfare} and {\GainsFromTrade}; the former is strictly easier than the latter by the same factor.)
In this context, fixed-price mechanisms, as the only EIR, DSIC, and BB mechanisms,\textsuperscript{\ref{footnote:fixed-price}} are natural candidates.
Indeed, they can achieve a constant-factor approximation to the {\FirstBest} {\SocialWelfare} \cite{CW23, LRW23}, but not to the {\FirstBest} {\GainsFromTrade} \cite{CGKLT17, BD21}.

To achieve a constant-factor approximation to the {\FirstBest} {\GainsFromTrade}, a relaxation of either the EIR constraint, the DSIC constraint, or both is necessary.
The first constant approximation was established in \cite{DMSW22} for the \textit{\textsf{Random-Offering}} mechanism and, by implication, for the {\SecondBest} mechanism; the best known bounds for the \textsf{Random-Offering} mechanism are given in \cite{BCWZ17, BDK21, F22}, and those for the {\SecondBest} mechanism are given in \cite{BM16, F22, DS24}.

\vspace{.1in}
\noindent
\textbf{Generalizations about Modelling.}
The Myerson-Satterthwaite model imposes strong restrictions:
\textit{``bilateral trade''}---a single seller, a single buyer, and a single item---\textit{``complete prior information''}, and
\textit{``independent values''}.
Abundance research has also sought to relax one or more restrictions:\\
1.~Beyond ``bilateral trade''---more generally study \textit{``double auctions''} and even \textit{``two-sided markets''} \cite{DTR17, CKLT16, BCWZ17, CGKLRT20, CGKLT17, BCGZ18, BGG20, CGMZ21, DFLLR21, BFN24, CLMZ24, LCM25a}.\\
2.~Beyond ``complete prior information''---more generally study settings with \textit{``incomplete prior information''} \cite{BSZ06, BGG20, DFLLR21, CCCFL24mor, KPV22, AFF24, CW23, CCCFL24jmlr, BCCF24, BFN24, LCM25a}.\\
3.~Beyond ``independent values''---more generally study \textit{``correlated/adversarial values''} \cite{BSZ06, CCCFL24mor, CCCFL24jmlr, AFF24, BCCF24, DS24, LCM25a}.

\subsection{\texorpdfstring{Concurrent Works \cite{LCM25a, LCM25b, DDFS25}}{}}
\label{sec:concurrent}

\ifthenelse{\equal{\version}{Full}}
{Around the time that a preliminary version of our work \cite{CJLZ25} was posted online, several other research groups independently explored related questions, leading to several concurrent works. We briefly discuss three that are most closely connected to our contributions.}
{\ifthenelse{\equal{\version}{Anonymous}}
{Several concurrent works have independently explored questions related to this paper. We briefly discuss three that are most closely connected to our contributions.}
{}}

% \yj{DeepSeek's answers about anonymity \href{https://chat.deepseek.com/share/zk84tbv9xzzhr2paje}{[link]}.}

Two consecutive works \cite{LCM25a,LCM25b} investigated related problems. The first \cite{LCM25a} independently introduced the \textit{semi-feedback} model (\Cref{sec:prelim}), which they termed \textit{asymmetric feedback}. The second \cite{LCM25b} independently proved a matching $\Omega(T^{3 / 4})$ lower bound for {\GBB} partial-feedback mechanisms in the ``correlated/adversarial values'' settings, using a different lower-bound construction.

Separately, \cite{DDFS25} studied the \textit{profit maximization} problem---the Budget Balance constraint transitions to the optimization objective---also under the EIR and DSIC constraints and also from the perspective of regret minimization.

\section{Notations and Preliminaries}
\label{sec:prelim}

Given two nonnegative integers $m \ge n \ge 0$, define the sets $[n \colon m] \defeq \{n, n + 1, \dots, m - 1, m\}$ and $[n] \defeq [1 \colon n] = \{1, 2, \dots, n\}$.
Given a (possibly random) event $\+E$, let ${\bb 1}[\+E] \in \{0, 1\}$ be the indicator function. Also, given a real number $x \in {\bb R}$, let $[x]_{+} \defeq \max\{x, 0\}$.

\vspace{.1in}
\noindent
\textbf{Repeated Bilateral Trade.}
Our model involves a $T$-round\footnote{Throughout this work, we fix $T \gg 1$ as a sufficiently large integer.} repeated game against an (oblivious) adversary:
In each round $t \in [T]$, a (new) seller and a (new) buyer seek to trade an indivisible item, which has value $\SVal{t}$ to the seller and value $\BVal{t}$ to the buyer.
We aim to maximize the \textit{economic efficiency} by designing a (repeated) mechanism that, ideally, enables trade whenever $\SVal{t} \le \BVal{t}$.
However, the adversary controls the value generation $\Val{t}_{t \in [T]}$; there are three classic models, listed below from the most to the least general.
\begin{flushleft}
\begin{itemize}
    \item \textbf{Adversarial Values:}
    The adversary determines an (arbitrary) $2T$-dimensional $[0, 1]^{2T}$-supported joint distribution $\+D$; then, the values $\Val{t}_{t \in [T]}$ across all rounds are drawn from it.\textsuperscript{\ref{footnote:adversarial}}
    
    \item \textbf{Correlated Values:}
    The adversary determines an (arbitrary) two-dimensional $[0, 1]^{2}$-supported joint distribution $\+D$; then, the values $\Val{t}$ in each round $t \in [T]$ are drawn i.i.d.\ from it.
    
    \item \textbf{Independent Values:}
    This is identical to ``correlated values'', except that $\+D$ is further required to be a product distribution $\+D_{S} \bigotimes \+D_{B}$, making all the $2T$ values $\Val{t}_{t \in [T]}$ mutually independent.
\end{itemize}
\end{flushleft}
Throughout \Cref{sec:GBB-independent,sec:GBB-independent:LB,sec:LB-GBB-Partial-Correlated,sec:appendix:GBB,sec:appendix}, we concentrate on ``correlated/independent values''.
(Our results in these settings, together with previous progress in \cite{CCCFL24mor, AFF24, BCCF24}, immediately imply the tight bounds for ``adversarial values''.)

Some of our lower bounds hold even under the \textit{density-boundedness} assumption, which was introduced to the repeated bilateral trade model by \cite{CCCFL24mor}. Clearly, this can only strengthen our hardness results.

\begin{assumption}[Density Boundedness {\cite{CCCFL24mor}}]
\label{asm:density}
\begin{flushleft}
Parameterized by $M \ge 1$, a joint distribution $\+D$ satisfies the density-boundedness assumption when its joint density function is upper-bounded by $M$.
\end{flushleft}
\end{assumption}

\noindent
\textbf{Fixed-Price Mechanisms.}
A fixed-price mechanism $\Mech = \Price{t}_{t \in [T]}$ posts two possibly randomized prices in each round $t \in [T]$, one for the seller $\SPrice{t}$ and one for the buyer $\BPrice{t}$. A trade occurs when both agents accept their respective prices.
This yields the \textit{\GainsFromTrade} $\GFT(\SVal{t}, \BVal{t}, \SPrice{t}, \BPrice{t})$ and the \textit{profit} $\Profit(\SVal{t}, \BVal{t}, \SPrice{t}, \BPrice{t})$, $\forall t \in [T]$:
\begin{align*}
    \GFT(\SVal{t}, \BVal{t}, \SPrice{t}, \BPrice{t})
    & ~\defeq~ (\BVal{t} - \SVal{t}) \cdot {\bb 1}[\SVal{t} \le \SPrice{t}] \cdot {\bb 1}[\BPrice{t} \le \BVal{t}],\\
    \Profit(\SVal{t}, \BVal{t}, \SPrice{t}, \BPrice{t})
    & ~\defeq~ (\BPrice{t} - \SPrice{t}) \cdot {\bb 1}[\SVal{t} \le \SPrice{t}] \cdot {\bb 1}[\BPrice{t} \le \BVal{t}].
\end{align*}
For notational brevity, we often simply write $\GFT^{t} = \GFT(\SVal{t}, \BVal{t}, \SPrice{t}, \BPrice{t})$ when the values $\Val{t}$ and the prices $\Price{t}$ are clear from the context; such conventions extend to other notations.

When no ambiguity arises, we use the term \textit{mechanisms} to refer to \textit{fixed-price mechanisms}.
A mechanism $\Mech$ is initially ignorant of the underlying joint distribution $\+D$ and the values $\Val{t}_{t \in [T]} \sim \+D$, but receives certain feedback at the end of each round $t \in [T]$; given the past prices $\Price{r}_{r \in [t]}$ and the past feedback, it proceeds to the next round $t + 1$ and computes the prices $\Price{t + 1}$.
There are three natural feedback models, listed below from the most to the least informative.
\begin{flushleft}
\begin{itemize}
    \item \textbf{Full Feedback:}
     $\Val{t} \in [0, 1]^{2}$ reveals both agents' values.
    
    \item \textbf{Semi Feedback:}
    This consists of four specific types:
    \begin{itemize}
        \item $(\SVal{t}, \BFeedback{t}) \in [0, 1] \times \{0, 1\}$ reveals the seller's value $\SVal{t}$ and the buyer's intention to trade $\BFeedback{t} = \BFeedback{}(\BVal{t}, \BPrice{t}) \defeq {\bb 1}[\BPrice{t} \le \BVal{t}]$.
        
        \item $(\SFeedback{t}, \BVal{t}) \in \{0, 1\} \times [0, 1]$ reveals the seller's intention to trade $\SFeedback{t} = \SFeedback{}(\SVal{t}, \SPrice{t}) \defeq {\bb 1}[\SVal{t} \le \SPrice{t}]$ and the buyer's value $\BVal{t}$.
        
        \item $(\SVal{t}, \Trade{t}) \in [0, 1] \times \{0, 1\}$ reveals the seller's value $\SVal{t}$ and the trade outcome $\Trade[]{t} = \Trade[]{}(\SVal{t}, \BVal{t}, \SPrice{t}, \BPrice{t}) \defeq {\bb 1}[\SVal{t} \le \SPrice{t} \land \BPrice{t} \le \BVal{t}] \equiv \SFeedback{t} \land \BFeedback{t}$.
        
        \item $(\Trade{t}, \BVal{t}) \in \{0, 1\} \times [0, 1]$ reveals the trade outcome $\Trade[]{t} \equiv \SFeedback{t} \land \BFeedback{t}$ and the buyer's value $\BVal{t}$.
    \end{itemize}
    
    \item \textbf{Partial Feedback:}
    This consists of two specific types:
    \begin{itemize}
        \item \textbf{Two-Bit Feedback:}
        $\Feedback{t} \in \{0, 1\}^{2}$ reveals both agents' intentions to trade.
        
        \item \textbf{One-Bit Feedback:}
        $\Trade[]{t} \equiv \SFeedback{t} \land \BFeedback{t} \in \{0, 1\}$ reveals the trade outcome.
    \end{itemize}
\end{itemize}
\end{flushleft}
Previous works \cite{CCCFL24mor, AFF24, BCCF24} have already acquired a thorough understanding of ``full feedback'' (cf.\ \Cref{tbl:full}), so we concentrate on ``semi feedback'' and ``partial feedback'' in this work.

\begin{remark}[Feedback Models]
The $1 + 4 + 2 = 7$ feedback models collectively form a partially ordered set, as depicted in \Cref{fig:feedback}. We adopt \textit{``semi feedback''} to unify the four specific models in the intermediate layer, as they all yield the same tight bound in every context. Similarly, \textit{``partial feedback''} unifies two-bit feedback and one-bit feedback for the same reason.
\end{remark}

\begin{figure}[t]
{\centering
\begin{tikzpicture}[scale = 0.8]
    \node[right] at (-10.5, 0) {full feedback};
    \node[right] at (-10.5, -2.25) {semi feedback};
    \node[right] at (-10.5, -5.25) {partial feedback};
    \draw[densely dotted] (-10.5, -0.75) to (7, -0.75);
    \draw[densely dotted] (-10.5, -3.75) to (7, -3.75);
    
    \node(n1) at (0, 0) {$\Val{t}$};
    \node(n21) at (-3, -1.5) {$(\SVal{t}, \BFeedback{t})$};
    \node(n22) at (3, -1.5) {$(\SFeedback{t}, \BVal{t})$};
    \node(n23) at (-6, -3) {$(\SVal{t}, \Trade[]{t})$};
    \node(n24) at (6, -3) {$(\Trade[]{t}, \BVal{t})$};
    \node(n31) at (0, -4.5) {$\Feedback{t}$};
    \node(n32) at (0, -6) {$\Trade[]{t}$};
    
    \draw[-latex] (n1.south west) to (n21.north east);
    \draw[-latex] (n1.south east) to (n22.north west);
    \draw[-latex] (n21.south west) to (n23.north east);
    \draw[-latex] (n21.south east) to (n31.north west);
    \draw[-latex] (n22.south west) to (n31.north east);
    \draw[-latex] (n22.south east) to (n24.north west);
    \draw[-latex] (n23.south east) to (n32.north west);
    \draw[-latex] (n31.south) to (n32.north);
    \draw[-latex] (n24.south west) to (n32.north east);
\end{tikzpicture}
\par}
\caption{A Hasse diagram of various feedback models. An arrow ``$\+F \to \+F'$'' indicates that feedback $\+F$ implies feedback $\+F'$, i.e., being more informative. (For instance, we have $\Val{t} \to (\SVal{t}, \BFeedback{t})$ given that $\BFeedback{t} = {\bb 1}[\BPrice{t} \le \BVal{t}]$.) Note that $(\SVal{t}, \BFeedback{t})$ is symmetric to $(\SFeedback{t}, \BVal{t})$, while $(\SVal{t}, \Trade[]{t})$ is symmetric to $(\Trade[]{t}, \BVal{t})$.}
\label{fig:feedback}
\end{figure}

As shown in \cite{HR87, CKLT16, DDFS25},\textsuperscript{\ref{footnote:fixed-price}} under the \textit{Ex-Post Individually Rational (EIR)} and \textit{Dominant-Strategy Incentive Compatible (DSIC)} constraints, we can restrict our attention to the family of fixed-price mechanisms.
For the sake of economic viability, we further impose the \textit{Budget Balance (BB)} constraint.
There are two notions, a stricter \textit{{\LocalBudgetBalance} ({\LBB})} constraint \cite{MS83} and a looser \textit{{\GlobalBudgetBalance} ({\GBB})} constraint \cite{BCCF24}.
\begin{flushleft}
\begin{itemize}
    \item {\LocalBudgetBalance}:
    This can be further differentiated into a stricter \textit{{\StrongBudgetBalance} ({\SBB})} constraint and a looser \textit{{\WeakBudgetBalance} ({\WBB})} constraint.
    \begin{itemize}
        \item {\StrongBudgetBalance}: 
        $\Profit^{t} = 0$ (or essentially $\SPrice{t} = \BPrice{t}$) ex-post, $\forall t \in [T]$.\\
        This requires \textit{zero profit} $\Profit^{t} = 0$ in each round, almost surely.
        
        \item {\WeakBudgetBalance}: 
        $\Profit^{t} \ge 0$ (or essentially $\SPrice{t} \le \BPrice{t}$) ex-post, $\forall t \in [T]$.\\
        This requires \textit{nonnegative profit} $\Profit^{t} \ge 0$ in each round, almost surely.
    \end{itemize}
    
    \item {\GlobalBudgetBalance}:
    $\sum_{t \in [T]} \Profit^{t} \ge 0$ ex-post.\\
    This requires \textit{nonnegative total profit} $\sum_{t \in [T]} \Profit^{t} \ge 0$, almost surely.
\end{itemize}
\end{flushleft}
As mentioned in \Cref{rmk:BB}, while prior works \cite{MS83, CCCFL24mor, CCCFL24jmlr, AFF24, BCCF24, DS24} studied {\SBB} and {\WBB} separately, they yield the same tight bound in every specific context, up to polylogarithmic factors. To reflect this unity, we unify them under {\LBB} in our presentation.

To summarize the scope of this work, there are $3 \times 3 \times 2 = 18$ specific contexts. After omitting the well-understood ``adversarial values'' and ``full feedback'', $2 \times 2 \times 2 = 8$ contexts remain, for which this work provides a thorough investigation.

% three versions, which are listed below from the most restricted one to the least restricted one.
% (The {\SBB}/{\WBB} constraints impose ``local restrictions'' to each round $t \in [T]$, while the {\GBB} constraint relaxes them to ``global restrictions'' over all rounds.)
% \begin{flushleft}
% \begin{itemize}
%     \item \textit{{\LocalBudgetBalance} ({\LBB}):}
%     \begin{itemize}
%         \item \textit{{\StrongBudgetBalance} ({\SBB}):}
%         $\SPrice{t} = \BPrice{t}$, $\forall t \in [T]$ (almost surely over all possible randomness).\\
%         Namely, we can neither run a deficit ($\SPrice{t} > \BPrice{t}$) nor extract profit ($\SPrice{t} < \BPrice{t}$) in every single round.
        
%         \item \textit{{\WeakBudgetBalance} ({\WBB}):}
%         $\SPrice{t} \le \BPrice{t}$, $\forall t \in [T]$ (almost surely over all possible randomness).\\
%         Namely, we cannot run a deficit but can extract profit in every single round.
%     \end{itemize}
    
%     \item \textit{{\GlobalBudgetBalance} ({\GBB}):}
%     $\sum_{t \in [T]} \Profit^{t} \ge 0$ (almost surely over all possible randomness).
%     Namely, we cannot run a deficit over all rounds but otherwise are unrestricted.
% \end{itemize}
% \end{flushleft}

\vspace{.1in}
\noindent
\textbf{Regret Minimization.}
We evaluate the economic efficiency of a fixed-price mechanism $\Mech$ within the \textit{regret minimization} framework.
% This ``unifies'' both measurements, {\GainsFromTrade} and {\SocialWelfare}, as the gaps $\GFT^{t} - \SW^{t} = -\SVal{t}$, $\forall t \in [T]$ are mechanism-independent; without loss of generality, we adopt {\GainsFromTrade} for our presentation.

Over the entire repeated game, this mechanism $\Mech$ induces \textit{\TotalGainsFromTrade} $\GFT_{\+D}^{\Mech}$ in expectation (over all possible randomness $\Val{t}_{t \in [T]} \sim \+D$ and $\Price{t}_{t \in [T]} \gets \Mech$).
We compare this to \textit{\textsf{Bayesian-Optimal Total Gains from Trade}} $\GFT_{\+D}^{*}$, which refers to the optimal-in-expectation fixed prices $(p^{*}, q^{*})$ with respect to the underlying joint distribution $\+D$.
\begin{align*}
    \GFT_{\+D}^{\Mech}
    & ~\defeq~ {\bb E}_{{\Val{t}_{t \in [T]} \sim \+D,\; \Price{t}_{t \in [T]} \gets \Mech}}\left[\sum_{t \in [T]}\GFT(\SVal{t}, \BVal{t}, \SPrice{t}, \BPrice{t})\right],\\
    \GFT_{\+D}^{*}
    & ~\defeq~ \max_{0 \le p^{*} \le q^{*} \le 1} {\bb E}_{{\Val{t}_{t \in [T]} \sim \+D}}\left[\sum_{t \in [T]}\GFT(\SVal{t}, \BVal{t}, p^{*}, q^{*})\right].
\end{align*}
Without ambiguity, we often write ${\bb E}_{\+D}[\cdot] = {\bb E}_{\Val{t}_{t \in [T]} \sim \+D}[\cdot]$ etc for notational brevity.

\begin{definition}[Benchmarks]
\label{def:benchmark}
% \hctodo{I think the benchmark is a definition, and (ii) seems ambiguous.}
The benchmark $\GFT_{\+D}^{*}$ is robust to the {\SBB}/{\WBB}/{\GBB} constraints, as even the optimal-in-expectation {\GBB} fixed prices $(p_{\GBB}^{*}, q_{\GBB}^{*})$, say, can satisfy the {\SBB} constraint $p_{\GBB}^{*} = q_{\GBB}^{*}$. Consider arbitrary {\GBB} fixed prices $(p, q)$. Regardless of the outcomes of the values $\Val{t}_{t \in [T]}$:
(i)~If $p < q$, then the {\SBB} fixed prices $(p', q') \defeq (p, p)$, say, must induce at least the same {\TotalGainsFromTrade}.
(ii)~If $p > q$, then the {\GBB} constraint must be violated as long as the trade succeeds $\Trade[]{t} = 1 \iff \SVal{t} \le p \land \BVal{t} \ge q$ in some round $t \in [T]$; but if otherwise the trade always fails, both $(p, q)$ and the {\SBB} fixed prices $(p'', q'') \defeq (p, p)$, say, must induce the same (zero) {\TotalGainsFromTrade}.

In contrast to such robustness of the benchmark $\GFT_{\+D}^{*}$, different choices of the {\SBB}/{\WBB}/{\GBB} constraints do affect the design and analysis of a mechanism $\Mech$.
\end{definition}

For a mechanism $\Mech$, we can define its \textit{(worst-case) regret} $\Regret^{\Mech}$ by taking into account all possible joint distributions $\+D$.
In this regard, we aim to find the \textit{minimax regret} $\Regret^{*}$ by designing a regret-optimal mechanism.
\begin{align*}
    \Regret^{\Mech}
    & ~\defeq~ \max_{\+D} \big(\GFT_{\+D}^{*} - \GFT_{\+D}^{\Mech}\big),\\
    \Regret^{*}
    & ~\defeq~ \min_{\Mech} \Regret^{\Mech}.
\end{align*}
We will also use $\Regret^{\Mech}_{\+D} \defeq \GFT_{\+D}^{*} - \GFT_{\+D}^{\Mech}$ to denote the regret of $\Mech$ on a specific distribution $\+D$. When $\Mech$ is clear from the context, we may drop it from the superscript.

\vspace{.1in}
\noindent
\textbf{Probability and Information.}
Let $(\Omega, \+F, {\bb P})$ be a probability space. For a sequence of random variables $\+R = (X_1, \dots, X_n)$, we denote ${\bb P}_{\+R}$ as the pushforward measure of ${\bb P}$ by the random variables in $\+R$; formally, ${\bb P}_{\+R}(A) \defeq {\bb P}[\set{\omega \in \Omega: (X_1(\omega), \dots, X_n(\omega) \in A}]$, for every measurable set $A$.\footnote{For readers unfamiliar with this notion, when $\+R$ consists of a single random variable $X$, we can informally interpret ${\bb P}_{\+R}$ as the marginal probability of $X$ when $X$ is discrete, or as the marginal density of $X$ when $X$ is absolutely continuous w.r.t.\ the Lebesgue measure. We adopt this more general measure-theoretic notion, since the random variables $\+R = (X_1, \dots, X_n)$ under consideration might be neither discrete nor absolutely continuous.} For a sub $\sigma$-algebra $\+F' \subseteq \+F$, we also define ${\bb P}_{\+R \;|\; \+F'}(A) \defeq {\bb P}[\set{\omega \in \Omega: (X_1(\omega), \dots, X_n(\omega)) \in A} \;|\; \+F']$ as the regular conditional pushforward measure of ${\bb P}$ by random variables in $\+R$ given $\+F'$.

Let ${\bb P}$ and ${\bb Q}$ be two probability measures on the same measurable space $(\Omega, \+F)$. We define their \emph{total variation distance} as
\[
    \distTV({\bb P}, {\bb Q})
    ~\defeq~ \sup_{A \in \+F} \abs{{\bb P}(A) - {\bb Q}(A)}.
\]
% When ${\bb P} \ll {\bb Q}$,
When ${\bb P}$ is absolutely continuous w.r.t.\ ${\bb Q}$, we also denote their \emph{Kullback-Leibler (KL) divergence} as
\[
    \textstyle
    \distKL({\bb P}, {\bb Q})
    ~\defeq~ {\bb E}_{\bb P}\left[\log\dv{{\bb P}}{{\bb Q}}\right],
\]
where $\log$ is the natural logarithm and $\dv{{\bb P}}{{\bb Q}}$ is the Radon-Nikodym derivative.

For two Bernoulli distributions, it is easy to show the following bound on their KL divergence.

\begin{fact}[KL Divergence]
\label{lem:KL-Bernoulli}
\begin{flushleft}
% $\distKL(\Bern(a),\ \Bern(b)) \le 5 \cdot (a - b)^{2}$, for $a \in [\frac{1}{5}, \frac{1}{5}]$ and $b \in [a - \frac{1}{10}, a + \frac{1}{10}]$.
$\distKL(\Bern(a),\ \Bern((1 \pm \delta) a)) \le 2a \delta^{2}$, when $0 \le a, \delta \le \frac{1}{2}$.
\end{flushleft}
\end{fact}

The following fact is known as the \textit{chain rule} for KL divergences (see, e.g., \cite[Chapter~14]{LS20}).
% \cite{LS20}

\begin{fact}
\label{prop:chain-rule-for-KL}
\begin{flushleft}
Let $X_1, X_2, \dots, X_n$ be a sequence of (possibly correlated) random variables, and let $\+F_i \defeq \sigma(X_1, \dots, X_i)$ for every $i \in [0: n]$, then $\distKL({\bb P}, {\bb Q}) = \sum_{i \in [n]} {\bb E}_{\bb P}\left[\distKL({\bb P}_{X_i \;|\; \+F_{i - 1}}, {\bb Q}_{X_i \;|\; \+F_{i - 1}})\right]$.
\end{flushleft}
\end{fact}

Also, \textit{Pinsker's inequality} relates total variation distance and KL divergence \cite[Lemma~2.5]{T09}.

\begin{fact}[Pinsker's inequality \cite{T09}]
\label{prop:pinsker}
\begin{flushleft}
$\distTV({\bb P}, {\bb Q}) \le \sqrt{\frac{1}{2} \distKL({\bb P}, {\bb Q})}$.
% \[
%     \distTV({\bb P}, {\bb Q})
%     ~\le~ \sqrt{\tfrac{1}{2} \distKL({\bb P}, {\bb Q})}.
% \]
\end{flushleft}
\end{fact}

\vspace{.1in}
\noindent
\textbf{Concentration Inequalities.}
Below we present two standard \textit{concentration inequalities}.

\begin{fact}[{Hoeffding's Inequality \cite{MU17}}]
\label{fact:Hoeffding}
\begin{flushleft}
Let  $X_{1}, X_{2}, \dots, X_{n} \in [0,1]$ be a sequence of independent random variables, and let $M_{n} = \frac{1}{n} \sum_{i \in [n]} X_{i}$ be their empirical mean, then
\begin{align*}
    & {\bb P}\big[\abs{M_{n} - {\bb E}[M_{n}]} \ge r\big]
    ~\le~ 2\exp(-2n r^{2}),
    && \forall r \ge 0.
\end{align*}
\end{flushleft}
\end{fact}

\begin{fact}[{Bernstein inequality \cite{MU17}}]
\label{fact:Bernstein}
\begin{flushleft}
Let $X_{1}, X_{2}, \dots, X_{n} \in [0, 1]$ be a sequence of independent random variables each with variance at most $s^{2}$, and let $M_{n} = \frac{1}{n} \sum_{i \in [n]} X_{i}$ be their empirical mean, then
\begin{align*}
    & {\bb P}\big[\abs{M_{n} - {\bb E}[M_{n}]} \ge r\big]
    ~\le~ 2\exp(-\tfrac{n r^{2}}{2 \cdot (s^{2} + r / 3)}),
    && \forall r \ge 0.
\end{align*}
\end{flushleft}
\end{fact}

% \newpage

\section{\texorpdfstring{$\tO(T^{2 / 3})$}{} {\GBB} Partial-Feedback Upper Bound for Independent Values}
\label{sec:GBB-independent}

In this section, we examine the power of ``fixed-price mechanisms with the {\GlobalBudgetBalance} ({\GBB}) constraint and one-bit feedback'' in the ``independent values'' setting. Specifically, we will establish (\Cref{thm:GBB-independent:one}) the following algorithmic result.

\begin{theorem}[{\GBB} One-Bit-Feedback Upper Bound for Independent Values]
\label{thm:GBB-independent:one}
\begin{flushleft}
In the ``independent values'' setting, there is a ``{\GBB} one-bit-feedback fixed-price mechanism'' achieving $\tO(T^{2 / 3})$ regret.
\end{flushleft}
\end{theorem}

\noindent
Later in \Cref{sec:GBB-independent:LB}, up to polylogarithmic factors, we will establish (\Cref{thm:GBB-independent:LB}) a matching lower bound (even if one-bit feedback is replaced by the more informative \textit{semi feedback}).

In the literature, only a trivial $\tO(T^{3 / 4})$ upper bound and a trivial $\Omega(T^{1 / 2})$ lower bound were known---the upper bound is an implication from \cite[Theorem~5.4]{BCCF24} for ``one-bit feedback, adversarial values'', and the lower bound is an implication from \Cref{thm:appendix:GBB} for ``full feedback, independent values''. Nonetheless, our algorithmic result here and our hardness result later in \Cref{sec:GBB-independent:LB} together close the gap.

\subsection{Mechanism Design}

Our fixed-price mechanism, called {\GBBOneBit} and presented in \Cref{alg:GBB-independent:one}, is built on the mechanism design framework proposed by \cite[Section~3]{BCCF24}.
This fixed-price mechanism has three phases:

\noindent
Phase~\ref{alg:GBB-independent:one:profit} invokes a subroutine, called {\ProfitMax} and depicted in \Cref{prop:BCCF24}, which takes actions only from the \textit{upper-left} action halfspace $\{(p, q) \in [0, 1]^{2} \;|\; p \le q\}$, thus \textit{nonnegative} profit $\ge 0$ per round.
The main purpose of this phase is to accumulate sufficient profit, say $\tOmega(T^{2 / 3})$, while just incurring tolerable regret, say $\tO(T^{2 / 3})$.
Regarding the {\GBB} constraint, this cumulative profit makes mechanism design in subsequent phases more flexible.

\begin{algorithm}[t]
\caption{\label{alg:GBB-independent:one}
{\GBBOneBit}}
\begin{algorithmic}[1]
    \State Run the subroutine $\ProfitMax(K, \beta)$.
    \Comment{Cf.\ \Cref{prop:BCCF24} and \Cref{cor:BCCF24}.}
    \label{alg:GBB-independent:one:profit}
    
    \State Run the subroutine $\FractalElimination(0, [1 : K], 0, 0)$.
    \Comment{Cf.\ \Cref{alg:GBB-independent:fractal}.}
    \label{alg:GBB-independent:one:exploration}
    
    \State Take actions $\{a_{k, k}\}_{k \in \CANDIDATE_{L + 1}}$ survived in Line~\ref{alg:GBB-independent:one:exploration}, in an arbitrary manner, for the remaining rounds.
    \label{alg:GBB-independent:one:exploitation}
\end{algorithmic}
\end{algorithm}

\noindent
Phase~\ref{alg:GBB-independent:one:exploration} invokes a subroutine, called {\FractalElimination} and shown in \Cref{alg:GBB-independent:fractal}, which takes actions only from the \textit{lower-right} action halfspace $\{(p, q) \in [0, 1]^{2} \;|\; p > q\}$, thus \textit{nonpositive} profit $\le 0$ per round.
Concretely, this subroutine begins with a set of $K = \tTheta(T^{1 / 3})$ many \textit{candidate nearly {\GFT}-optimal actions} (or \textit{candidates} in short) indexed by $\CANDIDATE_{0} = [1 : K]$; the {\GFT}-optimal candidate is ensured to be a good enough approximation to the {\GFT}-optimal action $(p^{*}, q^{*}) \in [0, 1]^{2}$ in the whole action space.
The subroutine works in $L + 1 \approx \log(K)$ many stages; a single stage $\ell \in [0 : L]$ leverages \textit{one-bit feedback} to distinguish the survival candidates $\CANDIDATE_{\ell}$ hitherto (by taking actions from not only $\CANDIDATE_{\ell}$ themselves, but also other actions in the lower-right action halfspace\ignore{ $\{(p, q) \in [0, 1]^{2} \;|\; p > q\}$}), obtaining a more accurate location $\CANDIDATE_{\ell + 1} \subseteq \CANDIDATE_{\ell}$ of the optimal candidate.
After all the $L + 1 \approx \log(K)$ many stages, the ultimate candidates $\CANDIDATE_{L + 1}$ all will be good enough, compared even with the optimal action $(p^{*}, q^{*}) \in [0, 1]^{2}$ in the whole action space.

\noindent
Phase~\ref{alg:GBB-independent:one:exploitation} simply exploits the ultimate candidates $\CANDIDATE_{L + 1}$, in an arbitrary manner.

\noindent
Remarkably, as it turns out, Phases~\ref{alg:GBB-independent:one:exploration} and \ref{alg:GBB-independent:one:exploitation} never exhaust the profit accumulated in Phase~\ref{alg:GBB-independent:one:profit}, so the whole fixed-price mechanism {\GBBOneBit} satisfies the {\GBB} constraint.

\vspace{.1in}
\noindent
{\bf Preliminaries.}
Our fixed-price mechanism {\GBBOneBit} will use the following parameters. Specifically, both subroutines {\ProfitMax} and {\FractalElimination} will use the \textit{discretization parameter} $K$, only the former will use the \textit{profit threshold} $\beta$, and only the latter will use the others $L$, $\delta$, and $\gamma_{\ell}$'s.
\begin{align*}
    K &\textstyle ~\defeq~ \frac{1}{8}T^{1 / 3}\log^{-2 / 3}(T),
    \tag{the discretization parameter} \\
    \beta &\textstyle ~\defeq~ 9T^{2 / 3}\log^{2 / 3}(T),
    \tag{the profit threshold} \\
    L &\textstyle ~\defeq~ \frac{1}{3}\log(T),
    \tag{the number of stages} \\
    \CANDIDATE_{0} &\textstyle ~\defeq~ [1 : K],
    \tag{for initialization} \\
    \delta &\textstyle ~\defeq~ T^{-4 / 3}\log^{-1 / 3}(T),
    \tag{for confidence levels} \\
    \gamma_{\ell} &\textstyle ~\defeq~ 2^{-\ell / 2} K^{-1 / 2} + (6\ell - 1) K^{-1},
    && \forall \ell \in [0 : L + 1].
    \tag{for confidence intervals}
\end{align*}

In expectation over the randomness of values $(S, B) \sim \+D_{S} \bigotimes \+D_{B}$, an action $(p, q) \in [0, 1]^{2}$ induces (expected) {\GainsFromTrade} $\GFT(p, q)$, (expected) profit $\Profit(p, q)$, and (expected) regret $\Regret(p, q)$. It is easy to see that these formulae are $[0, 1]$-valued.
\begin{align*}
    \GFT(p, q)
    & ~\defeq~ {\bb E}_{(S, B) \sim \+D_{S} \bigotimes \+D_{B}}[\GFT(S, B, p, q)],
    && \forall (p, q) \in [0, 1]^{2}, \\
    \Regret(p, q)
    & ~\defeq~ \big(\max_{0 \le p' \le q' \le 1} \GFT(p, q)\big) - \GFT(p, q),
    && \forall (p, q) \in [0, 1]^{2}, \\
    \Profit(p, q)
    & ~\defeq~ {\bb E}_{(S, B) \sim \+D_{S} \bigotimes \+D_{B}}[\Profit(S, B, p, q)] \\
    & ~\:=~  (q - p) \cdot \+D_{S}(p) \cdot (1 - \+D_{B}(q)),
    && \forall (p, q) \in [0, 1]^{2}.
\end{align*}
Here we abuse notation: $\+D_{S}(p) \defeq {\bb P}_{S \sim \+D_{S}}[S \le p]$ and $\+D_{B}(q) \defeq 1 - {\bb P}_{B \sim \+D_{B}}[B \ge q] = {\bb P}_{B \sim \+D_{B}}[B < q]$ denote the corresponding \textit{cumulative distribution functions (CDF's)}.

The following \Cref{lem:GFT-independent} shows a useful decomposition of the formula $\GFT(p, q)$.

\begin{lemma}[{\GainsFromTrade} for Independent Values]
\label{lem:GFT-independent}
\begin{flushleft}
In the ``independent values'' settings, 
\begin{align*}
    \GFT(p, q)
    & ~=~ \Horizonal{p, q} + \Vertical{p, q} + \Profit(p, q),
    && \forall (p, q) \in [0, 1]^{2}.
\end{align*}
Here the terms $\Horizonal{p, q} \defeq \int_{0}^{p} \+D_{S}(x) \dd x \cdot (1 - \+D_{B}(q))$ and
$\Vertical{p, q} \defeq \+D_{S}(p) \cdot \int_{q}^{1} (1 - \+D_{B}(y)) \dd y$.
\end{flushleft}
\end{lemma}

\begin{proof}
By the definition of $\GFT(p, q)$, we deduce that
\begin{align*}
    \textstyle
    \GFT(p, q)
    &\textstyle ~=~ {\bb E}_{(S, B) \sim \+D_{S} \bigotimes \+D_{B}}\big[\GFT(S, B, p, q)\big] \\
    &\textstyle ~=~ {\bb E}_{(S, B) \sim \+D_{S} \bigotimes \+D_{B}}\big[B \cdot {\bb 1}[S \le p \land q \le B]\big]
    - {\bb E}_{(S, B) \sim \+D_{S} \bigotimes \+D_{B}}\big[S \cdot {\bb 1}[S \le p \land q \le B]\big] \\
    &\textstyle ~=~ \+D_{S}(p) \cdot {\bb E}_{B \sim \+D_{B}}\big[B \cdot {\bb 1}[q \le B]\big]
    - {\bb E}_{S \sim \+D_{S}}\big[S \cdot {\bb 1}[S \le p]\big] \cdot (1 - \+D_{B}(q)) \\
    &\textstyle ~=~ \+D_{S}(p) \cdot \big(q \cdot (1 - \+D_{B}(q)) + \int_{q}^{1} (1 - \+D_{B}(y)) \dd y\big)
    - \int_{0}^{p} (\+D_{S}(p) - \+D_{S}(x)) \dd x \cdot (1 - \+D_{B}(q)) \\
    &\textstyle ~=~ \Horizonal{p, q} + \Vertical{p, q} + \Profit(p, q).
\end{align*}
Here the second step uses the formula $\GFT(S, B, p, q) = (B - S) \cdot {\bb 1}[S \le p \land q \le B]$ and the linearity of expectation.
The third step uses the independence of values $(S, B) \sim \+D_{S} \bigotimes \+D_{B}$.
The fourth step follows from elementary algebra.
And the last step uses the defining formulae of $\Horizonal{p, q}$, $\Vertical{p, q}$, and $\Profit(p, q)$.

This finishes the proof of \Cref{lem:GFT-independent}.
\end{proof}

In addition, we recall that \textit{one-bit feedback} $\Trade[]{t} = \Trade[]{}(\SVal{t}, \BVal{t}, \SPrice{t}, \BPrice{t}) \in \{0, 1\}$ reveals whether or not the trade succeeded in a single round $t \in [T]$.
\begin{align*}
    \Trade[]{}(\SVal{t}, \BVal{t}, \SPrice{t}, \BPrice{t}) ~=~ {\bb 1}[\SVal{t} \le \SPrice{t}] \cdot {\bb 1}[\BPrice{t} \le \BVal{t}].
\end{align*}

\noindent
{\bf The Subroutine {\ProfitMax}.}
To understand the performance guarantees of our fixed-price mechanism {\GBBOneBit}, all we need to know about (Phase~\ref{alg:GBB-independent:one:profit} of {\GBBOneBit}) the subroutine {\ProfitMax} can be summarized into the following \Cref{prop:BCCF24}, which is quoted (or, indeed, slightly rephrased) from \cite[Lemma~5.1]{BCCF24}.\footnote{This fixed-price mechanism {\ProfitMax} is built on the \textnormal{\textsc{EXP3.P}} learning algorithm by \cite{ACFS02}; its performance guarantees given in \Cref{prop:BCCF24} were shown by \cite[Lemma~5.1]{BCCF24} for ``adversarial values'', which accommodates ``independent values''.}
As mentioned, the purpose of {\ProfitMax} is to accumulate sufficient profit (at the cost of tolerable regret), making mechanism design in subsequent phases more flexible.
For detailed implementation of {\ProfitMax}, the interested reader can reference \cite[Sections~3 and 5]{BCCF24}.

\begin{proposition}[{\cite[Lemma~5.1]{BCCF24}}]
\label{prop:BCCF24}
\begin{flushleft}
There exists a fixed-price mechanism $\ProfitMax(K', \beta')$ with one-bit feedback, on input a discretization parameter $K' \ge 1$ and a profit threshold $\beta' > 0$, such that:
\begin{enumerate}
    \item\label{prop:BCCF24:1}
    It takes actions $\{\Price{t}\}_{t = 1, 2, \dots}$ only from a size-$\abs{\+{F}_{K'}} = 2K'(\log(T) + 1)$ discrete subset $\+{F}_{K'} \subseteq \{(p, q) \in [0, 1]^{2} \;|\; p \le q\}$ of the upper-left action halfspace.\\
    Thus, the per-round profit is nonnegative $\Profit(\SVal{t}, \BVal{t}, \SPrice{t}, \BPrice{t}) \ge 0$, $\forall t = 1, 2, \dots$, almost surely.\footnote{\label{footnote:GBB-independent:exploration-exploitation}Namely, the per-round profit $\Profit(\SVal{t}, \BVal{t}, \SPrice{t}, \BPrice{t})$ satisfies the claim, almost surely over the randomness of both the values $\Val{t} \sim \+D_{S} \bigotimes \+D_{B}$ and the action $\Price{t}$.
    Instead, the per-round regret $\Regret(\SPrice{t}, \BPrice{t})$ satisfies the claim, ``just'' almost surely over the randomness of the action $\Price{t}$.}
    
    \item\label{prop:BCCF24:2}
    It terminates at the end of some round $T' \in [T]$, which has two possibilities:\\
    (i)~$T' \in [T]$ is the first round such that $\sum_{t \in [T']} \Profit(\SVal{t}, \BVal{t}, \SPrice{t}, \BPrice{t}) \ge \beta'$, if existential.\\
    (ii)~$T' = T$, if $\sum_{t \in [T]} \Profit(\SVal{t}, \BVal{t}, \SPrice{t}, \BPrice{t}) < \beta'$.\\
    In either case, with probability $1 - T^{-1}$, the cumulative regret $\sum_{t \in [T']} \Regret(\SPrice{t}, \BPrice{t})$ satisfies that
    \begin{align*}
        \textstyle
        \sum_{t \in [T']} \Regret(\SPrice{t}, \BPrice{t})
        ~\le~ (8\beta' + 8)\log(T) + \frac{5T}{K'} + 256\sqrt{T \abs{\+{F}_{K'}} \log(T \abs{\+{F}_{K'}})} \cdot \log(T).
    \end{align*}
\end{enumerate}
\end{flushleft}
\end{proposition}

\begin{corollary}[{\ProfitMax}; Instantiation]
\label{cor:BCCF24}
\begin{flushleft}
In the context of \Cref{prop:BCCF24}, set $K' \gets K$ and $\beta' \gets \beta$.
Then in either case, with probability $1 - T^{-1}$, the cumulative regret $\sum_{t \in [T']} \Regret(\SPrice{t}, \BPrice{t}) \le 220T^{2 / 3}\log^{5 / 3}(T)$.
\end{flushleft}
\end{corollary}

\begin{proof}
By setting $K' \gets K$ and $\beta' \gets \beta$, we deduce from \Cref{prop:BCCF24:2} of \Cref{prop:BCCF24} that
\begin{align*}
    \textstyle
    \sum_{t \in [T']} \Regret(\SPrice{t}, \BPrice{t})
    &\textstyle ~\le~ (8\beta + 8)\log(T) + \frac{5T}{K} + 256\sqrt{T \abs{\+{F}_{K}} \log(T \abs{\+{F}_{K}})} \cdot \log(T) \\
    &\textstyle ~\le~ 72T^{2 / 3}\log^{5 / 3}(T) + 8\log(T) + 40T^{2 / 3}\log^{2 / 3}(T) \\
    &\textstyle \phantom{~=~} + 256\sqrt{(\frac{1}{3} \pm o(1))T^{4 / 3}\log^{4 / 3}(T)} \cdot \log(T) \\
    &\textstyle ~=~ (72 + 256 / \sqrt{3} \pm o(1))T^{2 / 3}\log^{5 / 3}(T) \\
    &\textstyle ~\le~ 220T^{2 / 3}\log^{5 / 3}(T).
\end{align*}
Here the second step substitutes $K = \frac{1}{8}T^{1 / 3}\log^{-2 / 3}(T)$, $\beta = 9T^{2 / 3}\log^{2 / 3}(T)$, $\abs{\+{F}_{K}} = 2K(\log(T) + 1) = (\frac{1}{4} \pm o(1))T^{1 / 3}\log^{1 / 3}(T)$, and $\log(T \abs{\+{F}_{K}}) = (\frac{4}{3} \pm o(1)) \log(T)$.
And the last step uses $256 / \sqrt{3} \pm o(1) \approx 147.8017 \pm o(1) < 148$, which holds for any large enough $T \gg 1$.

This finishes the proof of \Cref{cor:BCCF24}.
\end{proof}

In the rest of \Cref{sec:GBB-independent}, we would call the subroutine {\ProfitMax} \textit{``successful''} if its cumulative regret satisfies the bound $\sum_{t \in [T']} \Regret(\SPrice{t}, \BPrice{t}) \le 220T^{2 / 3}\log^{5 / 3}(T)$ given in \Cref{cor:BCCF24}, or \textit{``failed''} otherwise.

\vspace{.1in}
\noindent
{\bf The Subroutine {\FractalElimination}.}
Based on the profit accumulated above ($\ge \beta = \tTheta(T^{2 / 3})$, say), our fixed-price mechanism {\GBBOneBit} (Phase~\ref{alg:GBB-independent:one:exploration} thereof) then invokes the subroutine {\FractalElimination} to distinguish the optimal action $a^{*} = (p^{*}, q^{*}) \in [0, 1]^{2}$ (cf.\ \Cref{rmk:benchmark} and \Cref{lem:GFT-independent}), or rather, its good enough approximations, while preserving the {\GBB} constraint.

Now let us elaborate on this subroutine {\FractalElimination}; see \Cref{alg:GBB-independent:fractal} for its implementation and \Cref{fig:GBB-independent:1} for a diagram.
Before all else, we use our discretization parameter $K = \tTheta(T^{2 / 3})$ to construct the following \textit{$\frac{1}{K}$-net} $\{a_{i, j}\}_{1 \le i, j \le K}$ of the whole action space $[0, 1]^{2}$ (yet {\FractalElimination} only takes actions from the \textit{lower-right} half $\{a_{i, j}\}_{1 \le j \le i \le K}$).
Among these discrete actions, we designate $\{a_{k, k}\}_{k \in [1 : K]}$ as \textit{candidates} of ``good enough approximations to the optimal action $a^{*}$''; indeed, the optimal candidate $a_{\mu, \mu}$ (say) is a good enough $\frac{1}{K}$-approximation to the optimal action $a^{*}$; see the proof of \Cref{lem:GBB-independent:exploration}.
\begin{align*}
    a_{i, j}
    &\textstyle ~\defeq~ (\frac{i}{K}, \frac{j - 1}{K}),
    && \forall 1 \le i, j \le K.
\end{align*}

In regard to discrete actions $\{a_{i, j}\}_{1 \le i, j \le K}$ and one-bit feedback $\Trade[]{t} = {\bb 1}[\SVal{t} \le \SPrice{t}] \cdot {\bb 1}[\BPrice{t} \le \BVal{t}] \in \{0, 1\}$, we define the \textit{trade rates} $\{\Z_{i, j}\}_{1 \le i, j \le K}$ as follows.
It is easy to see that $\Z_{i, j} \in [0, 1]$ and the monotonicity $\Z_{1, j} \le \dots \le \Z_{i, j} \le \dots \le \Z_{K, j}$ and $\Z_{i, 1} \ge \dots \ge \Z_{i, j} \ge \dots \ge \Z_{i, K}$.
\begin{align}
    \Z_{i, j}
    &\textstyle ~\defeq~ {\bb E}_{(S, B) \sim \+D_{S} \bigotimes \+D_{B}}[Z(S, B, a_{i, j})]
    \notag \\
    &\textstyle ~\:=~ \+D_{S}(\frac{i}{K}) \cdot (1 - \+D_{B}(\frac{j - 1}{K})),
    && \forall 1 \le i, j \le K.
    \label{eq:GBB-independent:Trade}
\end{align}
In regard to (\Cref{lem:GFT-independent}) the decomposition $\GFT(p, q) = \Horizonal{p, q} + \Vertical{p, q} + \Profit(p, q)$, $\forall (p, q) \in [0, 1]^{2}$, we define the following terms $\ApxHorizonal{[\sigma : \tau], j}$ and $\ApxVertical{i, [\sigma : \tau]}$.\footnote{\label{footnote:ApxHorizonal-ApxVertical}For notational consistency, we let $\ApxHorizonal{[\sigma : \tau], j} \defeq 0$ when $[\sigma : \tau] = \emptyset \iff \sigma > \tau$; similarly for $\ApxVertical{i, [\sigma : \tau]}$.}
\begin{align}
    \ApxHorizonal{[\sigma : \tau], j}
    &\textstyle ~\defeq~ \frac{1}{K} \sum_{i \in [\sigma : \tau]} \Z_{i, j},
    && \forall [\sigma : \tau] \subseteq [1 : K],\ \forall j \in [1 : K],
    \label{eq:GBB-independent:ApxHorizonal} \\
    \ApxVertical{i, [\sigma : \tau]}
    &\textstyle ~\defeq~ \frac{1}{K} \sum_{j \in [\sigma : \tau]} \Z_{i, j},
    && \forall [\sigma : \tau] \subseteq [1 : K],\ \forall i \in [1 : K].
    \label{eq:GBB-independent:ApxVertical}
\end{align}
It is easy to check that all these terms are $[0, 1]$-bounded and, specifically, that $\ApxHorizonal{[1 : i], j} = \Horizonal{a_{i, j}} \pm K^{-1}$ and $\ApxVertical{i, [j : K]} = \Vertical{a_{i, j}} \pm K^{-1}$, $\forall 1 \le i, j \le K$; see the proof of \Cref{lem:GBB-independent:EstGFT}.
Moreover, every candidate $k \in [1 : K]$ induces negligible profit $\Profit(a_{k, k}) = \pm K^{-1}$; again, see the proof of \Cref{lem:GBB-independent:EstGFT}.

\begin{algorithm}[t]
\caption{\label{alg:GBB-independent:fractal}
$\FractalElimination(\ell, [\sigma : \tau], \horizonal, \vertical)$}
\begin{algorithmic}[1]
    \Require $\ell \in [0 : L]$---the current stage.
    
    $[\sigma : \tau] \subseteq [1 : K]$---the considered segment.
    
    $\horizonal \in [0, 1]$---an estimate of $\ApxHorizonal{[1 : \sigma - 1], \sigma}$.
    
    $\vertical \in [0, 1]$---an estimate of $\ApxVertical{\tau, [\tau + 1 : K]}$.
    
    \Statex
    
    \If{$\ell > L$}
        Quit.
    \EndIf
    
    \State Take actions $\{a_{i, \sigma}\}_{i \in [\sigma : \tau]} \cup \{a_{\tau, j}\}_{j \in [\sigma : \tau]}$ each for $2^{\ell + 2}K\ln(\frac{2}{\delta})$ rounds, thus one-bit feedback $\Trade[]{t}$'s.
    \label{alg:GBB-independent:fractal:1}
    
    \State $\{\EstTrade{i, \sigma}\}_{i \in [\sigma : \tau]} \cup \{\EstTrade{\tau, j}\}_{j \in [\sigma : \tau]} \gets \text{``empirical means of (index-wise) one-bit feedback $\Trade[]{t}$'s by Line~\ref{alg:GBB-independent:fractal:1}''}$.
    \label{alg:GBB-independent:fractal:EstTrade}
    
    \For{every candidate $k \in \CANDIDATE_{\ell} \bigcap [\sigma : \tau]$ in the considered segment}
    \label{alg:GBB-independent:fractal:segment}
        \State $\EstHorizonal{[\sigma : k], \sigma} \gets \underbrace{\frac{1}{K} \sum_{i \in [\sigma : k]} \EstTrade{i, \sigma}}_{\text{Line~\ref{alg:GBB-independent:fractal:EstTrade}}}$
        and $\EstVertical{\tau, [k : \tau]} \gets \underbrace{\frac{1}{K} \sum_{j \in [k : \tau]} \EstTrade{\tau, j}}_{\text{Line~\ref{alg:GBB-independent:fractal:EstTrade}}}$.
        \label{alg:GBB-independent:fractal:EstHorizonal-EstVertical}
        
        \State $\EstGFT[\ell]{k} \gets (\horizonal + \underbrace{\EstHorizonal{[\sigma : k], \sigma}}_{\text{Line~\ref{alg:GBB-independent:fractal:EstHorizonal-EstVertical}}}) \cdot \underbrace{[\EstTrade{\tau, k} / \EstTrade{\tau, \sigma}]_{\downarrow 1}}_{\text{Line~\ref{alg:GBB-independent:fractal:EstTrade}}}
        + (\vertical + \underbrace{\EstVertical{\tau, [k : \tau]}}_{\text{Line~\ref{alg:GBB-independent:fractal:EstHorizonal-EstVertical}}}) \cdot \underbrace{[\EstTrade{k, \sigma} / \EstTrade{\tau, \sigma}]_{\downarrow 1}}_{\text{Line~\ref{alg:GBB-independent:fractal:EstTrade}}}$.
        \label{alg:GBB-independent:fractal:EstGFT}
    \EndFor
    
    \Statex
    \Statex
    
    \State $\CANDIDATE_{\ell + 1} \gets \{k \in \CANDIDATE_{\ell} \;|\; \EstGFT[\ell]{k} \ge \max_{c \in \CANDIDATE_{\ell}} \EstGFT[\ell]{c} - 2\gamma_{\ell + 1}\}$.
    \Comment{$\CANDIDATE_{0} = [1 : K]$.}
    \label{alg:GBB-independent:fractal:CANDIDATE}
    
    \State $\sigma' \gets \min \CANDIDATE_{\ell + 1} \bigcap [\sigma : \frac{\sigma + \tau}{2}]$ and $\tau' \gets \max \CANDIDATE_{\ell + 1} \bigcap [\sigma : \frac{\sigma + \tau}{2}]$.
    \Comment{The ``lower-left'' half-segment.}
    \label{alg:GBB-independent:fractal:left-index}
    
    \State $\sigma'' \gets \min \CANDIDATE_{\ell + 1} \bigcap [\frac{\sigma + \tau}{2} + 1 : \tau]$ and $\tau'' \gets \max \CANDIDATE_{\ell + 1} \bigcap [\frac{\sigma + \tau}{2} + 1 : \tau]$.
    \Comment{The ``upper-right'' half-segment.}
    \label{alg:GBB-independent:fractal:right-index}
    
    \State Skip Lines~\ref{alg:GBB-independent:fractal:left-1} to \ref{alg:GBB-independent:fractal:left-4} (resp.\ Lines~\ref{alg:GBB-independent:fractal:right-1} to \ref{alg:GBB-independent:fractal:right-4}) when $\CANDIDATE_{\ell + 1} \bigcap [\sigma : \frac{\sigma + \tau}{2}] = \emptyset$ (resp.\ $\CANDIDATE_{\ell + 1} \bigcap [\frac{\sigma + \tau}{2} + 1 : \tau] = \emptyset$).
    \label{alg:GBB-independent:fractal:skip}
    
    \Statex
    
    \State Take actions $\{a_{\tau, \sigma},\ a_{\tau', \sigma},\ a_{\tau, \sigma'}\}$ each for $\frac{1}{2}K^{2}\ln(\frac{2}{\delta})$ rounds.
    \label{alg:GBB-independent:fractal:left-1}
    
    \State $\{\EstTrade{\tau, \sigma},\ \EstTrade{\tau', \sigma},\ \EstTrade{\tau, \sigma'}\} \gets \text{``empirical means of (index-wise) one-bit feedback $\Trade[]{t}$'s by Line~\ref{alg:GBB-independent:fractal:left-1}''}$.
    \label{alg:GBB-independent:fractal:left-2}
    
    \State $\horizonal' \gets (\horizonal + \underbrace{\EstHorizonal{[\sigma : \sigma' - 1], \sigma}}_{\text{Line~\ref{alg:GBB-independent:fractal:EstHorizonal-EstVertical}}}) \cdot \underbrace{[\EstTrade{\tau, \sigma'} / \EstTrade{\tau, \sigma}]_{\downarrow 1}}_{\text{Line~\ref{alg:GBB-independent:fractal:left-2}}}$
    and $\vertical' \gets (\vertical + \underbrace{\EstVertical{\tau, [\tau' + 1 : \tau]}}_{\text{Line~\ref{alg:GBB-independent:fractal:EstHorizonal-EstVertical}}}) \cdot \underbrace{[\EstTrade{\tau', \sigma} / \EstTrade{\tau, \sigma}]_{\downarrow 1}}_{\text{Line~\ref{alg:GBB-independent:fractal:left-2}}}$.
    \label{alg:GBB-independent:fractal:left-3}
    
    \State $\FractalElimination(\ell + 1, [\sigma' : \tau'], \horizonal', \vertical')$.
    \label{alg:GBB-independent:fractal:left-4}
    
    \Statex
    
    \State Take actions $\{a_{\tau, \sigma},\ a_{\tau'', \sigma},\ a_{\tau, \sigma''}\}$ each for $\frac{1}{2}K^{2}\ln(\frac{2}{\delta})$ rounds.
    \label{alg:GBB-independent:fractal:right-1}
    
    \State $\{\EstTrade{\tau, \sigma},\ \EstTrade{\tau'', \sigma},\ \EstTrade{\tau, \sigma''}\} \gets \text{``empirical means of (index-wise) one-bit feedback $\Trade[]{t}$'s by Line~\ref{alg:GBB-independent:fractal:right-1}''}$.
    \label{alg:GBB-independent:fractal:right-2}
    
    \State $\horizonal'' \gets (\horizonal + \underbrace{\EstHorizonal{[\sigma : \sigma'' - 1], \sigma}}_{\text{Line~\ref{alg:GBB-independent:fractal:EstHorizonal-EstVertical}}}) \cdot \underbrace{[\EstTrade{\tau, \sigma''} / \EstTrade{\tau, \sigma}]_{\downarrow 1}}_{\text{Line~\ref{alg:GBB-independent:fractal:right-2}}}$
    and $\vertical'' \gets (\vertical + \underbrace{\EstVertical{\tau, [\tau'' + 1 : \tau]}}_{\text{Line~\ref{alg:GBB-independent:fractal:EstHorizonal-EstVertical}}}) \cdot \underbrace{[\EstTrade{\tau'', \sigma} / \EstTrade{\tau, \sigma}]_{\downarrow 1}}_{\text{Line~\ref{alg:GBB-independent:fractal:right-2}}}$.
    \label{alg:GBB-independent:fractal:right-3}
    
    \State $\FractalElimination(\ell + 1, [\sigma'' : \tau''], \horizonal'', \vertical'')$.
    \label{alg:GBB-independent:fractal:right-4}
\end{algorithmic}
\end{algorithm}

\begin{figure}[t]
\centering
\tikzset{every picture/.style={line width = 0.75pt}} %set default line width to 0.75pt

\begin{tikzpicture}[x = 2pt, y = 2pt, scale = 1.5]
\fill[BrickRed, fill opacity = 0.2] (10, 5) -- (90, 5) -- (90, 85) -- (75, 70) -- (75, 55) -- (60, 55) -- (45, 40) -- (45, 15) -- (20, 15) -- cycle;
\fill[SkyBlue, fill opacity = 0.2] (20, 15) -- (45, 15) -- (45, 40) -- cycle;
\fill[YellowGreen, fill opacity = 0.2] (60, 55) -- (75, 55) -- (75, 70) -- cycle;
\draw[densely dotted] (20, 15) -- (90, 15);
\draw[densely dotted] (45, 40) -- (45, 5);
\draw[densely dotted] (60, 55) -- (90, 55);
\draw[densely dotted] (75, 70) -- (75, 5);

% corners
\draw (0, 0) node[anchor = 0] {$(0, 0)$};
\draw (0, 100) node[anchor = 0] {$(0, 1)$};
\draw (100, 0) node[anchor = 180] {$(1, 0)$};
\draw (100, 100) node[anchor = 180] {$(1, 1)$};

\draw (0,0) -- (0,100) -- (100,100) -- (100,0) -- cycle;
\draw (50, -3) node [below][inner sep=0.75pt] {seller};
\draw (-3, 50) node [above][inner sep = 0.75pt,rotate=90] {buyer};

% diagonal
\draw (0, 0) -- (100, 100);

% boxed line
% \foreach \y in {5, 10, 15, 20, 25, 30, 35, 40, 45, 50, 55, 60, 65, 70, 75, 80, 85, 90, 95} {
%     \draw[dotted, draw opacity = 0.5]  (\y, 0) -- (\y, 100);
%     \draw [dotted, draw opacity = 0.5]  (0, \y) -- (100, \y);
% }

% candidates
\foreach \x in {5, 15, 25, 30, 35, 40, 50, 55, 65, 70, 80, 85, 95, 100}
{\draw[black, fill = white] (\x, \x - 5) circle (2pt);}

\draw[BrickRed, ultra thick] (10, 5) -- (90, 5) -- (90, 85);
\draw[black, fill = BrickRed] (10, 5) circle (2pt);
\draw[black, fill = BrickRed] (90, 85) circle (2pt);
\draw (10, 5) node[anchor = 90] {$a_{\sigma, \sigma}$};
\draw (90, 85) node[anchor = 180] {$a_{\tau, \tau}$};
\fill[SkyBlue] (90+1, 5+1) -- (90-1, 5+1) -- (90-1, 5-1) -- cycle;
\fill[YellowGreen] (90+1, 5+1) -- (90-1, 5-1) -- (90+1, 5-1) -- cycle;
\draw[color = black] (90+1, 5+1) -- (90-1, 5+1) -- (90-1, 5-1) -- (90+1, 5-1) -- cycle;
\draw[color = black] (90+1, 5+1) -- (90-1, 5-1);
\draw (90, 5) node[anchor = 135] {$a_{\tau, \sigma}$};

\draw[black, fill = SkyBlue] (20, 15) circle (2pt);
\draw[black, fill = SkyBlue] (45, 40) circle (2pt);
\draw (20, 15) node[anchor = 90] {$a_{\sigma', \sigma'}$};
\draw (45, 40) node[anchor = 180] {$a_{\tau', \tau'}$};
\draw[color = black, fill = SkyBlue] (45+1, 5+1) -- (45-1, 5+1) -- (45-1, 5-1) -- (45+1, 5-1) -- cycle;
\draw (45, 5) node[anchor = 90] {$a_{\tau', \sigma}$};
\draw[color = black, fill = SkyBlue] (90+1, 15+1) -- (90-1, 15+1) -- (90-1, 15-1) -- (90+1, 15-1) -- cycle;
\draw (90, 15) node[anchor = 180] {$a_{\tau, \sigma'}$};

\draw[black, fill = YellowGreen] (60, 55) circle (2pt);
\draw[black, fill = YellowGreen] (75, 70) circle (2pt);
\draw (60, 55) node[anchor = 90] {$a_{\sigma'', \sigma''}$};
\draw (75, 70) node[anchor = 180] {$a_{\tau'', \tau''}$};
% \draw[color = black, fill = YellowGreen] (75+1, 55+1) -- (75-1, 55+1) -- (75-1, 55-1) -- (75+1, 55-1) -- cycle;
% \draw (75, 55) node[anchor = 135] {$a_{\tau'', \sigma''}$};
\draw[color = black, fill = YellowGreen] (75+1, 5+1) -- (75-1, 5+1) -- (75-1, 5-1) -- (75+1, 5-1) -- cycle;
\draw (75, 5) node[anchor = 90] {$a_{\tau'', \sigma}$};
\draw[color = black, fill = YellowGreen] (90+1, 55+1) -- (90-1, 55+1) -- (90-1, 55-1) -- (90+1, 55-1) -- cycle;
\draw (90, 55) node[anchor = 180] {$a_{\tau, \sigma''}$};

\draw (15-0.7071, 15-5-0.7071) -- (15+0.7071, 15-5+0.7071);
\draw (15-0.7071, 15-5+0.7071) -- (15+0.7071, 15-5-0.7071);

\draw (30-0.7071, 30-5-0.7071) -- (30+0.7071, 30-5+0.7071);
\draw (30-0.7071, 30-5+0.7071) -- (30+0.7071, 30-5-0.7071);

\draw (35-0.7071, 35-5-0.7071) -- (35+0.7071, 35-5+0.7071);
\draw (35-0.7071, 35-5+0.7071) -- (35+0.7071, 35-5-0.7071);

\draw (50-0.7071, 50-5-0.7071) -- (50+0.7071, 50-5+0.7071);
\draw (50-0.7071, 50-5+0.7071) -- (50+0.7071, 50-5-0.7071);

\draw (55-0.7071, 55-5-0.7071) -- (55+0.7071, 55-5+0.7071);
\draw (55-0.7071, 55-5+0.7071) -- (55+0.7071, 55-5-0.7071);

\draw (80-0.7071, 80-5-0.7071) -- (80+0.7071, 80-5+0.7071);
\draw (80-0.7071, 80-5+0.7071) -- (80+0.7071, 80-5-0.7071);

\draw (85-0.7071, 85-5-0.7071) -- (85+0.7071, 85-5+0.7071);
\draw (85-0.7071, 85-5+0.7071) -- (85+0.7071, 85-5-0.7071);

%legends
\draw[black, fill = white] (110, 75) circle (2pt) node[right] {$\CANDIDATE_{0} = \{a_{k, k}\}_{k \in [1 \colon K]}$};
\draw[black, fill = BrickRed] (110, 65) circle (2pt) node[right] {$\{a_{\sigma, \sigma},\ a_{\tau, \tau}\}$};
\draw[black, fill = SkyBlue] (110, 55) circle (2pt) node[right] {$\{a_{\sigma', \sigma'},\ a_{\tau', \tau'}\}$};
\draw[black, fill = YellowGreen] (110, 45) circle (2pt) node[right] {$\{a_{\sigma'', \sigma''},\ a_{\tau'', \tau''}\}$};
\draw[color = black, fill = SkyBlue] (110+1, 35+1) -- (110-1, 35+1) -- (110-1, 35-1) -- (110+1, 35-1) -- cycle;
\node at (110, 35) [right] {$\{a_{\tau, \sigma},\ a_{\tau', \sigma},\ a_{\tau, \sigma'}\}$};
\draw[color = black, fill = YellowGreen] (110+1, 25+1) -- (110-1, 25+1) -- (110-1, 25-1) -- (110+1, 25-1) -- cycle;
\node at (110, 25) [right] {$\{a_{\tau, \sigma},\ a_{\tau'', \sigma},\ a_{\tau, \sigma''}\}$};
\end{tikzpicture}
\caption{Diagram of a specific stage $\ell \in [0 : L]$ of the subroutine {\FractalElimination} (\Cref{alg:GBB-independent:fractal}).\\
Here, $\Circle$'s in general refer to candidates $\{a_{k, k}\}_{k \in [1 : K]}$, and $\otimes$'s in particular refer to candidates $\{a_{k, k}\}_{k \in [1 : K]}$ eliminated in the current stage $\ell \in [0 : L]$.\\
Also, the \textit{red} horizontal/vertical lines $a_{\sigma, \sigma}$---$a_{\tau, \sigma}$ and $a_{\tau, \sigma}$---$a_{\tau, \tau}$ refer to actions taken in Line~\ref{alg:GBB-independent:fractal:1}, and the six \textit{blue/green} $\square$'s (with $a_{\tau, \sigma}$ counted twice) refer to actions taken in Lines~\ref{alg:GBB-independent:fractal:left-1} and \ref{alg:GBB-independent:fractal:right-1}.\\
When {\FractalElimination} proceeds from the current stage $\ell \in [0 : L]$ to the next stage $\ell + 1 \in [1 : L + 1]$, the considered segment $[\sigma : \tau]$ (and its associated \textit{red} triangle) shrinks to two smaller segments $[\sigma' : \tau']$ and $[\sigma'' : \tau'']$ (and their associated \textit{blue/green} triangles).}
\label{fig:GBB-independent:1}
\end{figure}

{\FractalElimination} follows a \textit{divide-and-conquer} principle and, over $L + 1 \approx \log(K)$ stages, locates the optimal candidate $a_{\mu, \mu}$ more and more accurately. In more details:

\textit{Induction Hypothesis.}
Before a specific stage $\ell \in [0 : L]$, we have already located $a_{\mu, \mu}$ in a \textit{candidate set} $\CANDIDATE_{\ell} \subseteq [1 : K]$, and we have up to $2^{\ell}$ many disjoint \textit{segments} $[\sigma : \tau] \subseteq [1 : K]$ whose union covers $\CANDIDATE_{\ell}$.
For every considered segment $[\sigma : \tau]$, estimates $\horizonal \approx \ApxHorizonal{[1 : \sigma - 1], \sigma}$ and $\vertical \approx \ApxVertical{\tau, [\tau + 1 : K]}$ are good enough.

\textit{Base Case.}
Before the initial stage $\ell = 0$, we just consider a ``universal'' candidate set $\CANDIDATE_{0} \defeq [1 : K]$ and a single ``universal'' segment $[\sigma : \tau] = [1 : K]$. Therefore, (Phase~\ref{alg:GBB-independent:one:exploration} of {\GBBOneBit} and \Cref{footnote:ApxHorizonal-ApxVertical}) estimates $\horizonal = 0 = \ApxHorizonal{\emptyset, \sigma} = \ApxHorizonal{[1 : \sigma - 1], \sigma}$ and $\vertical = 0 = \ApxVertical{\tau, \emptyset} = \ApxVertical{\tau, [\tau + 1 : K]}$ are perfect.

\textit{Induction Step.}
In a specific stage $\ell \in [0 : L]$, we aim at locating $a_{\mu, \mu}$ more accurately $\CANDIDATE_{\ell + 1} \subseteq \CANDIDATE_{\ell}$, by leveraging one-bit feedback, and retain Induction Hypothesis for the next stage $\ell + 1 \in [1 : L + 1]$.

For every considered segment $[\sigma : \tau]$ (cf.\ the two \textit{red} $\Circle$'s in \Cref{fig:GBB-independent:1}), we \textit{can} obtain (Lines~\ref{alg:GBB-independent:fractal:1} and \ref{alg:GBB-independent:fractal:EstTrade}) the following good enough estimates for trade rates $\{\Z_{i, \sigma}\}_{i \in [\sigma : \tau]} \cup \{\Z_{\tau, j}\}_{j \in [\sigma : \tau]}$ (cf.\ the \textit{red} horizontal/vertical lines in \Cref{fig:GBB-independent:1}).
\begin{align*}
    \EstTrade{i, \sigma}
    & ~\approx~ \Z_{i, \sigma},
    && \forall i \in [\sigma : \tau], \\
    \EstTrade{\tau, j}
    & ~\approx~ \Z_{\tau, j},
    && \forall j \in [\sigma : \tau].
\end{align*}
Also, for these candidates $[\sigma : \tau]$, especially the survival ones $k \in \CANDIDATE_{\ell} \bigcap [\sigma : \tau]$, we \textit{can} obtain (Lines~\ref{alg:GBB-independent:fractal:segment} and \ref{alg:GBB-independent:fractal:EstHorizonal-EstVertical}) the following good enough estimates for terms $\ApxHorizonal{[\sigma : k], \sigma}$ and $\ApxVertical{\tau, [k : \tau]}$.
\begin{align*}
    \EstHorizonal{[\sigma : k], \sigma}
    &\textstyle ~=~ \frac{1}{K} \sum_{i \in [\sigma : k]} \EstTrade{i, \sigma}
    ~\approx~ \ApxHorizonal{[\sigma : k], \sigma},
    && \forall k \in \CANDIDATE_{\ell} \bigcap [\sigma : \tau], \\
    \EstVertical{\tau, [k : \tau]}
    &\textstyle ~=~ \frac{1}{K} \sum_{j \in [k : \tau]} \EstTrade{\tau, j}
    ~\approx~ \ApxVertical{\tau, [k : \tau]},
    && \forall k \in \CANDIDATE_{\ell} \bigcap [\sigma : \tau].
\end{align*}
Provided that (Induction Hypothesis) estimates $\horizonal \approx \ApxHorizonal{[1 : \sigma - 1], \sigma}$ and $\vertical \approx \ApxVertical{\tau, [\tau + 1 : K]}$ are also good enough, we can add either of them to the above ones, thus good enough estimates for terms $\ApxHorizonal{[1 : k], \sigma} = \Horizonal{a_{k, \sigma}} \pm K^{-1}$ and $\ApxVertical{\tau, [k : K]} = \Vertical{a_{\tau, k}} \pm K^{-1}$, $\forall k \in \CANDIDATE_{\ell} \bigcap [\sigma : \tau]$.
In regard to the independence of values $(S, B) \sim \+D_{S} \bigotimes \+D_{B}$ and that a candidate $a_{k, k}$ always induces negligible profit $\Profit(a_{k, k}) = \pm K^{-1}$, we \textit{can} obtain (\Cref{lem:GFT-independent} and Line~\ref{alg:GBB-independent:fractal:EstTrade}) the following good enough estimates $\EstGFT[\ell]{k} \approx \GFT(a_{k, k})$ for {\GainsFromTrade} from survival candidates $k \in \CANDIDATE_{\ell} \bigcap [\sigma : \tau]$ in the considered segment.
\begin{align*}
    \EstGFT[\ell]{k}
    &\textstyle ~=~ (\horizonal + \EstHorizonal{[\sigma : k], \sigma}) \cdot \frac{\EstTrade{\tau, k}}{\EstTrade{\tau, \sigma}}
    + (\vertical + \EstVertical{\tau, [k : \tau]}) \cdot \frac{\EstTrade{k, \sigma}}{\EstTrade{\tau, \sigma}} \\
    &\textstyle ~\approx~ \ApxHorizonal{[1 : K], \sigma} \cdot \frac{\Z_{\tau, k}}{\Z_{\tau, \sigma}}
    + \ApxVertical{\tau, [k : K]} \cdot \frac{\Z_{k, \sigma}}{\Z_{\tau, \sigma}} \\
    &\textstyle ~\approx~ \Horizonal{a_{k, k}}
    + \Vertical{a_{k, k}} \\
    & ~\approx~ \GFT(a_{k, k})
\end{align*}
Over all of the up to $2^{\ell}$ many disjoint segments $[\sigma : \tau]$, whose union covers $\CANDIDATE_{\ell}$ (Induction Hypothesis), we \textit{do} obtain good enough estimates $\EstGFT[\ell]{k} \approx \GFT(a_{k, k})$, for all survival candidates $k \in \CANDIDATE_{\ell}$ in the current stage $\ell \in [0 : L]$.
Then, we \textit{do} locate (Line~\ref{alg:GBB-independent:fractal:CANDIDATE}) the optimal candidate $a_{\mu, \mu}$ more accurately $\CANDIDATE_{\ell + 1} \subseteq \CANDIDATE_{\ell}$.

Moreover, we need to retain Induction Hypothesis for the next stage $\ell + 1 \in [1 : L + 1]$.
To this end, we simply follow the \textit{divide-and-conquer} principle:\\
(Lines~\ref{alg:GBB-independent:fractal:left-index} and \ref{alg:GBB-independent:fractal:right-index})
For every considered segment $[\sigma : \tau]$, find two disjoint \textit{half-segments} $[\sigma' : \tau']$ and $[\sigma'' : \tau'']$ to cover the new survival candidates $\CANDIDATE_{\ell + 1} \bigcap [\sigma : \tau]$ therein (cf.\ the two \textit{blue} $\Circle$'s and the two \textit{green} $\Circle$'s in \Cref{fig:GBB-independent:1}). Clearly, there are up to $2^{\ell + 1}$ many such half-segments in total, and their union covers the new candidate set $\CANDIDATE_{\ell + 1}$.\\
(Lines~\ref{alg:GBB-independent:fractal:left-1} to \ref{alg:GBB-independent:fractal:left-3} and \ref{alg:GBB-independent:fractal:right-1} to \ref{alg:GBB-independent:fractal:right-3})
For every half-segment $[\sigma' : \tau']$, i.e., a new segment for the next stage $\ell + 1 \in [1 : L + 1]$, obtain new good enough estimates $\horizonal' \approx \ApxHorizonal{[1 : \sigma' - 1], \sigma'}$ and $\vertical' \approx \ApxVertical{\tau', [\tau' + 1 : K]}$, in a similar manner as the above; similarly for every other half-segment $[\sigma'' : \tau'']$.\\
(Lines~\ref{alg:GBB-independent:fractal:left-4} and \ref{alg:GBB-independent:fractal:right-4})
Move on to the next stage $\ell + 1 \in [1 : L + 1]$.

In sum, initially there are $\abs{\CANDIDATE_{0}} = K = \tTheta(T^{2 / 3})$ many candidates.
After all the $L + 1 \approx \log(K)$ many stages, the ultimate candidates $\CANDIDATE_{L + 1}$ all will be good enough, compared with the optimal candidate $a_{\mu, \mu}$ or even the optimal action $a^{*}$ in the whole action space.

\begin{remark}[{\FractalElimination}]
There are another two remarkable issues:

\noindent
(i)~Despite the monotonicity $\Z_{1, j} \le \dots \le \Z_{i, j} \le \dots \le \Z_{K, j}$ and $\Z_{i, 1} \ge \dots \ge \Z_{i, j} \ge \dots \ge \Z_{i, K}$ of trade rates---we do know this---their estimates $\EstTrade{i, j}$ in Lines~\ref{alg:GBB-independent:fractal:EstTrade}, \ref{alg:GBB-independent:fractal:left-2} and \ref{alg:GBB-independent:fractal:right-2} may violate such monotonicity. Thus, we may have $\EstTrade{\tau, k} / \EstTrade{\tau, \sigma} > 1$ and/or $\EstTrade{k, \sigma} / \EstTrade{\tau, \sigma} > 1$ in Line~\ref{alg:GBB-independent:fractal:EstGFT} (etc), incurring estimation errors in estimates $\EstGFT[\ell]{k}$. (Likely, such estimation errors are severer when the ``dividend'' trade rates $\Z_{\tau, \sigma}$ themselves are small $\ll 1$.)
Actually, we use $[\EstTrade{\tau, k} / \EstTrade{\tau, \sigma}]_{\downarrow 1}$ and $[\EstTrade{k, \sigma} / \EstTrade{\tau, \sigma}]_{\downarrow 1}$ in place of $\EstTrade{\tau, k} / \EstTrade{\tau, \sigma}$ and $\EstTrade{k, \sigma} / \EstTrade{\tau, \sigma}$, where the function $[x]_{\downarrow 1} \defeq \min \{x, 1\}$.
To conclude, we actually use (Lines~\ref{alg:GBB-independent:fractal:EstGFT} and \ref{alg:GBB-independent:fractal:CANDIDATE}) the following estimates $\EstGFT[\ell]{k}$ for {\GainsFromTrade} $\GFT(a_{k, k})$'s from survival candidates $k \in \CANDIDATE_{\ell} \bigcap [\sigma : \tau]$ in the considered segment.
\begin{align*}
    \textstyle
    \EstGFT[\ell]{k} ~=~ (\horizonal + \EstHorizonal{[\sigma : k], \sigma}) \cdot \big[\frac{\EstTrade{\tau, k}}{\EstTrade{\tau, \sigma}}\big]_{\downarrow 1}
    + (\vertical + \EstVertical{\tau, [k : \tau]}) \cdot \big[\frac{\EstTrade{k, \sigma}}{\EstTrade{\tau, \sigma}}\big]_{\downarrow 1}.
\end{align*}
The circumstances of estimates $\horizonal'$, $\vertical'$, $\horizonal''$, and $\vertical''$ in Lines~\ref{alg:GBB-independent:fractal:left-3} and \ref{alg:GBB-independent:fractal:right-3} are analogous.

\noindent
(ii)~{\FractalElimination} intrinsically follows the \textit{divide-and-conquer} principle, with regard to the current/new segments $[\sigma : \tau] \supseteq [\sigma' : \tau'], [\sigma'' : \tau'']$ especially.
Here, to better reflect this and for ease of presentation, \Cref{alg:GBB-independent:fractal} is implemented recursively, in a \textit{depth-first-search} manner.
However, (Lines~\ref{alg:GBB-independent:fractal:segment}, \ref{alg:GBB-independent:fractal:EstGFT} and \ref{alg:GBB-independent:fractal:CANDIDATE}) a single stage $\ell \in [0 : L]$ needs to address up to $2^{\ell}$ many disjoint segments $[\sigma : \tau]$, which $\bigcup_{[\sigma : \tau]} (\CANDIDATE_{\ell} \bigcap [\sigma : \tau]) = \CANDIDATE_{\ell}$ all are necessary to determine the next candidate set $\CANDIDATE_{\ell + 1}$;
thus, indeed, {\FractalElimination} must proceed stage by stage $\ell = 0, 1, 2, \dots, L$ strictly.
(In this regard, {\FractalElimination} may be better implemented in a \textit{breadth-first-search} manner, which however will further complicate \Cref{alg:GBB-independent:fractal}.)
\end{remark}

\subsection{Performance Analysis of {\FractalElimination}}

In this part, we will disclose the performance guarantees of the subroutine {\FractalElimination} through a sequence of lemmas; the conclusions are summarized into \Cref{cor:GBB-independent:success-probability} and \Cref{lem:GBB-independent:exploration}.

To begin with, the following \Cref{lem:GBB-independent:estimate} establishes standard concentration bounds for estimates $\EstTrade{i, j}$, $\EstHorizonal{[\sigma : k], \sigma}$, and $\EstVertical{\tau, [k : \tau]}$ in Lines~\ref{alg:GBB-independent:fractal:EstTrade}, \ref{alg:GBB-independent:fractal:EstHorizonal-EstVertical}, \ref{alg:GBB-independent:fractal:left-2} and \ref{alg:GBB-independent:fractal:right-2}.

\begin{lemma}[{{\FractalElimination}; Estimates in Lines~\ref{alg:GBB-independent:fractal:EstTrade}, \ref{alg:GBB-independent:fractal:EstHorizonal-EstVertical}, \ref{alg:GBB-independent:fractal:left-2} and \ref{alg:GBB-independent:fractal:right-2}}]
\label{lem:GBB-independent:estimate}
\begin{flushleft}
Throughout the whole recursion of the subroutine {\FractalElimination} (invoked in Phase~\ref{alg:GBB-independent:one:exploration} of {\GBBOneBit}), the following hold:
\begin{enumerate}
    \item \label{lem:GBB-independent:estimate:1}
    $\EstTrade{i, j} = \Z_{i, j} \pm 2^{-(\ell + 3) / 2} K^{-1 / 2}$ with probability $1 - \delta$, for a single estimate $\EstTrade{i, j}$ in Line~\ref{alg:GBB-independent:fractal:EstTrade}.
    
    \item \label{lem:GBB-independent:estimate:2}
    $\EstHorizonal{[\sigma : k], \sigma} = \ApxHorizonal{[\sigma : k], \sigma} \pm \frac{1}{K}$ with probability $1 - \delta$, for a single estimate $\EstHorizonal{[\sigma : k], \sigma}$ in Line~\ref{alg:GBB-independent:fractal:EstHorizonal-EstVertical}.
    $\EstVertical{\tau, [k : \tau]} = \ApxVertical{\tau, [k : \tau]} \pm \frac{1}{K}$ with probability $1 - \delta$, for a single estimate $\EstVertical{\tau, [k : \tau]}$ in Line~\ref{alg:GBB-independent:fractal:EstHorizonal-EstVertical}.
    
    \item \label{lem:GBB-independent:estimate:3}
    $\EstTrade{i, j} = \Z_{i, j} \pm K^{-1}$ with probability $1 - \delta$, for a single estimate $\EstTrade{i, j}$ in Lines~\ref{alg:GBB-independent:fractal:left-2} and \ref{alg:GBB-independent:fractal:right-2}.
\end{enumerate}
\end{flushleft}
\end{lemma}

\begin{proof}
By construction, $\EstTrade{i, j}$ is the empirical mean of i.i.d.\ one-bit feedback $\Trade[i, j]{t} \in \{0, 1\}$ and is an unbiased estimate of $\Z_{i, j}$ (Lines~\ref{alg:GBB-independent:fractal:EstTrade}, \ref{alg:GBB-independent:fractal:left-2} and \ref{alg:GBB-independent:fractal:right-2} and \Cref{eq:GBB-independent:Trade}), so Hoeffding's inequality (\Cref{fact:Hoeffding}) gives \Cref{lem:GBB-independent:estimate:1,lem:GBB-independent:estimate:3}, using either $r = 2^{-(\ell + 3) / 2} K^{-1 / 2}$ and $n = 2^{\ell + 2}K\ln(\frac{2}{\delta})$ (\Cref{lem:GBB-independent:estimate:1}) or $r = K^{-1}$ and $n = \frac{1}{2}K^{2} \ln(\frac{2}{\delta})$ (\Cref{lem:GBB-independent:estimate:3}).

Moreover, $\EstHorizonal{[\sigma : k], \sigma} = \frac{1}{K} \sum_{i \in [\sigma : k]} \EstTrade{i, \sigma}$ is an unbiased estimate of $\ApxHorizonal{[\sigma : k], \sigma} = \frac{1}{K} \sum_{i \in [\sigma : k]} \Z_{i, \sigma}$ and, by implication from the above, is the empirical mean of $n = 2^{\ell + 2}K\ln(\frac{2}{\delta})$ many i.i.d.\ random variables of the form $\frac{1}{K} \sum_{i \in [\sigma : k]} \Trade[i, \sigma]{t}$.
Every such random variable is $[0, 1]$-bounded (given that $\frac{1}{K} \cdot \abs{[\sigma : k]} \cdot 1 \le 1$) and has variance at most $K^{-1}$ (given that $\frac{1}{K^{2}} \cdot \abs{[\sigma : k]} \cdot 1 \le K^{-1}$); similarly for $\EstVertical{\tau, [k : \tau]}$.
Thus, \Cref{lem:GBB-independent:estimate:2} follows from Bernstein inequality (\Cref{fact:Bernstein}), using $s^{2} = K^{-1}$, $r = K^{-1}$, and $n = 2^{\ell + 2}K\ln(\frac{2}{\delta}) \ge \frac{2 \cdot (s^{2} + r / 3)}{r^{2}}\ln(\frac{2}{\delta})$.

This finishes the proof of \Cref{lem:GBB-independent:estimate}.
\end{proof}

Moreover, the following \Cref{lem:GBB-independent:number-of-estimates} counts, throughout the whole recursion, how many estimates $\EstTrade{i, j}$, $\EstHorizonal{[\sigma : k], \sigma}$, and $\EstVertical{\tau, [k : \tau]}$ we will encounter in Lines~\ref{alg:GBB-independent:fractal:EstTrade}, \ref{alg:GBB-independent:fractal:EstHorizonal-EstVertical}, \ref{alg:GBB-independent:fractal:left-2} and \ref{alg:GBB-independent:fractal:right-2}.

\begin{lemma}[{{\FractalElimination}; The Total Number of Estimates}]
\label{lem:GBB-independent:number-of-estimates}
\begin{flushleft}
Throughout the whole recursion of the subroutine {\FractalElimination} (invoked in Phase~\ref{alg:GBB-independent:one:exploration} of {\GBBOneBit}), the total number of estimates $\EstTrade{i, j}$, $\EstHorizonal{[\sigma : k], \sigma}$, and $\EstVertical{\tau, [k : \tau]}$ in Lines~\ref{alg:GBB-independent:fractal:EstTrade}, \ref{alg:GBB-independent:fractal:EstHorizonal-EstVertical}, \ref{alg:GBB-independent:fractal:left-2} and \ref{alg:GBB-independent:fractal:right-2}
is at most $(\frac{1}{6} \pm o(1))T^{1 / 3}\log^{1 / 3}(T)$.
\end{flushleft}
\end{lemma}

\begin{proof}
Omitting the recursion in Lines~\ref{alg:GBB-independent:fractal:left-4} and \ref{alg:GBB-independent:fractal:right-4}, a single invocation of the subroutine {\FractalElimination} with generic input $(\ell, [\sigma : \tau], \horizonal, \vertical)$, $\forall \ell \in [0 : L]$, $\forall [\sigma : \tau] \subseteq [1 : K]$, involves\\
(i)~$2 \cdot \abs{[\sigma : \tau]}$ many estimates $\{\EstTrade{i, \sigma}\}_{i \in [\sigma : \tau]} \cup \{\EstTrade{\tau, j}\}_{j \in [\sigma : \tau]}$ in Line~\ref{alg:GBB-independent:fractal:EstTrade},\\
(ii)~$2 \cdot \abs{[\sigma : \tau]}$ many estimates $\{\EstHorizonal{[\sigma : k], \sigma},\ \EstVertical{\tau, [k : \tau]}\}_{k \in [\sigma : \tau]}$ in Line~\ref{alg:GBB-independent:fractal:EstHorizonal-EstVertical}, and\\
(iii)~six estimates $\{\EstTrade{\tau, \sigma},\ \EstTrade{\tau', \sigma},\ \EstTrade{\tau, \sigma'}\} \cup \{\EstTrade{\tau, \sigma},\ \EstTrade{\tau'', \sigma},\ \EstTrade{\tau, \sigma''}\}$ in Lines~\ref{alg:GBB-independent:fractal:left-2} and \ref{alg:GBB-independent:fractal:right-2}.\\
In a specific stage $\ell \in [0 : L]$, by the divide-and-conquer essence of the subroutine {\FractalElimination} (Lines~\ref{alg:GBB-independent:fractal:left-4} and \ref{alg:GBB-independent:fractal:right-4} and \Cref{fig:GBB-independent:1}),
all invocations have \textit{disjoint} input segments $[\sigma : \tau] \subseteq [1 : K]$, and
there are at most $2^{\ell}$ many different invocations.
Accordingly, the total number of estimates throughout the whole recursion is at most
\begin{align*}
    \textstyle
    \underbrace{(L + 1) \cdot 2K}_{\text{Line~\ref{alg:GBB-independent:fractal:EstTrade}}}
    + \underbrace{(L + 1) \cdot 2K}_{\text{Line~\ref{alg:GBB-independent:fractal:EstHorizonal-EstVertical}}}
    + \underbrace{\sum_{\ell \in [0 : L]} 2^{\ell} \cdot 6}_{\text{Lines~\ref{alg:GBB-independent:fractal:left-2} and \ref{alg:GBB-independent:fractal:right-2}}}
    &\textstyle ~\le~ (L + 1) \cdot 4K + 13 \cdot 2^{L} \\
    &\textstyle ~=~ T^{1 / 3}\log^{1 / 3}(T) \cdot \Big(\frac{1}{6} + \frac{1}{2\log(T)} + \frac{13}{\log^{1 / 3}(T)}\Big) \\
    &\textstyle ~=~ (\frac{1}{6} \pm o(1))T^{1 / 3}\log^{1 / 3}(T).
\end{align*}
Here the second step substitutes $K = \frac{1}{8}T^{1 / 3}\log^{-2 / 3}(T)$ and $L = \frac{1}{3}\log(T)$.

This finishes the proof of \Cref{lem:GBB-independent:number-of-estimates}.
\end{proof}

For the sake of completeness, the following \Cref{lem:GBB-independent:number-of-rounds} shows that the whole recursion will only take a sublinear number of rounds $o(T)$.

\begin{lemma}[{{\FractalElimination}; The Total Number of Rounds}]
\label{lem:GBB-independent:number-of-rounds}
\begin{flushleft}
Throughout the whole recursion of the subroutine {\FractalElimination} (invoked in Phase~\ref{alg:GBB-independent:one:exploration} of {\GBBOneBit}), the total number of rounds taken in Lines~\ref{alg:GBB-independent:fractal:1}, \ref{alg:GBB-independent:fractal:left-1} and \ref{alg:GBB-independent:fractal:right-1} (for estimates in Lines~\ref{alg:GBB-independent:fractal:EstTrade}, \ref{alg:GBB-independent:fractal:EstHorizonal-EstVertical}, \ref{alg:GBB-independent:fractal:left-2} and \ref{alg:GBB-independent:fractal:right-2}) is at most $T\log^{-1 / 3}(T) = o(T)$.
\end{flushleft}
\end{lemma}

\begin{proof}
Omitting the recursion in Lines~\ref{alg:GBB-independent:fractal:left-4} and \ref{alg:GBB-independent:fractal:right-4}, a single invocation of the subroutine {\FractalElimination} with generic input $(\ell, [\sigma : \tau], \horizonal, \vertical)$, for $\ell \in [0 : L]$ and $[\sigma : \tau] \subseteq [1 : K]$, takes\\
(i)~$2 \cdot \abs{[\sigma : \tau]}$ many actions $\{a_{i, \sigma}\}_{i \in [\sigma : \tau]} \cup \{a_{\tau, j}\}_{j \in [\sigma : \tau]}$ each for $2^{\ell + 2}K\ln(\frac{2}{\delta})$ rounds in Line~\ref{alg:GBB-independent:fractal:1}, and\\
(ii)~six actions $\{a_{\tau, \sigma},\ a_{\tau', \sigma},\ a_{\tau, \sigma'}\} \cup \{a_{\tau, \sigma},\ a_{\tau'', \sigma},\ a_{\tau, \sigma''}\}$ each for $\frac{1}{2}K^{2}\ln(\frac{2}{\delta})$ rounds in Lines~\ref{alg:GBB-independent:fractal:left-1} and \ref{alg:GBB-independent:fractal:right-1}.\\
Given these and reusing the arguments in the proof of \Cref{lem:GBB-independent:number-of-estimates}, the total number of rounds throughout the whole recursion is at most
\begin{align*}
    \textstyle
    \sum_{\ell \in [0 : L]} \big(\underbrace{2K \cdot 2^{\ell + 2}K\ln(\frac{2}{\delta})}_{\text{Line~\ref{alg:GBB-independent:fractal:1}}}
    + \underbrace{2^{\ell} \cdot 6 \cdot \frac{1}{2}K^{2}\ln(\frac{2}{\delta})}_{\text{Lines~\ref{alg:GBB-independent:fractal:left-1} and \ref{alg:GBB-independent:fractal:right-1}}}\big)
    &\textstyle ~=~ \sum_{\ell \in [0 : L]} 11 \cdot 2^{\ell} K^{2}\ln(\frac{2}{\delta}) \\
    &\textstyle ~\le~ 22 \cdot 2^{L}K^{2}\ln(\frac{2}{\delta}) \\
    &\textstyle ~=~ (\frac{11}{24}\ln(2) \pm o(1))T\log^{-1 / 3}(T) \\
    &\textstyle ~\le~ T\log^{-1 / 3}(T).
\end{align*}
Here the third step substitutes $K = \frac{1}{8}T^{1 / 3}\log^{-2 / 3}(T)$, $L = \frac{1}{3}\log(T)$, and $\ln(\frac{2}{\delta}) = \ln{(2T^{4 / 3}\log^{1 / 3}(T))} = (\frac{4}{3}\ln(2) \pm o(1)) \log(T)$.
And the last step uses $\frac{11}{24}\ln(2) \pm o(1) \approx 0.3177 \pm o(1) < 1$, which holds for any large enough $T \gg 1$.

This finishes the proof of \Cref{lem:GBB-independent:number-of-rounds}.
\end{proof}

In the rest of \Cref{sec:GBB-independent}, we would call the subroutine {\FractalElimination} \textit{``successful''} if all estimates $\EstTrade{i, j}$, $\EstHorizonal{[\sigma : k], \sigma}$, and $\EstVertical{\tau, [k : \tau]}$ throughout its whole recursion satisfy the bounds given in \Cref{lem:GBB-independent:estimate} (namely $\EstTrade{i, j} = \Z_{i, j} \pm 2^{-(\ell + 3) / 2} K^{-1 / 2}$ etc), or \textit{``failed''} otherwise.

The following \Cref{cor:GBB-independent:success-probability} directly follows from a combination of \Cref{lem:GBB-independent:estimate,lem:GBB-independent:number-of-estimates}.

\begin{corollary}[{{\FractalElimination}; The Success Probability}]
\label{cor:GBB-independent:success-probability}
\begin{flushleft}
The whole recursion of the subroutine {\FractalElimination} (invoked in Phase~\ref{alg:GBB-independent:one:exploration} of {\GBBOneBit}) succeeds with probability $1 - T^{-1}$.
\end{flushleft}
\end{corollary}

\begin{proof}
A single estimate $\EstTrade{i, j}$, $\EstHorizonal{[\sigma : k], \sigma}$, or $\EstVertical{\tau, [k : \tau]}$ violates the bounds given in \Cref{lem:GBB-independent:estimate} with probability at most $\delta = T^{-4 / 3}\log^{-1 / 3}(T)$, and we have at most $(\frac{1}{6} \pm o(1))T^{1 / 3}\log^{1 / 3}(T)$ many such estimates throughout the whole recursion (\Cref{lem:GBB-independent:number-of-estimates}).
Thus the failure probability is at most $(\frac{1}{6} \pm o(1))T^{-1} \le T^{-1}$, where the last step holds for any large enough $T \gg 1$.
This finishes the proof of \Cref{cor:GBB-independent:success-probability}.
\end{proof}

The failure probability $\le T^{-1}$ of the subroutine {\FractalElimination} is negligibly small. In the rest of \Cref{sec:GBB-independent}, we would concentrate more on the ``success'' case.

First of all, the following \Cref{lem:GBB-independent:horizonal-vertical} establishes useful concentration bounds, in an inductive manner, for estimates $\horizonal'$, $\vertical'$, $\horizonal''$, and $\vertical''$ in Lines~\ref{alg:GBB-independent:fractal:left-3} and \ref{alg:GBB-independent:fractal:right-3}, namely part of input to (Lines~\ref{alg:GBB-independent:fractal:left-4} and \ref{alg:GBB-independent:fractal:right-4}) the recursion of {\FractalElimination} in the next stage $\ell + 1 \in [1 : L + 1]$.

\begin{lemma}[{{\FractalElimination}; Estimates in Lines~\ref{alg:GBB-independent:fractal:left-3} and \ref{alg:GBB-independent:fractal:right-3}}]
\label{lem:GBB-independent:horizonal-vertical}
\begin{flushleft}
Conditional on ``success'', throughout the whole recursion of the subroutine {\FractalElimination} (invoked in Phase~\ref{alg:GBB-independent:one:exploration} of {\GBBOneBit}), all estimates $\horizonal'$, $\vertical'$, $\horizonal''$, and $\vertical''$ in Lines~\ref{alg:GBB-independent:fractal:left-3} and \ref{alg:GBB-independent:fractal:right-3} satisfy the following, almost surely.
\begin{align}
    \abs{\horizonal'
    - \ApxHorizonal{[1 : \sigma' - 1], \sigma'}}
    & ~\le~ \abs{\horizonal - \ApxHorizonal{[1 : \sigma - 1], \sigma}} + 3K^{-1},
    \label{lem:GBB-independent:horizonal-vertical:1} \\
    \abs{\vertical'
    - \ApxVertical{\tau', [\tau' + 1 : K]}}
    & ~\le~ \abs{\vertical - \ApxVertical{\tau, [\tau + 1 : K]}} + 3K^{-1},
    \label{lem:GBB-independent:horizonal-vertical:2} \\
    \abs{\horizonal''
    - \ApxHorizonal{[1 : \sigma'' - 1], \sigma''}}
    & ~\le~ \abs{\horizonal - \ApxHorizonal{[1 : \sigma - 1], \sigma}} + 3K^{-1},
    \label{lem:GBB-independent:horizonal-vertical:3} \\
    \abs{\vertical''
    - \ApxVertical{\tau'', [\tau'' + 1 : K]}}
    & ~\le~ \abs{\vertical - \ApxVertical{\tau, [\tau + 1 : K]}} + 3K^{-1}.
    \label{lem:GBB-independent:horizonal-vertical:4}
\end{align}
\end{flushleft}
\end{lemma}

\begin{proof}
We would only prove \Cref{lem:GBB-independent:horizonal-vertical:1}; similarly, \Cref{lem:GBB-independent:horizonal-vertical:2,lem:GBB-independent:horizonal-vertical:3,lem:GBB-independent:horizonal-vertical:4} follow from symmetric arguments.
Note that $\horizonal = \ApxHorizonal{[1 : \sigma - 1], \sigma} \pm \abs{\horizonal - \ApxHorizonal{[1 : \sigma - 1], \sigma}}$ and $[\EstTrade{\tau', \sigma'} / \EstTrade{\tau', \sigma}]_{\downarrow 1} \in [0, 1]$, almost surely.
Conditional on ``success'', we have $\EstHorizonal{[\sigma : \sigma' - 1], \sigma} = \ApxHorizonal{[\sigma : \sigma' - 1], \sigma} \pm K^{-1}$ (\Cref{lem:GBB-independent:estimate:2} of \Cref{lem:GBB-independent:estimate}). Thus, we have
\begin{align}
    \horizonal'
    & ~=~ (\horizonal + \EstHorizonal{[\sigma : \sigma' - 1], \sigma}) \cdot [\EstTrade{\tau', \sigma'} / \EstTrade{\tau', \sigma}]_{\downarrow 1}
    \notag \\
    & ~=~ (\ApxHorizonal{[1 : \sigma - 1], \sigma}
    + \ApxHorizonal{[\sigma : \sigma' - 1], \sigma}) \cdot [\EstTrade{\tau', \sigma'} / \EstTrade{\tau', \sigma}]_{\downarrow 1}
    ~\pm~ \abs{\horizonal - \ApxHorizonal{[1 : \sigma - 1], \sigma}} ~\pm~ K^{-1}
    \notag \\
    & ~=~ \ApxHorizonal{[1 : \sigma' - 1], \sigma} \cdot [\EstTrade{\tau', \sigma'} / \EstTrade{\tau', \sigma}]_{\downarrow 1}
    ~\pm~ \abs{\horizonal - \ApxHorizonal{[1 : \sigma - 1], \sigma}} ~\pm~ K^{-1}.
    \label{eq:GBB-one-stage-error:1}
\end{align}
Here the last step uses the additivity $\ApxHorizonal{[1 : \sigma - 1], \sigma} + \ApxHorizonal{[\sigma : \sigma' - 1], \sigma} = \ApxHorizonal{[1 : \sigma' - 1], \sigma}$ (\Cref{eq:GBB-independent:ApxHorizonal}).
Below, we would reason about \Cref{eq:GBB-one-stage-error:1} in either case $\Z_{\tau', \sigma} < 2K^{-1}$ or $\Z_{\tau', \sigma} \ge 2K^{-1}$ separately.

\vspace{.1in}
\noindent
{\bf Case~1: $\Z_{\tau', \sigma} < 2K^{-1}$.}
By \Cref{eq:GBB-independent:Trade,eq:GBB-independent:ApxHorizonal} and Line~\ref{alg:GBB-independent:fractal:left-index}, we have $\ApxHorizonal{[1 : \sigma' - 1], \sigma'}
\le \ApxHorizonal{[1 : \sigma' - 1], \sigma}
\le 2K^{-1}
\impliedby
\frac{1}{K} \sum_{i \in [1 : \sigma' - 1]} \Z_{i, \sigma'}
\le \frac{1}{K} \sum_{i \in [1 : \sigma' - 1]} \Z_{i, \sigma}
\le 2K^{-1}$,
because $\Z_{i, \sigma'} \le \Z_{i, \sigma}$, $\forall i \in [1 : K] \impliedby \sigma' \ge \sigma$ and $\Z_{i, \sigma} \le \Z_{\tau', \sigma} < 2K^{-1}$, $\forall i \in [1 : \tau'] \impliedby \tau' \ge \sigma'$.

In combination with \Cref{eq:GBB-one-stage-error:1}, we deduce \Cref{lem:GBB-independent:horizonal-vertical:1} as follows.
\begin{align*}
    \horizonal'
    & ~\ge~ -\abs{\horizonal - \ApxHorizonal{[1 : \sigma - 1], \sigma}} ~-~ K^{-1} \\
    & ~\ge~ \ApxHorizonal{[1 : \sigma' - 1], \sigma'}
    ~-~ 2K^{-1} ~-~ \abs{\horizonal - \ApxHorizonal{[1 : \sigma - 1], \sigma}} ~-~ K^{-1}, \\
    \horizonal'
    & ~\le~ 2K^{-1} ~+~ \abs{\horizonal - \ApxHorizonal{[1 : \sigma - 1], \sigma}} ~+~ K^{-1} \\
    & ~\le~ \ApxHorizonal{[1 : \sigma' - 1], \sigma'}
    ~+~ 2K^{-1} ~+~ \abs{\horizonal - \ApxHorizonal{[1 : \sigma - 1], \sigma}} ~+~ K^{-1}.
\end{align*}

\noindent
{\bf Case~2: $\Z_{\tau', \sigma} \ge 2K^{-1}$.}
We have $[\EstTrade{\tau', \sigma'} / \EstTrade{\tau', \sigma}]_{\downarrow 1}
= [(\Z_{\tau', \sigma'} \pm K^{-1}) / (\Z_{\tau', \sigma} \mp K^{-1})]_{\downarrow 1}
= (\Z_{\tau', \sigma'} \pm 2K^{-1}) / \Z_{\tau', \sigma}$, where the first step uses \Cref{lem:GBB-independent:estimate:3} of \Cref{lem:GBB-independent:estimate}, and the second step uses the facts $\min\{\frac{x - y}{1 + y}, 1\} \le x - 2y$ and $\min\{\frac{x + y}{1 - y}, 1\} \le x + 2y$, $\forall (x, y) \in [0, 1] \times [0, \frac{1}{2}]$.

In combination with \Cref{eq:GBB-one-stage-error:1}, we deduce \Cref{lem:GBB-independent:horizonal-vertical:1} as follows.
\begin{align*}
    \horizonal'
    & ~=~ \ApxHorizonal{[1 : \sigma' - 1], \sigma} \cdot (\Z_{\tau', \sigma'} \pm 2K^{-1}) / \Z_{\tau', \sigma}
    ~\pm~ \abs{\horizonal - \ApxHorizonal{[1 : \sigma - 1], \sigma}} ~\pm~ K^{-1} \\
    & ~=~ \ApxHorizonal{[1 : \sigma' - 1], \sigma'} ~\pm~ \ApxHorizonal{[1 : \sigma' - 1], \sigma} \cdot 2K^{-1} / \Z_{\tau', \sigma}
    ~\pm~ \abs{\horizonal - \ApxHorizonal{[1 : \sigma - 1], \sigma}} ~\pm~ K^{-1} \\
    & ~=~ \ApxHorizonal{[1 : \sigma' - 1], \sigma'} ~\pm~ 2K^{-1} ~\pm~ \abs{\horizonal - \ApxHorizonal{[1 : \sigma - 1], \sigma}} ~\pm~ K^{-1}.
\end{align*}
Here the second step uses $\ApxHorizonal{[1 : \sigma' - 1], \sigma} \cdot \Z_{\tau', \sigma'} / \Z_{\tau', \sigma} = \ApxHorizonal{[1 : \sigma' - 1], \sigma'}$ (\Cref{eq:GBB-independent:Trade,eq:GBB-independent:ApxHorizonal}).
And the last step uses $\ApxHorizonal{[1 : \sigma' - 1], \sigma} \cdot 2K^{-1} / \Z_{\tau', \sigma}
= \frac{1}{K} \sum_{i \in [1 : \sigma' - 1]} \frac{\+D_{S}(i / K)}{\+D_{S}(\tau' / K)} \cdot 2K^{-1}
\le 2K^{-1}$ (\Cref{eq:GBB-independent:Trade,eq:GBB-independent:ApxHorizonal} and Line~\ref{alg:GBB-independent:fractal:left-index}), given that $\+D_{S}(i / K) \le \+D_{S}(\tau' / K)$, $\forall i \in [1 : \sigma' - 1] \impliedby \sigma' \le \tau'$.

Combining both cases gives \Cref{lem:GBB-independent:horizonal-vertical:1}. This finishes the proof of \Cref{lem:GBB-independent:horizonal-vertical}.
\end{proof}

Further, the following \Cref{lem:GBB-independent:EstGFT} shows useful concentration bounds for estimates $\EstGFT[\ell]{k}$ in Line~\ref{alg:GBB-independent:fractal:EstGFT}.

\begin{lemma}[{{\FractalElimination}; Estimates $\EstGFT[\ell]{k}$ in Line~\ref{alg:GBB-independent:fractal:EstGFT}}]
\label{lem:GBB-independent:EstGFT}
\begin{flushleft}
Conditional on ``success'', throughout the whole recursion of the subroutine {\FractalElimination} (invoked in Phase~\ref{alg:GBB-independent:one:exploration} of {\GBBOneBit}), all estimates $\EstGFT[\ell]{k}$ in Line~\ref{alg:GBB-independent:fractal:EstGFT} satisfy the following, almost surely.
\begin{align*}
    & \abs{\EstGFT[\ell]{k} - \GFT(a_{k, k})}
    ~\le~ \gamma_{\ell + 1}
    ~=~ 2^{-(\ell + 1) / 2} K^{-1 / 2} + (6\ell + 5) K^{-1},
    && \forall \ell \in [0 : K],\ \forall k \in \CANDIDATE_{\ell}.
\end{align*}
\end{flushleft}
\end{lemma}

\begin{proof}
Based on \Cref{lem:GFT-independent}, we can formulate $\GFT(a_{k, k}) = \GFT(\frac{k}{K}, \frac{k - 1}{K})$ as follows.
\begin{align*}
    \GFT(a_{k, k})
    &\textstyle ~=~ \Horizonal{a_{k, k}}
    + \Vertical{a_{k, k}}
    + \Profit(a_{k, k}), \\
    \Horizonal{a_{k, k}}
    &\textstyle ~=~ \int_{0}^{k / K} \+D_{S}(x) \dd x \cdot (1 - \+D_{B}(\frac{k - 1}{K})) \\
    \Vertical{a_{k, k}}
    &\textstyle ~=~ \+D_{S}(\frac{k}{K}) \cdot \int_{(k - 1) / K}^{1} (1 - \+D_{B}(y)) \dd y \\
    \Profit(a_{k, k})
    &\textstyle ~=~ -\frac{1}{K} \cdot \+D_{S}(\frac{k}{K}) \cdot (1 - \+D_{B}(\frac{k - 1}{K})).
\end{align*}
Further, by induction of \Cref{lem:GBB-independent:horizonal-vertical} on stages $0, 1, 2, \dots, \ell \in [0 : L]$,\footnote{In the base case $\ell = 0$, we have $\horizonal = 0$, $\vertical = 0$, $\sigma = 1$, and $\tau = K$, so both ``estimates'' $\horizonal = \ApxHorizonal{[1 : \sigma - 1], \sigma}$ and $\vertical = \ApxVertical{\tau, [\tau + 1 : K]}$ are perfect (Phase~\ref{alg:GBB-independent:one:exploration} of {\GBBOneBit} and \Cref{footnote:ApxHorizonal-ApxVertical}).
And to validate the induction on stages $\ell \in [0 : L]$, rigorously speaking, we shall note the shrinkage $[1 : K] = C_{0} \supseteq C_{1} \supseteq \dots \supseteq C_{L + 1}$ (Line~\ref{alg:GBB-independent:fractal:CANDIDATE}) and the inclusion $C_{\ell} \subseteq$ ``the union of all the (disjoint) input segments $[\sigma : \tau] \subseteq [1 : K]$ for all the stage-$\ell$ invocations'', $\forall \ell \in [0 : L]$, given the divide-and-conquer essence of the subroutine {\FractalElimination} (Lines~\ref{alg:GBB-independent:fractal:left-4} and \ref{alg:GBB-independent:fractal:right-4} and \Cref{fig:GBB-independent:1}).} we have $\horizonal = \ApxHorizonal{[1 : \sigma - 1], \sigma} \pm 3\ell K^{-1}$ and $\vertical = \ApxVertical{\tau, [\tau + 1 : K]} \pm 3\ell K^{-1}$, $\forall \ell \in [0 : K]$.
In addition, we know from \Cref{lem:GBB-independent:estimate:1,lem:GBB-independent:estimate:2} of \Cref{lem:GBB-independent:estimate} that $\EstTrade{i, j} = \Z_{i, j} \pm 2^{-(\ell + 3) / 2} K^{-1 / 2}$, $\forall (i, j) \in \{(\tau, \sigma),\ (\tau, k),\ (k, \sigma)\}$, that $\EstHorizonal{[\sigma : k], \sigma} = \ApxHorizonal{[\sigma : k], \sigma} \pm K^{-1}$, and that $\EstVertical{\tau, [k : \tau]} = \ApxVertical{\tau, [k : \tau]} \pm K^{-1}$.
Given these, we deduce that
\begin{align*}
    \EstGFT[\ell]{k}
    & ~=~ (\horizonal + \EstHorizonal{[\sigma : k], \sigma}) \cdot [\EstTrade{\tau, k} / \EstTrade{\tau, \sigma}]_{\downarrow 1}
    ~+~ (\vertical + \EstVertical{\tau, [k : \tau]}) \cdot [\EstTrade{k, \sigma} / \EstTrade{\tau, \sigma}]_{\downarrow 1} \\
    & ~=~ (\ApxHorizonal{[1 : K], \sigma}) \cdot [\EstTrade{\tau, k} / \EstTrade{\tau, \sigma}]_{\downarrow 1}
    ~+~ (\ApxVertical{\tau, [k : K]}) \cdot [\EstTrade{k, \sigma} / \EstTrade{\tau, \sigma}]_{\downarrow 1}
    ~\pm~ 6\ell K^{-1} ~\pm~ 2K^{-1} \\
    & ~=~ \ApxHorizonal{[1 : K], k}
    ~+~ \ApxVertical{k, [k : K]}
    ~\pm~ 2^{-(\ell + 1) / 2} K^{-1 / 2} ~\pm~ 6\ell K^{-1} ~\pm~ 2K^{-1}.
\end{align*}
Here the first step uses the defining formula of $\EstGFT[\ell]{k}$ given in Line~\ref{alg:GBB-independent:fractal:EstGFT}.
The second step uses arguments symmetric to those for \Cref{eq:GBB-one-stage-error:1}.
And the last step uses arguments symmetric to those for Case~1 and Case~2 in the proof of \Cref{lem:GBB-independent:horizonal-vertical}.

By comparing both equations above with the claim of \Cref{lem:GBB-independent:EstGFT}, it remains to show the following.
\begin{align*}
    \ApxHorizonal{[1 : K], k}
    &\textstyle ~\overset{\eqref{eq:GBB-independent:ApxHorizonal}}{=}~ \big(\frac{1}{K} \sum_{i \in [1 : K]} \+D_{S}(\frac{i}{K})\big) \cdot (1 - \+D_{B}(\frac{k - 1}{K}))
    ~=~ \Horizonal{a_{k, k}} ~\pm~ K^{-1}, \\
    \ApxVertical{k, [k : K]}
    &\textstyle ~\overset{\eqref{eq:GBB-independent:ApxVertical}}{=}~ \+D_{S}(\frac{k}{K}) \cdot \big(\frac{1}{K} \sum_{j \in [1 : K]} (1 - \+D_{B}(\frac{j - 1}{K}))\big)
    ~=~ \Vertical{a_{k, k}} ~\pm~ K^{-1}, \\
    \Profit(a_{k, k})
    &\textstyle ~\overset{\phantom{\eqref{eq:GBB-independent:ApxVertical}}}{=}~ -\frac{1}{K} \cdot \+D_{S}(\frac{k}{K}) \cdot (1 - \+D_{B}(\frac{k - 1}{K}))
    ~=~ \pm K^{-1}.
\end{align*}
It is easy to see all these equations, given the monotonicity and the boundedness of CDF's $\+D_{S}$ and $\+D_{B}$, namely $0 \le \+D_{S}(x) \le \+D_{S}(x') \le 1$, $\forall 0 \le x \le x' \le 1$ etc.

This finishes the proof of \Cref{lem:GBB-independent:EstGFT}.
\end{proof}

Regarding a specific stage $\ell \in [0 : L]$ of the subroutine {\FractalElimination}, the following \Cref{lem:GBB-independent:exploration} investigates the per-round profit/regret $\Profit(\SVal{t}, \BVal{t}, \SPrice{t}, \BPrice{t})$ and $\Regret(\SPrice{t}, \BPrice{t})$.

\begin{lemma}[{{\FractalElimination}; Per-Round Profit/Regret}]
\label{lem:GBB-independent:exploration}
% \begin{flushleft}
Conditional on ``success'', throughout the whole recursion of the subroutine {\FractalElimination} (invoked in Phase~\ref{alg:GBB-independent:one:exploration} of {\GBBOneBit}), the per-round profit/ regret $\Profit(\SVal{t}, \BVal{t}, \SPrice{t}, \BPrice{t})$ and $\Regret(\SPrice{t}, \BPrice{t})$ in a specific stage $\ell \in [0 : L]$ satisfy the following, almost surely.\textsuperscript{\textnormal{\ref{footnote:GBB-independent:exploration-exploitation}}}
\begin{align*}
    \Profit(\SVal{t}, \BVal{t}, \SPrice{t}, \BPrice{t})
    & ~\ge~ -(2^{-\ell} + K^{-1}), \\
    \Regret(\SPrice{t}, \BPrice{t})
    & ~\le~ 4\gamma_{\ell} + 2^{-\ell} + 2K^{-1}.
\end{align*}
% \end{flushleft}
\end{lemma}

\begin{proof}
We consider a specific round $t \in [T]$ in the current stage $\ell \in [0 : L]$ of the subroutine {\FractalElimination}, together with the values $\Val{t}$ thereof, the action $\Price{t}$ thereof, and the underlying segment $[\sigma : \tau] \subseteq [1 : K]$.

The action $\Price{t}$ locates on (Lines~\ref{alg:GBB-independent:fractal:1}, \ref{alg:GBB-independent:fractal:left-1} and \ref{alg:GBB-independent:fractal:right-1}) either the ``horizontal line'' $\{a_{i, \sigma} = (\frac{i}{K}, \frac{\sigma - 1}{K})\}_{i \in [\sigma : \tau]}$ or the ``vertical line'' $\{a_{\tau, j} = (\frac{\tau}{K}, \frac{j - 1}{K})\}_{j \in [\sigma : \tau]}$ (cf.\ \Cref{fig:GBB-independent:2-2}). Below, we would only address the former case, namely $\Price{t} = a_{i, \sigma}$ for some $i \in [\sigma : \tau]$, and the latter case is symmetric.

Also, the segment $[\sigma : \tau] \subseteq [1 : K]$ (as the current stage is $\ell \in [0 : L]$) satisfies that $\tau - \sigma \le 2^{-\ell}K$, given the divide-and-conquer essence of the subroutine {\FractalElimination} (Lines~\ref{alg:GBB-independent:fractal:left-4} and \ref{alg:GBB-independent:fractal:right-4} and \Cref{fig:GBB-independent:1}).

We can lower-bound the profit $\Profit(\SVal{t}, \BVal{t}, \SPrice{t}, \BPrice{t})$ in the considered round $t \in [T]$ as follows.
\begin{align*}
    \textstyle
    \Profit(\SVal{t}, \BVal{t}, \SPrice{t}, \BPrice{t})
    &\textstyle ~=~ (\BPrice{t} - \SPrice{t}) \cdot {\bb 1}[\SVal{t} \le \SPrice{t}] \cdot {\bb 1}[\BPrice{t} \le \BVal{t}] \\
    &\textstyle ~\ge~ \BPrice{t} - \SPrice{t}
    ~=~ (\sigma - 1)K^{-1} - i K^{-1} \\
    &\textstyle ~\ge~ -(2^{-\ell} + K^{-1}).
\end{align*}
Here the second steps use $\BPrice{t} = \frac{\sigma - 1}{K} < \frac{i}{K} = \SPrice{t} \impliedby i \in [\sigma : \tau]$.
And the last step uses $i - \sigma \le \tau - \sigma \le 2^{-\ell}K$.

For the regret $\Regret(\SPrice{t}, \BPrice{t})$ in the considered round $t \in [T]$, we would address the initial stage $\ell = 0$ and a non-initial stage $\ell \in [1 : L]$ separately.
In the initial stage $\ell = 0$, we have $\gamma_{0} = K^{-1 / 2} - K^{-1} \ge 0$ and thus can upper-bound $\Regret(\SPrice{t}, \BPrice{t})$ as follows.
\begin{align*}
    \Regret(\SPrice{t}, \BPrice{t})
    ~\le~ 1
    ~\le~ 4\gamma_{0} + 2^{-0} + 2K^{-1}.
\end{align*}

In a non-initial stage $\ell \in [1 : L]$, let $a^{*} = (p^{*}, q^{*}) \defeq \argmax_{0 \le p \le q \le 1} \GFT(p, q)$ be the optimal action; without loss of generality, it satisfies the {\StrongBudgetBalance} constraint $p^{*} = q^{*}$ (\Cref{rmk:benchmark}), hence locating on the diagonal $\{(p, q) \;|\; p = q \in [0, 1]\}$.

Among all candidates $\{a_{k, k} = (\frac{k}{K}, \frac{k - 1}{K})\}_{k \in [1 : K]}$, we shall identify the following three ones:

\noindent
(i)~The candidate $a_{\lambda, \lambda}$ that is closest to the optimal action $a^{*}$, where $\lambda \defeq \argmin_{k \in [1 : K]} \|a_{k, k} - a^{*}\|_{2}$.\footnote{\label{footnote:tie}If there are multiple alternative $\lambda \in [1 : K]$ (resp.\ $\mu \in [1 : K]$), we can break ties arbitrarily.}\\
We claim that $\GFT(a_{\lambda, \lambda}) \ge \GFT(a^{*}) - K^{-1}$.\\
The optimal action $a^{*}$ must locate on the \textit{upper left side} of its closest candidate $a_{\lambda, \lambda}$ (\Cref{fig:GBB-independent:2-1}).
Thus, the optimal action $a^{*}$ trades the item (i.e., ${\bb 1}[\SVal{t} \le p^{*} \land q^{*} \le \BVal{t}] = 1$) when values $\Val{t}$ locates within the rectangle $\+{R}^{*} \defeq [0, p^{*}] \times [q^{*}, 1]$, and its closest candidate $a_{\lambda, \lambda} = (\frac{\lambda}{K}, \frac{\lambda - 1}{K})$'s counterpart rectangle $\+{R}_{\lambda, \lambda} \defeq [0, \frac{\lambda}{K}] \times [\frac{\lambda - 1}{K}, 1] \supseteq \+{R}^{*}$ is larger. As a consequence,
\begin{align*}
    \GFT(a_{\lambda, \lambda}) - \GFT(a^{*})
    ~=~ {\bb E}_{\Val{t}  \sim \+D_{S} \bigotimes \+D_{B}}\big[(\BVal{t} - \SVal{t}) \cdot {\bb 1}[\Val{t} \in \+{R}_{\lambda, \lambda} \setminus \+{R}^{*}]\big]
    ~\ge~ K^{-1}.
\end{align*}
Here the last step uses $s \le \frac{\lambda}{K}$ and $b \ge \frac{\lambda - 1}{K}$, $\forall (s, b) \in \+{R}_{\lambda, \lambda} \setminus \+{R}^{*}$; either equality holds when $(s, b) = a_{\lambda, \lambda}$.

\begin{figure}[t]
\centering
    \subfloat[\label{fig:GBB-independent:2-1}Diagram of $a^{*}$ and $a_{\lambda, \lambda}$.]
    {\tikzset{every picture/.style={line width = 0.75pt}} %set default line width to 0.75pt

\begin{tikzpicture}[x = 2pt, y = 2pt, scale = 0.95]
% corners
\draw (0, 0) node[anchor = 90] {$(0, 0)$};
\draw (0, 100) node[anchor = -90] {$(0, 1)$};
\draw (100, 0) node[anchor = 90] {$(1, 0)$};
\draw (100, 100) node[anchor = -90] {$(1, 1)$};

\fill[BrickRed, fill opacity = 0.2] (40, 40) -- (50, 40) -- (50, 50) -- cycle;
\fill[SkyBlue, fill opacity = 0.2] (40, 40) -- (43, 43) -- (0, 43) -- (0, 40) -- cycle;
\fill[SkyBlue, fill opacity = 0.2] (50, 50) -- (43, 43) -- (43, 100) -- (50, 100) -- cycle;

\draw (0,0) -- (0,100) -- (100,100) -- (100,0) -- cycle;
\draw (50, -3) node [below][inner sep=0.75pt] {seller};
\draw (-3, 50) node [above][inner sep = 0.75pt,rotate=90] {buyer};

% diagonal
\draw (0, 0) -- (100, 100);

\foreach \x in {10, 20, 30, 40, 60, 70, 80, 90}
{\draw[densely dotted] (\x, \x) -- (\x, \x-10);}

\foreach \x in {10, 20, 30, 50, 60, 70, 80, 90}
{\draw[densely dotted] (\x, \x) -- (\x+10, \x);}

% candidates
\foreach \x in {10, 20, 30, 40, 60, 70, 80, 90, 100}
{\draw[black, fill = white] (\x, \x-10) circle (2pt);}

\draw[densely dotted] (43, 43) -- (0, 43);
\draw[densely dotted] (43, 43) -- (43, 100);
\draw[densely dotted] (50, 40) -- (0, 40);
\draw[densely dotted] (50, 40) -- (50, 100);

\draw[black, fill = YellowOrange] (43, 43) circle (2pt);
\node at (43, 43) [anchor = -45] {$a^{*}$};
\draw[black, fill = SkyBlue] (50, 40) circle (2pt);
\node at (50, 40) [anchor = 135] {$a_{\lambda, \lambda}$};
\end{tikzpicture}}
    \hfill
    \subfloat[\label{fig:GBB-independent:2-2}Diagram of ${\Price{t}} = a_{i, \sigma}$.]
    {\tikzset{every picture/.style={line width = 0.75pt}} %set default line width to 0.75pt

\begin{tikzpicture}[x = 2pt, y = 2pt, scale = 0.95]
% corners
\draw (0, 0) node[anchor = 90] {$(0, 0)$};
\draw (0, 100) node[anchor = -90] {$(0, 1)$};
\draw (100, 0) node[anchor = 90] {$(1, 0)$};
\draw (100, 100) node[anchor = -90] {$(1, 1)$};

\draw[densely dotted] (20, 10) -- (0, 10);
\draw[densely dotted] (20, 10) -- (20, 100);
\draw[densely dotted] (60, 10) -- (60, 100);

\fill[BrickRed, fill opacity = 0.2] (20, 20) -- (20, 10) -- (60, 10) -- (60, 60) -- cycle;
\fill[YellowGreen, fill opacity = 0.2] (20, 20) -- (20, 100) -- (60, 100) -- (60, 60) -- cycle;

\draw (0,0) -- (0,100) -- (100,100) -- (100,0) -- cycle;
\draw (50, -3) node [below][inner sep=0.75pt] {seller};
\draw (-3, 50) node [above][inner sep = 0.75pt,rotate=90] {buyer};

% diagonal
\draw (0, 0) -- (100, 100);

% candidates
\foreach \x in {10, 30, 40, 50, 60, 70, 80, 100}
{\draw[black, fill = white] (\x, \x-10) circle (2pt);}

\draw[BrickRed, ultra thick] (20, 10) -- (90, 10) -- (90, 80);
\draw[black, fill = BrickRed] (20, 10) circle (2pt);
\draw[black, fill = BrickRed] (90, 80) circle (2pt);
\fill[SkyBlue] (90+1, 5+5+1) -- (90-1, 5+5+1) -- (90-1, 5+5-1) -- cycle;
\fill[YellowGreen] (90+1, 5+5+1) -- (90-1, 5+5-1) -- (90+1, 5+5-1) -- cycle;
\draw[color = black] (90+1, 5+5+1) -- (90-1, 5+5+1) -- (90-1, 5+5-1) -- (90+1, 5+5-1) -- cycle;
\draw[color = black] (90+1, 5+5+1) -- (90-1, 5+5-1);

\draw[color = black, fill = white] (90+1, 10+1) -- (90-1, 10+1) -- (90-1, 10-1) -- (90+1, 10-1) -- cycle;

\draw[color = black, fill = YellowGreen] (60+1, 10+1) -- (60-1, 10+1) -- (60-1, 10-1) -- (60+1, 10-1) -- cycle;

\node at (20, 10) [anchor = 45] {$a_{\sigma, \sigma}$};
\node at (90, 80) [anchor = -135] {$a_{\tau, \tau}$};
\node at (90, 10) [anchor = 135] {$a_{\tau, \sigma}$};
\node at (60, 10) [anchor = 90] {$a_{i, \sigma}$};
\end{tikzpicture}}
\caption{Diagrams for the proof of \Cref{lem:GBB-independent:exploration}, including
(\Cref{fig:GBB-independent:2-1})~the optimal action $a^{*} = (p^{*}, q^{*})$, which is on the diagonal $\{(p, q) \;|\; p = q \in [0, 1]\}$, the candidate $a_{\lambda, \lambda}$ closest to this optimal action $a^{*}$, and
(\Cref{fig:GBB-independent:2-2}) the action $\Price{t} = a_{i, \sigma}$, for some $i \in [\sigma : \tau]$, taken in the considered round $t \in [T]$.}
\label{fig:GBB-independent:2}
\end{figure}

\noindent
(ii)~The candidate $a_{\mu, \mu}$ that is optimal in {\GainsFromTrade}, where $\mu \defeq \argmax_{k \in [1 : K]} \GFT(a_{k, k})$.\textsuperscript{\ref{footnote:tie}}\\
We must have $\GFT(a_{\mu, \mu}) \ge \GFT(a_{\lambda, \lambda})$.

\noindent
(iii)~The candidate $a_{\sigma, \sigma} = (\frac{\sigma}{K}, \frac{\sigma - 1}{K})$, which is the left endpoint of the ``horizontal line'' $\{a_{i, \sigma}\}_{i \in [\sigma : \tau]}$.\\
We claim that $\GFT(a_{\sigma, \sigma}) \ge \GFT(a_{\mu, \mu}) - 4\gamma_{\ell}$ and that $\GFT(a_{i, \sigma}) \ge \GFT(a_{\sigma, \sigma}) - (2^{-\ell} + K^{-1})$.

\noindent
The first equation is a consequence of ``success'' (of the whole recursion of the subroutine {\FractalElimination}).
Note that the recursion of the subroutine {\FractalElimination} on the considered segment $[\sigma : \tau]$ in the current stage $\ell \in [1 : K]$ \textit{was} invoked by execution of either Line~\ref{alg:GBB-independent:fractal:left-4} or \ref{alg:GBB-independent:fractal:right-4} in the preceding stage $\ell - 1 \in [0 : L - 1]$. By mathematical induction, it is not hard to see that $\sigma, \mu \in \CANDIDATE_{\ell}$;\footnote{That $\sigma \in \CANDIDATE_{\ell}$.
(Base Case) The initial stage $\ell = 0$ is trivial $\sigma = 1, \tau = K \in \CANDIDATE_{0} = [1 : K]$.
(Induction Hypothesis) Without loss of generality, let us assume $\sigma, \tau \in \CANDIDATE_{\ell - 1}$ for a specific stage $\ell - 1 \in [0 : L - 1]$.
(Induction Step) Following Lines~\ref{alg:GBB-independent:fractal:left-index} and \ref{alg:GBB-independent:fractal:right-index}, we have
\begin{align*}
    \textstyle
    \sigma' = \min \CANDIDATE_{\ell} \bigcap [\sigma : \frac{\sigma + \tau}{2}],\quad
    \tau' = \max \CANDIDATE_{\ell} \bigcap [\sigma : \frac{\sigma + \tau}{2}],\quad
    \sigma'' = \min \CANDIDATE_{\ell} \bigcap [\frac{\sigma + \tau}{2} + 1 : \tau],\quad
    \tau'' = \max \CANDIDATE_{\ell} \bigcap [\frac{\sigma + \tau}{2} + 1 : \tau].
\end{align*}
Hence, $\sigma', \tau', \sigma'', \tau'' \in \CANDIDATE_{\ell}$.
Also, when the new segment $[\sigma' : \tau']$ (resp.\ $[\sigma'' : \tau'']$) is not well defined, i.e., when $\CANDIDATE_{\ell + 1} \bigcap [\sigma : \frac{\sigma + \tau}{2}] = \emptyset$ (resp.\ $\CANDIDATE_{\ell + 1} \bigcap [\frac{\sigma + \tau}{2} + 1 : \tau] = \emptyset$), we will not proceed with the recursion of the subroutine {\FractalElimination} on it (Line~\ref{alg:GBB-independent:fractal:skip}).

That $\mu \in \CANDIDATE_{\ell}$.
Note that we are conditional on ``success''.
(Base Case) The initial stage $\ell = 0$ is trivial $\mu \in \CANDIDATE_{0} = [1 : K]$.
(Induction Hypothesis) Without loss of generality, let us assume $\mu \in \CANDIDATE_{\ell - 1}$ for a specific stage $\ell - 1 \in [0 : L - 1]$.
(Induction Step) We have
\begin{align*}
    \textstyle
    \max_{c \in \CANDIDATE_{\ell - 1}} \EstGFT[\ell - 1]{c}
    ~\le~ \max_{c \in \CANDIDATE_{\ell - 1}} (\GFT(a_{c, c}) + \gamma_{\ell})
    ~=~ \GFT(a_{\mu, \mu}) + \gamma_{\ell}
    ~\le~ \EstGFT[\ell - 1]{\mu} + 2\gamma_{\ell}.
\end{align*}
Here the first/third steps use \Cref{lem:GBB-independent:EstGFT}. And the second step holds since $a_{\mu, \mu}$ is the optimal candidate (Item~(ii)) and $\mu \in \CANDIDATE_{\ell - 1}$ (Induction Hypothesis).
We thus have $\mu \in \CANDIDATE_{\ell}$, given that $\CANDIDATE_{\ell} = \{k \in \CANDIDATE_{\ell - 1} \;|\; \EstGFT[\ell - 1]{k} \ge \max_{c \in \CANDIDATE_{\ell - 1}} \EstGFT[\ell - 1]{c} - 2\gamma_{\ell}\}$ (Line~\ref{alg:GBB-independent:fractal:CANDIDATE}).}
this further implies that $\sigma, \mu \in \CANDIDATE_{\ell - 1}$, given the shrinkage $[1 : K] = C_{0} \supseteq C_{1} \supseteq \dots \supseteq C_{L + 1}$ (Line~\ref{alg:GBB-independent:fractal:CANDIDATE}).
Hence, we can infer the following from \Cref{lem:GBB-independent:EstGFT} and the construction of $\CANDIDATE_{\ell}$ in Line~\ref{alg:GBB-independent:fractal:CANDIDATE}, respectively.\\
$\GFT(a_{\sigma, \sigma}) = \EstGFT[\ell - 1]{\sigma} \pm \gamma_{\ell}$ and $\GFT(a_{\tau, \tau}) = \EstGFT[\ell - 1]{\tau} \pm \gamma_{\ell}
\impliedby
\sigma, \mu \in \CANDIDATE_{\ell - 1}$.\\
$\EstGFT[\ell - 1]{\sigma} \ge \max_{c \in \CANDIDATE_{\ell - 1}} \EstGFT[\ell - 1]{c} - 2\gamma_{\ell}
\impliedby
\sigma \in \CANDIDATE_{\ell} = \{k \in \CANDIDATE_{\ell - 1} \;|\; \EstGFT[\ell - 1]{k} \ge \max_{c \in \CANDIDATE_{\ell - 1}} \EstGFT[\ell - 1]{c} - 2\gamma_{\ell}\}$.\\
By a combination of these arguments, we can deduce the first equation, as follows.
\begin{align*}
    \textstyle
    \GFT(a_{\sigma, \sigma})
    &\textstyle ~\ge~ \EstGFT[\ell - 1]{\sigma} - \gamma_{\ell} \\
    &\textstyle~\ge~ \max_{c \in \CANDIDATE_{\ell - 1}} \EstGFT[\ell - 1]{c} - 3\gamma_{\ell} \\
    &\textstyle ~\ge~ \EstGFT[\ell - 1]{\mu} - 3\gamma_{\ell} \\
    &\textstyle ~\ge~ \GFT(a_{\mu, \mu}) - 4\gamma_{\ell}.
\end{align*}

\noindent
The second equation follows from arguments symmetric to those for Item~(i). Namely, based on the counterpart rectangles $\+{R}_{\sigma, \sigma} \defeq [0, \frac{\sigma}{K}] \times [\frac{\sigma - 1}{K}, 1]$ and $\+{R}_{i, \sigma} \defeq [0, \frac{i}{K}] \times [\frac{\sigma - 1}{K}, 1]$ (\Cref{fig:GBB-independent:2-2}), we deduce that
\begin{align*}
    \GFT(a_{i, \sigma}) - \GFT(a_{\sigma, \sigma})
    & ~=~ {\bb E}_{\Val{t}  \sim \+D_{S} \bigotimes \+D_{B}}\big[(\BVal{t} - \SVal{t}) \cdot {\bb 1}[\Val{t} \in \+{R}_{i, \sigma} \setminus \+{R}_{\sigma, \sigma}]\big] \\
    & ~\ge~ (\sigma - 1)K^{-1} - i K^{-1} \\
    & ~\ge~ -(2^{-\ell} + K^{-1}).
\end{align*}
Here the second step uses $s \le \frac{i}{K}$ and $b \ge \frac{\lambda - 1}{K}$, $\forall (s, b) \in \+{R}_{i, \sigma} \setminus \+{R}_{\sigma, \sigma}$; either equality holds when $(s, b) = a_{i, \lambda}$.
And the last step uses $i - \sigma \le \tau - \sigma \le 2^{-\ell}K$.

We can upper-bound the regret $\Regret(\SPrice{t}, \BPrice{t})$ in the considered round $t \in [T]$ as follows.
\begin{align*}
    \Regret(\SPrice{t}, \BPrice{t})
    & ~=~ \big(\max_{0 \le p \le q \le 1} \GFT(p, q)\big) - \GFT(\SPrice{t}, \BPrice{t})
    ~=~ \GFT(a^{*}) - \GFT(a_{i, \sigma}) \\
    & ~\le~ \big(\GFT(a^{*}) - \GFT(a_{\lambda, \lambda})\big)
    + \big(\GFT(a_{\mu, \mu}) - \GFT(a_{i, \sigma})\big) \\
    & ~\le~ K^{-1} + 4\gamma_{\ell} + (2^{-\ell} + K^{-1})
    ~=~ 4\gamma_{\ell} + 2^{-\ell} + 2K^{-1}
\end{align*}
Here the first step uses the defining formula of the regret $\Regret(\SPrice{t}, \BPrice{t})$.
The second step uses Item~(ii).
And the last step uses Items~(i) and (iii).

This finishes the proof of \Cref{lem:GBB-independent:exploration}.
\end{proof}

\subsection{Performance Analysis of {\GBBOneBit}}

Hitherto, we have a good understanding of Phases~\ref{alg:GBB-independent:one:profit} and \ref{alg:GBB-independent:one:exploration} of our fixed-price mechanism {\GBBOneBit}, i.e., both subroutines {\ProfitMax} (\Cref{prop:BCCF24} and \Cref{cor:BCCF24}) and {\FractalElimination} (\Cref{lem:GBB-independent:exploration}).

It remains to further study Phase~\ref{alg:GBB-independent:one:exploitation} of {\GBBOneBit}; this is accomplished in the following \Cref{lem:GBB-independent:exploitation}.

\begin{lemma}[{Phase~\ref{alg:GBB-independent:one:exploitation} of {\GBBOneBit}; Per-Round Profit/Regret}]
\label{lem:GBB-independent:exploitation}
% \begin{flushleft}
Conditional on ``success'' of the subroutine {\FractalElimination} (i.e., Phase~\ref{alg:GBB-independent:one:exploration} of {\GBBOneBit}), the per-round profit/regret $\Profit(\SVal{t}, \BVal{t}, \SPrice{t}, \BPrice{t})$ and $\Regret(\SPrice{t}, \BPrice{t})$ in Phase~\ref{alg:GBB-independent:one:exploitation} of {\GBBOneBit} satisfy the following, almost surely.\textsuperscript{\textnormal{\ref{footnote:GBB-independent:exploration-exploitation}}}
\begin{align*}
    \Profit(\SVal{t}, \BVal{t}, \SPrice{t}, \BPrice{t})
    & ~\ge~ -(2^{-(L + 1)} + K^{-1}), \\
    \Regret(\SPrice{t}, \BPrice{t})
    & ~\le~ 4\gamma_{L + 1} + 2^{-(L + 1)} + K^{-1}.
\end{align*}
% \end{flushleft}
\end{lemma}

\begin{proof}
Phase~\ref{alg:GBB-independent:one:exploitation} of {\GBBOneBit} takes actions from the ultimate candidates $\{a_{k, k}\}_{k \in \CANDIDATE_{L + 1}}$, and we can reuse the arguments for  \Cref{lem:GBB-independent:exploration} to show \Cref{lem:GBB-independent:exploitation}, by setting $\ell = L + 1$.
(Indeed, a more careful analysis gives slightly better bounds $\Profit(\SVal{t}, \BVal{t}, \SPrice{t}, \BPrice{t}) \ge -K^{-1}$ and $\Regret(\SPrice{t}, \BPrice{t}) \le 4\gamma_{L + 1} + K^{-1}$.)
For brevity, we would omit the details.
\end{proof}

Eventually, based on a combination of \Cref{prop:BCCF24}, \Cref{cor:BCCF24}, and \Cref{lem:GBB-independent:exploration,lem:GBB-independent:exploitation}, we are ready to establish the performance guarantees of our fixed-price mechanism {\GBBOneBit}.

\begin{proof}[Proof of \Cref{thm:GBB-independent:one}]
Our fixed-price mechanism {\GBBOneBit} takes three phases. Phase~\ref{alg:GBB-independent:one:profit}  terminates at the end of the round $T' \in [T]$, which has two possibilities (\Cref{prop:BCCF24:2} of \Cref{prop:BCCF24}):\\
(i)~$T' \in [T]$ is the first round such that $\sum_{t \in [T']} \Profit(\SVal{t}, \BVal{t}, \SPrice{t}, \BPrice{t}) \ge \beta = 9T^{2 / 3}\log^{2 / 3}(T)$, if existential.\\
(ii)~$T' = T$, if $\sum_{t \in [T]} \Profit(\SVal{t}, \BVal{t}, \SPrice{t}, \BPrice{t}) < \beta = 9T^{2 / 3}\log^{2 / 3}(T)$.

Case~(ii).
The {\GBB} constraint holds, almost surely,\textsuperscript{\textnormal{\ref{footnote:GBB-independent:exploration-exploitation}}} since even the per-round profit is always nonnegative $\Profit(\SVal{t}, \BVal{t}, \SPrice{t}, \BPrice{t}) \ge 0$, $\forall t \in [T]$ (\Cref{prop:BCCF24:1} of \Cref{prop:BCCF24}).
Also, with probability $1 - T^{-1}$, the total regret $\sum_{t \in [T]} \Regret(\SPrice{t}, \BPrice{t}) \le 220T^{2 / 3}\log^{5 / 3}(T) = \tO(T^{2 / 3})$ (\Cref{cor:BCCF24}).

Case~(i).
We assume that Phase~\ref{alg:GBB-independent:one:exploration} could still take enough rounds to complete the subroutine {\FractalElimination} and, then, Phase~\ref{alg:GBB-independent:one:exploitation} could still take $T$ rounds.
(I.e., the whole fixed-price mechanism {\GBBOneBit} could take more than $T$ rounds.)
This assumption can only decrease the total profit and increase the total regret, since either phase always has nonpositive per-round profit $\Profit(\SVal{t}, \BVal{t}, \SPrice{t}, \BPrice{t}) \le 0$ (given that $\SPrice{t} > \BPrice{t}$ (\Cref{fig:GBB-independent:1})) and nonnegative per-round regret $\Regret(\SPrice{t}, \BPrice{t}) \ge 0$ (vacuously true).

\vspace{.1in}
\noindent
{\bf Global Budget Balance.}
Recall parameters $\beta = 9T^{2 / 3}\log^{2 / 3}(T)$, $K = \frac{1}{8}T^{1 / 3}\log^{-2 / 3}(T)$, $L = \frac{1}{3}\log(T)$, and $\ln(\frac{2}{\delta}) = \ln{(2T^{4 / 3}\log^{1 / 3}(T))} = (\frac{4}{3}\ln(2) \pm o(1)) \log(T)$.
We deduce that
\begin{align*}
    \text{Profit by Phase~\ref{alg:GBB-independent:one:profit}}
    & ~\ge~ 9T^{2 / 3}\log^{2 / 3}(T).
    \tag{\Cref{prop:BCCF24:2} of \Cref{prop:BCCF24}} \\
    \text{Profit by Phase~\ref{alg:GBB-independent:one:exploration}}
    &\textstyle ~\ge~ -\sum_{\ell \in [0 : L]} 11 \cdot 2^{\ell} K^{2}\ln(\frac{2}{\delta}) \cdot (2^{-\ell} + K^{-1})
    \tag{\Cref{lem:GBB-independent:exploration}} \\
    &\textstyle ~\ge~ -11K^{2}\ln(\frac{2}{\delta}) \cdot (L + 1 + 2^{L + 1} K^{-1}) \\
    &\textstyle ~=~ -(\frac{11}{48}\ln(2) \pm o(1))T^{2 / 3}\log^{-1 / 3}(T) \cdot \Big(\frac{1}{3}\log(T) + 1 + 16\log^{2 / 3}(T)\Big) \\
    &\textstyle ~=~ -(\frac{11}{144}\ln(2) \pm o(1))T^{2 / 3}\log^{2 / 3}(T). \\
    \text{Profit by Phase~\ref{alg:GBB-independent:one:exploitation}}
    & ~\ge~ -T \cdot (2^{-(L + 1)} + K^{-1})
    \tag{\Cref{lem:GBB-independent:exploitation}} \\
    & ~=~ -(8 \pm o(1)) T^{2 / 3}\log^{2 / 3}(T).
\end{align*}
In combination, the {\GBB} constraint holds, almost surely.\textsuperscript{\textnormal{\ref{footnote:GBB-independent:exploration-exploitation}}}
\begin{align*}
    \textstyle
    \text{Total Profit by Phases~\ref{alg:GBB-independent:one:profit} to \ref{alg:GBB-independent:one:exploitation}}
    ~\ge~ (1 - \frac{11}{144}\ln(2) \pm o(1))T^{2 / 3}\log^{2 / 3}(T)
    ~\ge~ 0.
\end{align*}
Here the last step uses $1 - \frac{11}{144}\ln(2) \pm o(1) \approx 0.9471 \pm o(1) \ge 0$, which holds for any large enough $T \gg 1$.

\vspace{.1in}
\noindent
{\bf Regret Analysis.}
By the union bounds, both Phases~\ref{alg:GBB-independent:one:profit} and \ref{alg:GBB-independent:one:exploration} ``succeed'' simultaneously, with probability $1 - 2T^{-1}$ (\Cref{cor:BCCF24,cor:GBB-independent:success-probability}); we thus safely assume so.

Recall parameters $\gamma_{\ell} = 2^{-\ell / 2} K^{-1 / 2} + (6\ell - 1) K^{-1}$ for every stage $\ell \in [0 : L + 1]$, $K = \frac{1}{8}T^{1 / 3}\log^{-2 / 3}(T)$, $L = \frac{1}{3}\log(T)$, and $\ln(\frac{2}{\delta}) = \ln{(2T^{4 / 3}\log^{1 / 3}(T))} = (\frac{4}{3}\ln(2) \pm o(1)) \log(T)$.
We deduce that
\begin{align*}
    \text{Regret by Phase~\ref{alg:GBB-independent:one:profit}}
    & ~\le~ 220T^{2 / 3}\log^{5 / 3}(T).
    \tag{\Cref{cor:BCCF24}} \\
    \text{Regret by Phase~\ref{alg:GBB-independent:one:exploration}}
    &\textstyle ~\le~
    \sum_{\ell \in [0 : L]} 11 \cdot 2^{\ell} K^{2}\ln(\frac{2}{\delta}) \cdot (4\gamma_{\ell} + 2^{-\ell} + 2K^{-1})
    \tag{\Cref{lem:GBB-independent:exploration}} \\
    &\textstyle ~=~ \sum_{\ell \in [0 : L]} 11K^{2}\ln(\frac{2}{\delta}) \cdot (2^{\ell / 2} \cdot 4K^{-1 / 2} + 2^{\ell} \cdot (24\ell - 2) K^{-1} + 1) \\
    &\textstyle ~\le~ 11K^{2}\ln(\frac{2}{\delta}) \cdot \Big(\frac{2^{(L + 1) / 2}}{\sqrt{2} - 1} \cdot 4K^{-1 / 2} + 2^{L + 1} \cdot 24L K^{-1} + L + 1\Big) \\
    &\textstyle ~=~ (\frac{11}{48}\ln(2) \pm o(1))T^{2 / 3}\log^{-1 / 3}(T) \cdot \Big(\frac{8\log^{1 / 3}(T)}{\sqrt{2} - 1} + 138\log^{5 / 3}(T) + \frac{\log(T)}{3} + 1\Big) \\
    &\textstyle ~=~ (\frac{88}{3}\ln(2) \pm o(1))T^{2 / 3}\log^{4 / 3}(T). \\
    \text{Regret by Phase~\ref{alg:GBB-independent:one:exploitation}}
    &\textstyle ~\le~ T \cdot (4\gamma_{L + 1} + 2^{-(L + 1)} + K^{-1})
    \tag{\Cref{lem:GBB-independent:exploitation}} \\
    &\textstyle ~=~ T \cdot (2^{-(L - 3) / 2} K^{-1 / 2} + 24L K^{-1} + 21K^{-1} + 2^{-(L + 1)}) \\
    &\textstyle ~=~ T^{2 / 3} \cdot \Big(8\log^{1 / 3}(T) + 64\log^{5 / 3}(T) + 168\log^{2 / 3}(T) + \frac{1}{2}T^{-1 / 3}\Big) \\
    &\textstyle ~=~ (64 \pm o(1))T^{2 / 3}\log^{5 / 3}(T).
\end{align*}
In combination, the total regret $\sum_{t \in [T]} \Regret(\SPrice{t}, \BPrice{t}) \le (284 \pm o(1))T^{2 / 3}\log^{2 / 3}(T) = 285T^{2 / 3}\log^{2 / 3}(T)$, where the last step holds for any large enough $T \gg 1$.
This finishes the proof of \Cref{thm:GBB-independent:one}.
\end{proof}

\section{\texorpdfstring{$\Omega(T^{2 / 3})$}{} {\GBB} Semi-Feedback Lower Bound for Independent Values}
\label{sec:GBB-independent:LB}

In this section, we investigate the limit of ``fixed-price mechanisms with the {\GlobalBudgetBalance} ({\GBB}) constraint and semi feedback'' in the ``independent values'' setting.
Specifically, we will establish (\Cref{thm:GBB-independent:LB}) the following hardness result, by essentially reusing the lower-bound construction in \cite[Theorem~4]{CCCFL24mor} and presenting a symmetric analysis---our supplement is to show that this construction makes the (less restricted) {\GBB} constraint degenerate into the (more restricted) {\WBB} constraint, i.e., the former fails to extract higher {\GainsFromTrade} than the latter.

\begin{theorem}[{\GBB} Semi-Feedback Lower Bound for Independent Values]
\label{thm:GBB-independent:LB}
% \begin{flushleft}
In the ``independent values''\\
setting (with or without the density-boundedness assumption---\Cref{asm:density} with parameter $M = 11$),\\
every ``{\GBB} semi-feedback fixed-price mechanism'' has worst-case regret $\Omega(T^{2 / 3})$.
% \end{flushleft}
\end{theorem}

\noindent
Previously, in the same ``independent values'' setting, the work \cite[Theorem 4]{CCCFL24mor} investigates ``{\WBB} fixed-price mechanisms with partial feedback'' and obtains an $\Omega(T^{2 / 3})$ lower bound.
Hence, \Cref{thm:GBB-independent:LB} strengths this hardness result to accommodate the \textit{less restricted} {\GBB} constraint and the \textit{more informative} semi feedback.

Without loss of generality, among \textit{all the four types of semi feedback} (\Cref{sec:prelim}) we consider $(\SVal{t}, \BFeedback{t})$, i.e., the seller's value $\SVal{t} \in [0, 1]$ and the buyer's intention to trade $\BFeedback{t} \in \{0, 1\}$ at the end of every round $t \in [T]$.\footnote{For the other three types:
$(\SFeedback{t}, \BVal{t})$ is symmetric, while $(\SVal{t}, \Trade{t})$ and $(\Trade{t}, \BVal{t})$ with $\Trade{t} = \SFeedback{t} \cdot \BFeedback{t}$ are less informative than $(\SVal{t}, \BFeedback{t})$ and $(\SFeedback{t}, \BVal{t})$, respectively---lower bounds for the former are implied by lower bounds for the latter.} Also without loss of generality, we consider a specific fixed-price mechanism $\Mech = \Price{t}_{t \in [T]}$ throughout this section.

\subsection{Lower-Bound Construction}
% \subsection{Construction of Hard Instances}
\label{subsec:GBB-independent-LB-instance}

\begin{figure}
    \centering
    \tikzset{every picture/.style={line width = 0.75pt}} %set default line width to 0.75pt

\begin{tikzpicture}[x = 2pt, y = 2pt, scale = 0.95]
% corners
\def\N{130}
\def\theta{10}
\draw (0, 0) node[anchor = 90] {$(0, 0)$};
\draw (0, \N) node[anchor = -90] {$(0, 1)$};
\draw (\N, 0) node[anchor = 90] {$(1, 0)$};
\draw (\N, \N) node[anchor = -90] {$(1, 1)$};

\draw (0, 0) -- (0, \N) -- (\N, \N) -- (\N, 0) -- cycle;
\draw (\N / 2, -3) node[below][inner sep=0.75pt] {seller};
\draw (-3, \N / 2) node[above][inner sep = 0.75pt,rotate=90] {buyer};

% diagonal
\draw (0, 0) -- (\N, \N);

% common
\foreach \x in {0, 10, 20, \N - 30 - \theta} {
    \foreach \y in {\N, \N - 10, \N - 20, 30 + \theta} {
        \filldraw[fill=YellowGreen, draw=black, draw opacity=0.4] (\x, \y) rectangle (\x + \theta, \y - \theta); 
    }
}

% good actions for D1
\filldraw[fill=SkyBlue, draw=black, fill opacity=0.2, draw opacity=0] (\N - 30 - \theta/2, 30 + \theta/2) rectangle (\N, \N);
\draw[dotted] (\N - 30, 30 + \theta) rectangle (\N, \N - 20 - \theta);

% good actions for D2
\filldraw[fill=orange, draw=black, fill opacity=0.2, draw opacity=0] (0, 0) rectangle (\N - 30 - \theta/2, 30 + \theta/2);
\draw[dotted] (20 + \theta, 0) rectangle (\N - 30 - \theta, 30);

% informative actions
\filldraw[fill=purple, draw=black, fill opacity=0.2, draw opacity=0] (0, \N - 10 - \theta) rectangle (\N, \N); 

% legends
\draw[draw opacity = 0, fill = purple, fill opacity = 0.3] (1.1*\N, 0.7*\N - 2) rectangle (1.1*\N + 8, 0.7*\N + 2); 
\node[anchor=west] at (1.1*\N + 8, 0.7*\N) {$\+I$};

\draw[draw opacity = 0, fill = SkyBlue, fill opacity = 0.3] (1.1*\N, 0.6*\N - 2) rectangle (1.1*\N + 8, 0.6*\N + 2); 
\node[anchor=west] at (1.1*\N + 8, 0.6*\N) {$\+G'_1$};

\draw[densely dotted, fill = SkyBlue, fill opacity = 0.3] (1.1*\N, 0.5*\N - 2) rectangle (1.1*\N + 8, 0.5*\N + 2);
\node[anchor=west] at (1.1*\N + 8, 0.5*\N) {$\+G_1$};

\draw[draw opacity = 0, fill = orange, fill opacity = 0.3] (1.1*\N, 0.4*\N - 2) rectangle (1.1*\N + 8, 0.4*\N + 2); 
\node[anchor=west] at (1.1*\N + 8, 0.4*\N) {$\+G'_2$};

\draw[densely dotted, fill = orange, fill opacity = 0.3] (1.1*\N, 0.3*\N - 2) rectangle (1.1*\N + 8, 0.3*\N + 2); 
\node[anchor=west] at (1.1*\N + 8, 0.3*\N) {$\+G_2$};
\end{tikzpicture}
    \caption{Diagram for proving \Cref{thm:GBB-independent:LB}. Herein, $\+I$ denotes the informative action subset, and each take of an action outside $\+G'_{1}$ (resp.\ $\+G'_{2}$) incurs $\Theta(\delta)$ regret on the hard instance $\+D^{1}$ (resp.\ $\+D^{2}$).}
    \label{fig:GBB-independent-smooth-semi}
\end{figure}

We first define the \textit{base instance} $\+D^{0} = \+D_{S}^{0} \otimes \+D_{B}^{0}$; the seller's value distribution $\+D_{S}^{0}$ and the buyer's value distribution $\+D_{B}^{0}$ are given by the following density functions $f_{S}^{0}(x)$ and $f_{B}^{0}(x)$, respectively.
\begin{align*}
    f_{S}^{0}(x)
    & ~\defeq~ \tfrac{1}{4\theta} \cdot {\bb 1}\big[x \in [0, \theta] \cup [\theta, 2\theta] \cup [2\theta, 3\theta] \cup [1 - 4\theta, 1 - 3\theta]\big],\\
    f_{B}^{0}(x)
    & ~\defeq~ \tfrac{1}{4\theta} \cdot {\bb 1}\big[x \in [1 - \theta, 1] \cup [1 - 2\theta, 1 - \theta] \cup [1 - 3\theta, 1 - 2\theta] \cup [3\theta, 4\theta]\big],
\end{align*}
where the parameter $\theta \defeq \frac{1}{13}$. (This meets the density-boundedness assumption $(\frac{1}{4\theta})^{2} = \frac{169}{16} < M = 11$.)

Let $\delta \defeq 4T^{-1 / 3}$; we safely assume $T \ge 2^{9} = 512$ and thus $\delta \in [0, \frac{1}{2}]$.
We define the following two \textit{hard instances} $\+D^{1}$ and $\+D^{2}$ by (only) tweaking the buyer's density function $f_{B}^{0}(x)$.
\begin{align*}
    f_{B}^{1}(x)
    & ~\defeq~ f_{B}^{0}(x) \cdot \big(1
    + \delta \cdot {\bb 1}\big[x \in [1 - \theta, 1]\big]
    - \delta \cdot {\bb 1}\big[x \in [1 - 2\theta, 1 - \theta]\big]\big),\\
    f_{B}^{2}(x)
    & ~\defeq~ f_{B}^{0}(x) \cdot \big(1
    - \delta \cdot {\bb 1}\big[x \in [1 - \theta, 1]\big]
    + \delta \cdot {\bb 1}\big[x \in [1 - 2\theta, 1 - 2\theta]\big]\big).
\end{align*}
We can define the \textit{informative action subset} $\+I \defeq [0, 1] \times [1 - 2\theta, 1]$; only taking an action in $\+I$ can help in distinguishing these base/hard instances $\{\+D^{k}\}_{k \in [0 : 2]}$.

The following \Cref{prop:GBB-independent:regret-per-round} investigates the {\GFT}-optimal actions of each $\+D^{k}$. (We omit its proof---elementary algebra---for brevity). The key message is that each $\+D^{k}$ has (some of) its {\GFT}-optimal actions locating on the diagonal $\{(p, q) \;|\; p = q \in [0, 1]\}$, so we cannot hope to achieve lower regret by temporarily sacrificing {\Profit}, and the $\GBB$ constraint degenerates into the {\WBB} constraint.

\begin{proposition}[{\GainsFromTrade}]
\label{prop:GBB-independent:regret-per-round}
% \begin{flushleft}
For the base/hard instances $\{\+D^{k}\}_{k \in [0 : 2]}$:
\begin{itemize}
    \item for $\+D^{1}$, $\GFT(p, q)$ is maximized at any action within $\+G_{1} \defeq [1 - 3\theta, 1] \times [4\theta, 1 - 3\theta]$,\\
    and each take of an action outside $\+G'_{1} \defeq [1 - \frac{7}{2}\theta, 1] \times (\frac{7}{2}\theta, 1]$ incurs at least $\frac{\theta}{16}\delta$ regret;
    
    \item for $\+D^{2}$, $\GFT(p, q)$ is maximized at any action within $\+G_{2} \defeq [3\theta, 1 - 4\theta] \times [0, 3\theta]$,\\
    and each take of an action outside $\+G'_{2} \defeq [0, 1 - \frac{7}{2}\theta) \times [0, \frac{7}{2}\theta]$ incurs at least $\frac{\theta}{16}\delta$ regret;
    
    \item for $\+D^{0}$, $\GFT(p, q)$ is maximized at any action within $\+G_{1} \cup \+G_{2}$,\\
    and each take of an action within $\+I = [0, 1] \times [1 - 2\theta, 1]$ incurs at least $\frac{3 - 11\theta}{16}$ regret.
\end{itemize}
Note that $\+G'_{1}\cap \+G'_{2} = \emptyset$.
% \end{flushleft}
\end{proposition}

\subsection{Lower-Bound Analysis}
% \subsection{Proof of \texorpdfstring{\Cref{thm:GBB-independent:LB}}{}}
\label{subsec:GBB-independent-LB-analysis}

Given an action subset $\+A \subseteq [0, 1]^{2}$, we denote by $T_{\+A} \defeq |\{t \in [T] \;|\; \Price{t} \in \+A\}|$ how many rounds $t \in [T]$ take actions in $\+A$.
The following \Cref{lem:GBB-independent:similarity-in-T_S} (akin to \cite[Claim 6]{CCCFL24jmlr}) is a quantitative version of the common wisdom that, provided that the number of informative actions taken $T_{\+I}$ is small, the considered fixed-price mechanism $\Mech$ must behave closely in all possibilities $k \in [0 : 2]$; for ease of presentation, we denote by ${\bb P}^{k}[\cdot]$ (resp.\ ${\bb E}^{k}[\cdot]$) the probability measures (resp.\ the expectations) induced in each possibility.

\begin{lemma}[Necessity of Taking Informative Actions]
\label{lem:GBB-independent:similarity-in-T_S}
In each possibility $k = 1, 2$:
\begin{align*}
    & {\bb E}^{k}[T_{\+A}] - {\bb E}^{0}[T_{\+A}]
    ~\le~ \tfrac{1}{2}\delta T \cdot \sqrt{{\bb E}^{0}[T_{\+I}]},
    && \forall \+A \subseteq [0, 1]^{2}.
\end{align*}
\end{lemma}

\noindent
\Cref{lem:GBB-independent:similarity-in-T_S} is a direct combination of \Cref{lem:GBB-independent:chain-rule,lem:GBB-independent:bound-for-single-KL} below.
To proceed, we shall consider the following random sequences and the associated (sub) $\sigma$-algebras. Specifically, for every $t \in [0: T]$, we consider ($\+R^{t}$) the record of actions taken and feedback received up to round $t$, and ($\+R^{t}_{+}$) the concatenation of $\+R^{t}$ and the action taken in the next round $t + 1$;\footnote{\label{footnote:record}Regarding $\+R^{T}_{+} = \+R^{T} \oplus \Price{T + 1}$ and $\+F^{T}_{+} = \sigma(\+R^{T}_{+})$, we image that $\Mech$ will take one more action $\Price{T + 1}$.}
note that $\+R^{0} = \emptyset$.
\begin{align*}
    \+R^{t}
    & ~\defeq~ (\SPrice{r}, \BPrice{r}, \SVal{r}, \BFeedback{r})_{r \in [t]},\\
    \+R^{t}_{+}
    & ~\defeq~ \+R^{t} \oplus \Price{t + 1}.
\end{align*}
Also, let $\+F^{t} = \sigma(\+R^{t})$ and $\+F^{t}_{+} = \sigma(\+R^{t}_{+})$ be the associated (sub) $\sigma$-algebras; note that $\+F^{0} = \emptyset$.

Firstly, \Cref{lem:GBB-independent:chain-rule} establishes a telescoping upper bound on ${\bb E}^{k}[T_{\+A}] - {\bb E}^{0}[T_{\+A}]$, using \textit{Pinsker's inequality} and the \textit{chain rule} (\Cref{prop:pinsker,prop:chain-rule-for-KL}).
% \red{For ease of presentation, we simply write
% $\distKL^{0, k}(\+R^{t}_{+}) = \distKL({\bb P}^{0}_{\+R^{t}_{+}},\ {\bb P}^{k}_{\+R^{t}_{+}})$
% % $\distKL^{0, k}((\SVal{r}, \BFeedback{r}) \;|\; \+F^{r - 1}_{+}) = \distKL({\bb P}^{0}_{(\SVal{r}, \BFeedback{r}) \;|\; \+F^{r - 1}_{+}},\ {\bb P}^{k}_{(\SVal{r}, \BFeedback{r}) \;|\; \+F^{r - 1}_{+}})$;
% % analogous notational conventions extend
% in the rest of this section; likewise for other analogous notations.}

\begin{lemma}[Telescoping Upper Bounds]
\label{lem:GBB-independent:chain-rule}
% \begin{flushleft}
In each possibility $k = 1, 2$:
\begin{align*}
    &\textstyle {\bb E}^{k}[T_{\+A}] - {\bb E}^{0}[T_{\+A}]
    ~\le~ T \cdot \sqrt{\tfrac{1}{2} \sum_{r \in [T]} {\bb E}^{0}[\distKL({\bb P}^{0}_{(\SVal{r}, \BFeedback{r}) \;|\; \+F^{r - 1}_{+}},\ {\bb P}^{k}_{(\SVal{r}, \BFeedback{r}) \;|\; \+F^{r - 1}_{+}})]},
    && \forall \+A \subseteq [0, 1]^{2}.
\end{align*}
% \end{flushleft}
\end{lemma}

\begin{proof}
Since $[\Price{t} \in \+A] \in \+F^{t - 1}_{+} = \sigma(\+R^{t - 1}_{+})$ for every $t \in [T]$, we can deduce that
\begin{align*}
    {\bb E}^{k}[T_{\+A}] - {\bb E}^{0}[T_{\+A}]
    &\textstyle ~=~ {\bb E}^{k}[\{t \in [T] \;|\; \Price{t} \in \+A\}] - {\bb E}^{0}[\{t \in [T] \;|\; \Price{t} \in \+A\}]\\
    \mr{linearity of expectation}
    &\textstyle ~=~ \sum_{t \in [T]} \big({\bb P}^{k}[\Price{t} \in \+A] - {\bb P}^{0}[\Price{t} \in \+A]\big)\\
    \mr{definition of $\distTV(\cdot)$}
    &\textstyle ~\le~ \sum_{t \in [T]} \distTV({\bb P}^{0}_{\+R^{t - 1}_{+}},\ {\bb P}^{k}_{\+R^{t - 1}_{+}})\\
    \mr{Pinsker's inequality (\Cref{prop:pinsker})}
    &\textstyle ~\le~ \sum_{t \in [T]} \sqrt{\frac{1}{2} \distKL({\bb P}^{0}_{\+R^{t - 1}_{+}},\ {\bb P}^{k}_{\+R^{t - 1}_{+}})}\\
    &\textstyle ~\le~ T \cdot \sqrt{\frac{1}{2} \distKL({\bb P}^{0}_{\+R^{T}_{+}},\ {\bb P}^{k}_{\+R^{T}_{+}})}.
\end{align*}
Regarding every term $\distKL({\bb P}^{0}_{\+R^{T}_{+}},\ {\bb P}^{k}_{\+R^{T}_{+}})$ for $t \in [T]$, by the chain rule for KL divergences (\Cref{prop:chain-rule-for-KL}), we have
\begin{align*}
    \distKL({\bb P}^{0}_{\+R^{T}_{+}},\ {\bb P}^{k}_{\+R^{T}_{+}})
    &\textstyle ~=~ \sum_{r \in [T]} \Big(\underbrace{{\bb E}^{0}[\distKL({\bb P}^{0}_{\Price{r} \;|\; \+F^{r - 1}},\ {\bb P}^{k}_{\Price{r} \;|\; \+F^{r - 1}})]}_{= 0}\\
    &\textstyle \phantom{~=~ \sum_{r \in [T]} \Big(} ~+~ {\bb E}^{0}[\distKL({\bb P}^{0}_{(\SVal{r}, \BFeedback{r}) \;|\; \+F^{r - 1}_{+}},\ {\bb P}^{k}_{(\SVal{r}, \BFeedback{r}) \;|\; \+F^{r - 1}_{+}})]\Big)\\
    &\textstyle \phantom{~=~ \sum_{r \in [T]} \Big(} ~+~ \underbrace{{\bb E}^{0}[\distKL({\bb P}^{0}_{\Price{1} \;|\; \+F^{0} = \emptyset},\ {\bb P}^{k}_{\Price{1} \;|\; \+F^{0} = \emptyset})]}_{= 0}\\
    &\textstyle ~=~ \sum_{r \in [T]} {\bb E}^{0}[\distKL({\bb P}^{0}_{(\SVal{r}, \BFeedback{r}) \;|\; \+F^{r - 1}_{+}},\ {\bb P}^{k}_{(\SVal{r}, \BFeedback{r}) \;|\; \+F^{r - 1}_{+}})].
\end{align*}
Here the last step holds because, for $t \in [T]$, the base/hard instances $\+D^{0}$ and $\+D^{k}$ conditioned on the same record $\+F^{r - 1}$ induce two identically distributed actions $\Price{r}$.

% the action taken $\Price{r}$ on both possibilities  must be identically distributed.

Combining everything together finishes the proof of \Cref{lem:GBB-independent:chain-rule}.
\end{proof}

Secondly, \Cref{lem:GBB-independent:bound-for-single-KL} upper bounds the above KL divergences as desired, through a careful analysis of our lower-bound construction.

\begin{lemma}[Upper Bounds on KL Divergences]
\label{lem:GBB-independent:bound-for-single-KL}
\begin{flushleft}
In each possibility $k = 1, 2$:
\begin{align*}
    \textstyle
    \sum_{r \in [T]} {\bb E}^{0}[\distKL({\bb P}^{0}_{(\SVal{r}, \BFeedback{r}) \;|\; \+F^{r - 1}_{+}},\ {\bb P}^{k}_{(\SVal{r}, \BFeedback{r}) \;|\; \+F^{r - 1}_{+}})]
    ~\le~ \frac{1}{2}\delta^{2} \cdot {\bb E}^{0}[T_{\+I}].
\end{align*}
\end{flushleft}
\end{lemma}

\begin{proof}
The seller's value $\SVal{r}$ cannot help in distinguishing $\+D^{0}$ and $\+D^{k}$, since it is independent of the buyer's value $\BVal{r}$ and intention to trade $\BFeedback{r} = {\bb 1}[\BPrice{t} \le \BVal{t}]$, and it is identically distributed on both instances $\+D^{0}$ and $\+D^{k}$.
By applying the chain rule (\Cref{prop:chain-rule-for-KL}) again, we can reformulate each index-$(r \in [T])$ term as follows:\footnote{\label{footnote:abuse_notation}Here, we slightly abuse the notation---the integration is with respect to the pushforward measure induced by the considered fixed-price mechanism $\Mech = \Price{t}_{t \in [T]}$ on the base instance $\+D^{0}$.}
\begin{align*}
    & {\bb E}^{0}[\distKL({\bb P}^{0}_{(\SVal{r}, \BFeedback{r}) \;|\; \+F^{r - 1}_{+}},\ {\bb P}^{k}_{(\SVal{r}, \BFeedback{r}) \;|\; \+F^{r - 1}_{+}})]\\
    &\textstyle ~=~ \iint_{(p, q) \in [0, 1]^{2}} \distKL({\bb P}^{0}_{(\SVal{r}, \BFeedback{r}) \;|\; \Price{r} = (p, q)},\ {\bb P}^{k}_{(\SVal{r}, \BFeedback{r}) \;|\; \Price{r} = (p, q)})
    \cdot {\bb P}^{0}[\Price{r} = (\dd p, \dd q)]\\
    &\textstyle ~=~ \iint_{(p, q) \in [0, 1]^{2}} \int_{s \in [0, 1]} \underbrace{\distKL({\bb P}^{0}_{\BFeedback{r} \;|\; (\SVal{r}, \SPrice{r}, \BPrice{r}) = (s, p, q)},\ {\bb P}^{k}_{\BFeedback{r} \;|\; (\SVal{r}, \SPrice{r}, \BPrice{r}) = (s, p, q)})}_{(\dagger)}\\
    &\textstyle \phantom{~=~ \iint_{(p, q) \in [0, 1]^{2}} \int_{s \in [0, 1]}} \quad \cdot f_{S}^{0}(s) \cdot \dd s \cdot {\bb P}^{0}[\Price{r} = (\dd p, \dd q)].
\end{align*}
We can reason about the term $(\dagger)$ through case analysis:
\begin{itemize}
    \item \textbf{Case~1: $(p, q) \notin \+I = [0, 1] \times [1 - 2\theta, 1]$.}
    For such a \textit{non-informative action}, we must have $(\dagger) = 0$.
    
    \item \textbf{Case~2: $(p, q) \in \+I = [0, 1] \times [1 - 2\theta, 1]$.}
    For such an \textit{informative action}, we can establish $(\dagger) \le \frac{1}{2}\delta^{2}$ based on \Cref{lem:KL-Bernoulli} and that $\delta \in [0, \frac{1}{2}]$---upper bounds on KL divergences of Bernoulli distributions:\\
    $(\dagger) = \distKL\big(\Bern(\frac{1 - q}{4\theta}),\ \Bern((1 \pm \delta) \cdot \frac{1 - q}{4\theta})\big)
    \le \frac{1 - q}{4\theta} \cdot 2\delta^{2}
    \le \frac{1}{2}\delta^{2}$,\hfill when $q \in [1 - \theta, 1]$,\\
    $(\dagger) = \distKL\big(\Bern(\frac{1 - q}{4\theta}),\ \Bern(\frac{1 - q}{4\theta} \pm (\frac{1}{2} - \frac{1 - q}{4\theta}) \cdot \delta)\big)
    \le \frac{(\frac{1}{2} - \frac{1 - q}{4\theta})^{2}}{\frac{1 - q}{4\theta}} \cdot 2\delta^{2}
    \le \frac{1}{2}\delta^{2}$,\hfill when $q \in [1 - 2\theta, 1 - \theta]$.
\end{itemize}
Putting everything together gives
\begin{align*}
    &\textstyle \sum_{r \in [T]} {\bb E}^{0}[\distKL({\bb P}^{0}_{(\SVal{r}, \BFeedback{r}) \;|\; \+F^{r - 1}_{+}},\ {\bb P}^{k}_{(\SVal{r}, \BFeedback{r}) \;|\; \+F^{r - 1}_{+}})]\\
    &\textstyle ~\le~ \sum_{r \in [T]} \iint_{(p, q) \in \+I} \int_{s \in [0, 1]} \frac{1}{2}\delta^{2} \cdot f_{S}^{0}(s) \cdot \dd s \cdot {\bb P}^{0}[\Price{r} = (\dd p, \dd q)]\\
    &\textstyle ~=~ \frac{1}{2}\delta^{2} \cdot \sum_{r \in [T]} \iint_{(p, q) \in \+I} {\bb P}^{0}[\Price{r} = (\dd p, \dd q)]\\
    &\textstyle ~=~ \frac{1}{2}\delta^{2} \cdot {\bb E}^{0}[T_{\+I}].
\end{align*}
This completes the proof of \Cref{lem:GBB-independent:bound-for-single-KL}.
\end{proof}

Finally, we are ready to accomplish \Cref{thm:GBB-independent:LB}.

\begin{proof}[Proof of \Cref{thm:GBB-independent:LB}]
For the base instance $\+D^{0}$, each take of an informative action in $\+I$ incurs $\frac{3 - 11\theta}{16}$ regret (\Cref{prop:GBB-independent:regret-per-round}), so the total regret is at least
\begin{align*}
    \textstyle
    \Regret_{\+D^{0}}
    ~\ge~ \frac{3 - 11\theta}{16} \cdot {\bb E}^{0}[T_{\+I}].
\end{align*}
For each hard instance $\+D^{k}$, $\forall k = 1, 2$, we can also lower-bound the total regret (\Cref{prop:GBB-independent:regret-per-round}):
\begin{align*}
    \textstyle \Regret_{\+D^{k}}
    ~\ge~ \frac{\theta}{16}\delta \cdot \big(T - {\bb E}^{1}[T_{\+G'_{1}}]\big).
\end{align*}
Taking the average of $\Regret_{\+D^{k}}$ for $k = 1, 2$ gives
\begin{align*}
    \textstyle \frac{1}{2} \cdot \sum_{k = 1, 2} \Regret_{\+D^{k}}
    &\textstyle ~\ge~ \frac{\theta}{32}\delta \cdot \big(2T - {\bb E}^{1}[T_{\+G'_{1}}] - {\bb E}^{2}[T_{\+G'_{2}}]\big)\\
    \mr{\Cref{lem:GBB-independent:similarity-in-T_S}}
    &\textstyle ~\ge~ \frac{\theta}{32}\delta \cdot \big(2T - {\bb E}^{0}[T_{\+G'_{1}}] - {\bb E}^{0}[T_{\+G'_{2}}] - \frac{1}{2}\delta T \cdot \sqrt{{\bb E}^{0}[T_{\+I}]}\big)\\
    \mr{$T_{\+G'_{1}} + T_{\+G'_{2}} \le T$ (almost surely)}
    &\textstyle ~\ge~ \frac{\theta}{32}\delta T \cdot \big(1 - \frac{1}{2}\delta \cdot \sqrt{{\bb E}^{0}[T_{\+I}]}\big).
\end{align*}
Here the last step uses $T_{\+G'_{1}} + T_{\+G'_{2}} \le T$, as a consequence of that $\+G'_{1}$ and $\+G'_{2}$ are disjoint.

Plugging in $\theta = \frac{1}{13}$ and $\delta = 4T^{-1 / 3}$, we obtain
\begin{align*}
    \textstyle
    \max \big\{\Regret_{\+D^{0}},\ \frac{1}{2} \cdot \sum_{k = 1, 2} \Regret_{\+D^{k}}\big\}
    &\textstyle ~\ge~ \max\big\{\frac{7}{52} {\bb E}^{0}[T_{\+I}], \quad \frac{1}{104}T^{2 / 3} \cdot \big(1 - 2T^{-1 / 3} \cdot \sqrt{{\bb E}^{0}[T_{\+I}]}\big)\big\}\\
    &\textstyle ~\ge~ \frac{1}{177}T^{2 / 3}.
\end{align*}
Here the last step can be verified through elementary algebra.

In sum, the considered fixed-price mechanism $\Mech$ incurs $\Omega(T^{2 / 3})$ regret in at least one possibility $k \in [0 : 2]$. Then, the arbitrariness of $\Mech$ implies \Cref{thm:GBB-independent:LB}.
\end{proof}

\section{\texorpdfstring{$\Omega(T^{3 / 4})$}{} {\GBB} Partial-Feedback Lower Bound for Correlated Values}
\label{sec:LB-GBB-Partial-Correlated}

In this section, we study the limit of ``fixed-price mechanisms with the {\GlobalBudgetBalance} ({\GBB}) constraint and partial feedback'' in the ``correlated values'' setting.
Specifically, we will establish (\Cref{thm:GBB-correlated:LB}) the following hardness result.
By implication, the same lower bound extends to the \textit{more general} ``adversarial values'' setting.

\begin{theorem}[{\GBB} Partial-Feedback Lower Bound for Correlated Values]
\label{thm:GBB-correlated:LB}
\begin{flushleft}
In the ``correlated values'' setting (with or without the density-boundedness assumption---\Cref{asm:density} with parameter $M = 224$),\\
every ``{\GBB} fixed-price mechanism with two-bit feedback'' has worst-case regret $\Omega(T^{3 / 4})$.
\end{flushleft}
\end{theorem}

\noindent
Previously, in both the ``correlated values'' setting and the ``adversarial values'' setting, merely an $\tO(T^{3 / 4})$ upper bound \cite[Theorem~5.4]{BCCF24} and an unmatching $\Omega(T^{5 / 7})$ lower bound \cite[Theorem~5.5]{BCCF24} were known.
Nonetheless, our hardness result closes this gap by (up to polylogarithmic factors) establishing a matching $\Omega(T^{3 / 4})$ lower bound.

% In this section, we investigate the limit of ``fixed-price mechanisms the \textit{{\GlobalBudgetBalance} ({\GBB})}'' in the following $2 \times 2 = 4$ settings:
% \begin{center}
%     \textit{``adversarial/correlated values, two-bit/one-bit feedback''.}
% \end{center}

For ease of presentation, we first establish the ``discrete values'' version of \Cref{thm:GBB-correlated:LB} in \Cref{sec:LB-GBB-Partial-Correlated:construction,sec:LB-GBB-Partial-Correlated:analysis} and then extend it to the ``density-bounded values'' version in \Cref{sec:LB-GBB-Partial-Correlated:density}.

\subsection{Lower-Bound Construction for Discrete Values}
\label{sec:LB-GBB-Partial-Correlated:construction}

Our construction utilizes two parameters $K = \Theta(T^{1 / 4})$ and $\delta = \Theta(T^{-1 / 4})$, as follows.
\begin{align*}
    K & ~\defeq~ T^{1 / 4}, \\
    \delta & ~\defeq~ \tfrac{0.1}{5K + 2}.
\end{align*}

\noindent
\textbf{The Value Support.}
We will construct $(K + 1)$ base/hard instances $\{\+D^{k}\}_{k \in [0 : K]}$, which have a common \textit{discrete} support $\VAL \subseteq [0, 1]^{2}$, including four types of points that serve different purposes---a size-$(3K + 1)$ ``upper-left'' subset $\VALupperleft$,
a size-$(2K + 1)$ ``lower-right'' subset $\VALlowerright$,
four ``corner'' points $\VALcorner = \{0, 1\}^{2}$,
and one ``majority'' point $\valmajority = (0.4, 0.6)$; see \Cref{fig:GBB-correlated:1} for a diagram.
\begin{align*}
    \VAL
    & ~\defeq~ \VALupperleft \cup \VALlowerright \cup \VALcorner \cup \{\valmajority\}, \\
    \VALupperleft
    &\textstyle ~\defeq~ \{(\frac{k}{5K},\ 0.4 + \frac{k}{5K})\}_{k \in [0 : 3K]}, \\
    \VALlowerright
    &\textstyle ~\defeq~ \{(0.2 + \frac{k}{5K},\ 0.4 + \frac{k}{5K})\}_{k \in [0 : 2K]}, \\
    \VALcorner
    & ~\defeq~ \{0, 1\}^{2}, \\
    \valmajority
    & ~\defeq~ (0.4, 0.6).
\end{align*}
Note that the majority point $\valmajority$ is also the index-$K$ point in the lower-right subset $\VALlowerright[K] = (0.4, 0.6)$; a base/hard instance $\+D^{k}$ assigns two probability masses to this point, one for $\valmajority$ and one for $\VALlowerright[k]$. However, we can safely treat it as \textit{two} isolated points (or, interchangeably, treat it as \textit{one} point by adding both probability masses together).
All other points $v \in \VAL \setminus \{\valmajority\}$ are isolated.

\begin{figure}[t]
\centering
\tikzset{every picture/.style={line width = 0.75pt}} %set default line width to 0.75pt        

\begin{tikzpicture}[x = 2pt, y = 2pt, scale = 1.5]
% corners
\draw (0, 0) node[anchor = 0] {$(0, 0)$};
\draw (0, 100) node[anchor = 0] {$(0, 1)$};
\draw (100, 0) node[anchor = 180] {$(1, 0)$};
\draw (100, 100) node[anchor = 180] {$(1, 1)$};

\draw (0,0) -- (0,100) -- (100,100) -- (100,0) -- cycle;
\draw (50, -3) node [below][inner sep=0.75pt] {seller};
\draw (-3, 50) node [above][inner sep = 0.75pt,rotate=90] {buyer};

% diagonal
\draw (0, 0) -- (100, 100);

% boxed line
\foreach \y in {5, 10, 15, 20, 25, 30, 35, 40, 45, 50, 55, 60, 65, 70, 75, 80, 85, 90, 95} {
    \draw[dotted, draw opacity = 0.5]  (\y, 0) -- (\y, 100);
    \draw [dotted, draw opacity = 0.5]  (0, \y) -- (100, \y);
}

% good action
\fill[green, fill opacity = 0.2] (40, 40) -- (60, 40) -- (60, 60) -- (40, 60) -- cycle;

% bad action
\fill[red, fill opacity = 0.2] (0, 40) -- (40, 40) -- (40, 60) -- (60, 60) -- (60, 100) -- (0, 100) -- cycle;

% value'
\foreach \x in {20, 25, 30, 35, 45, 50, 55, 60} {
    \draw[black, fill = white] (\x, \x + 20) circle (2pt);
}
\draw[black, fill = white] (40 + 0.7071, 60 + 0.7071) arc (45:225:2pt);

% value''
\foreach \x in {0, 5, 10, 15, 20, 25, 30, 35, 40, 45, 50, 55, 60} {
    \draw[black, fill = black] (\x, \x + 40) circle (2pt);
}

% value majority
\draw[black, fill = SkyBlue] (40 - 0.7071, 60 - 0.7071) arc (-135:45:2pt);
\draw[black] (40 - 0.7071, 60 - 0.7071) -- (40 + 0.7071, 60 + 0.7071);

% value corner
\draw[black, fill = BrickRed] (0, 0) circle (2pt);
\draw[black, fill = BrickRed] (0, 100) circle (2pt);
\draw[black, fill = BrickRed] (100, 0) circle (2pt);
\draw[black, fill = BrickRed] (100, 100) circle (2pt);

%legends
\draw[black, fill = black] (110, 65) circle (2pt) node[right] {$\VALupperleft$};
\draw[black, fill = white] (110, 55) circle (2pt) node[right] {$\VALlowerright$};
\draw[black, fill = BrickRed] (110, 45) circle (2pt) node[right] {$\VALcorner = \set{0, 1}^2$};
\draw[black, fill = SkyBlue] (110, 35) circle (2pt) node[right] {$\valmajority = (0.4, 0.6)$};
\end{tikzpicture}
\caption{A diagram of the support $\VAL = \VALupperleft \cup \VALlowerright \cup \VALcorner \cup \{\valmajority\} \subseteq [0, 1]^{2}$.}
\label{fig:GBB-correlated:1}
\end{figure}

\vspace{.1in}
\noindent
\textbf{The Base/Hard Instances.}
Among the $K + 1$ instances, $\+D^{0}$ is the base instance and is given as follows. 
\begin{align*}
    {\bb P}_{{\Val{} \sim \+D^{0}}}[\Val{} = v]
    ~\defeq~
    \begin{cases}
        \delta, & v \in \VALupperleft \cup \VALlowerright \\
        0.1, & v \in \VALcorner \\
        0.5, & v = \valmajority
    \end{cases}.
\end{align*}
This $\+D^{0}$ is a well-defined distribution, given that $\delta \cdot |\VALupperleft \cup \VALlowerright| + 0.1 \cdot |\VALcorner| + 0.5 \cdot |\{\valmajority\}| = 1$.

In contrast, each $\+D^{k}$ for $k \in [K]$ is a hard instance and tweaks the probability masses at two upper-left points $\VALupperleft[k - 1],\ \VALupperleft[k]$ and two lower-right points $\VALlowerright[k - 1],\ \VALlowerright[k]$ (but otherwise is identical to the base instance $\+D^{0}$); see \Cref{fig:GBB-correlated-informativeLine} for a diagram.
Again, this $\+D^{k}$ is a well-defined distribution.
\begin{align*}
    {\bb P}_{{\Val{} \sim \+D^{k}}}[\Val{} = v]
    ~\defeq~
    \begin{cases}
        {\bb P}_{{\Val{} \sim \+D^{0}}}[\Val{} = v] + \delta,
        & v \in \{\VALupperleft[k],\ \VALlowerright[k - 1]\} \\
        {\bb P}_{{\Val{} \sim \+D^{0}}}[\Val{} = v] - \delta,
        & v \in \{\VALupperleft[k - 1],\ \VALlowerright[k]\} \\
        {\bb P}_{{\Val{} \sim \+D^{0}}}[\Val{} = v],
        & \text{otherwise}
    \end{cases}.
\end{align*}

\subsection{Lower-Bound Analysis for Discrete Values}
\label{sec:LB-GBB-Partial-Correlated:analysis}

For the base/hard instances $\{\+D^{k}\}_{k \in [0 : K]}$ above, we will establish an $\Omega(T^{3 / 4})$ lower bound in three steps:
\begin{itemize}
    \item Firstly, we show that a mechanism $\Mech$'s actions $\Price{t}_{t \in [T]}$, without loss of generality, can be restricted to a particular discrete set $\ACTION$ (or, more precisely, a $(3K + 1) \times (3K + 1)$ grid).
    
    \item Secondly, we show that the {\GlobalBudgetBalance} constraint can be relaxed to another constraint (which we call \blackref{eq:GlobalPriceBalance}); the $\Omega(T^{3 / 4})$ lower bound holds even after this relaxation.
    
    \item Thirdly, we adapt arguments in \Cref{sec:GBB-independent:LB} to accomplish the entire lower-bound analysis.
\end{itemize}

% \red{
% Our proof for the lower bound begins by first ``preprocessing'' the family of mechanisms we will consider. Interestingly, the first step narrows down while the second step broadens the set of mechanisms in consideration. Specifically,
% \begin{itemize}
%     \item 
    
%     \item we then show that we can relax the {\GBB} constraint for the mechanisms to a weaker constraint which is easier to deal with. This only strengthens the lower bound since it applies to broader class of mechanisms.
% \end{itemize}}

% \red{Recall that for a set $S$ of actions, we use $T_S$ to denote the number of rounds playing actions in $T$.}

\subsection*{Step~1: Discretization of Actions}

\newcommand{\BarMech}{\Bar{\Mech}}
\newcommand{\BarVal}[1]{(\Bar{S}^{#1}, \Bar{B}^{#1})}
\newcommand{\BarSVal}[1]{\Bar{S}^{#1}}
\newcommand{\BarBVal}[1]{\Bar{B}^{#1}}
\newcommand{\BarPrice}[1]{(\Bar{P}^{#1}, \Bar{Q}^{#1})}
\newcommand{\BarSPrice}[1]{\Bar{P}^{#1}}
\newcommand{\BarBPrice}[1]{\Bar{Q}^{#1}}
\newcommand{\BarFeedback}[1]{(\Bar{X}^{#1}, \Bar{Y}^{#1})}
\newcommand{\BarSFeedback}[1]{\Bar{X}^{#1}}
\newcommand{\BarBFeedback}[1]{\Bar{Y}^{#1}}

\begin{algorithm}[t]
\caption{\label{alg:GBB-correlated:discretization}
Discretization of Actions (\Cref{lem:GBB-correlated:discretization})}

\Input A (generic) fixed-price mechanism $\BarMech = \BarPrice{t}_{t \in [T]}$.

\Output A new fixed-price mechanism $\Mech = \Price{t}_{t \in [T]}$ with actions restricted to $\ACTION$.

\begin{algorithmic}[1]
\For{every round $t = 1, 2, \dots, T$}
    \State $\BarPrice{t} \gets \BarMech$
    \label{alg:GBB-correlated:discretization:old-action}
    
    \State $\Price{t} \gets (\min(\frac{\lfloor \BarSPrice{t} \cdot 5K \rfloor}{5K}, 0.6),\ \min(\frac{\lceil \BarBPrice{t} \cdot 5K \rceil}{5K}, 0.4))$
    \label{alg:GBB-correlated:discretization:new-action}
    
    % \State $\Price{t} \gets (\max(\projection_{S}(\VAL) \setminus \{1\} \setminus (\BarSPrice{t}, 1]),\ \min(\projection_{B}(\VAL) \setminus \{0\} \setminus [0, \BarBPrice{t})))$
    
    \State $\Feedback{t} \gets ({\bb 1}[\SVal{t} \le \SPrice{t}],\ {\bb 1}[\BPrice{t} \le \BVal{t}])$
    \Comment{$\Val{t} \sim \+D$.}
    \label{alg:GBB-correlated:discretization:new-feedback}
    
    \State $\BarFeedback{t} \gets (\SFeedback{t} \lor {\bb 1}[\BarSPrice{t} = 1],\ \BFeedback{t} \lor {\bb 1}[\BarBPrice{t} = 0])$
    \label{alg:GBB-correlated:discretization:old-feedback}
    
    \State Feed $\BarMech$ the two-bit feedback $\BarFeedback{t}$
\EndFor
\end{algorithmic}
\end{algorithm}

It turns out that a \textit{regret-optimal} mechanism $\Mech = \Price{t}_{t \in [T]}$ can take actions $\Price{t}$ only from $\ACTION$, the Cartesian product of the coordinate projections---excluding the ``trivial'' seller value of $1$ and the ``trivial'' buyer value of $0$---of the support $\VAL \subseteq [0, 1]^{2}$; cf.\ the union of the \textit{red/green regions} in \Cref{fig:GBB-correlated:1}.
\begin{align*}
    \ACTION
    & \textstyle ~\defeq~ (\projection_{S}(\VAL) \setminus \{1\}) \times (\projection_{B}(\VAL) \setminus \{0\})\\
    & \textstyle ~\:=~ \{\frac{k}{5K}\}_{k \in [0 : 3K]} \times \{0.4 + \frac{k}{5K}\}_{k \in [0 : 3K]}.
\end{align*}

\begin{lemma}[Discretization of Actions]
\label{lem:GBB-correlated:discretization}
% \begin{flushleft}
A (generic) fixed-price mechanism $\BarMech = \BarPrice{t}_{t \in [T]}$ can transform into a new fixed-price mechanism $\Mech = \Price{t}_{t \in [T]}$ such that, in any possibility $\+D = \+D^{k}$ for $k \in [0 : K]$:
\begin{itemize}
    \item $\Price{t}_{t \in [T]} \subseteq \ACTION$, almost surely.
    
    \item On the same realization $\Val{t}_{t \in [T]} \sim \+D^{k}$, almost surely over the randomness of $\BarMech$ and $\Mech$,\\
    both the {\GainsFromTrade} and the profit in every round $t \in [T]$ can only increase:
    \begin{align*}
        \GFT(\SVal{t}, \BVal{t}, \SPrice{t}, \BPrice{t})
        & ~\ge~ \GFT(\SVal{t}, \BVal{t}, \BarSPrice{t}, \BarBPrice{t}),
        && \forall t \in [T],\\
        \Profit(\SVal{t}, \BVal{t}, \SPrice{t}, \BPrice{t})
        & ~\ge~ \Profit(\SVal{t}, \BVal{t}, \BarSPrice{t}, \BarBPrice{t}),
        && \forall t \in [T].
    \end{align*}
    
    % \red{the total regret can only decrease $\Regret_{\+D} \le \Bar{\Regret}_{\+D}$;}
    
    % \item \red{$\Mech$ satisfies the {\GBB} constraint, whenever so does $\BarMech$.}
\end{itemize}
% \end{flushleft}
\end{lemma}

\begin{proof}
\Cref{alg:GBB-correlated:discretization} explicitly transforms a given mechanism $\BarMech$ into a new one $\Mech$.

In every round $t \in [T]$, the new action $\Price{t} = (\min(\frac{\lfloor \BarSPrice{t} \cdot 5K \rfloor}{5K}, 0.6),\ \min(\frac{\lceil \BarBPrice{t} \cdot 5K \rceil}{5K}, 0.4))$ by construction (Line~\ref{alg:GBB-correlated:discretization:new-action}) satisfies the first property $\Price{t} \in \ACTION$.

On the same realization $\Val{t} \sim \+D^{k}$, actions $\BarPrice{t}$ and $\Price{t}$ induce the same trade outcome---either both successes or both failures---except for three cases where $\Val{t} \in \{(0, 0), (1, 0), (1, 1)\}$. But in each of these three cases, because $\SVal{t} \ge \BVal{t}$, both {\GainsFromTrade} and profit can only be \textit{non-positive}.
Moreover, $\BarFeedback{t} = (\SFeedback{t} \lor {\bb 1}[\BarSPrice{t} = 1],\ \BFeedback{t} \lor {\bb 1}[\BarBPrice{t} = 0])$ by construction (Line~\ref{alg:GBB-correlated:discretization:old-feedback}) identically simulates $\BarMech$'s two-bit feedback $\BarFeedback{t} \equiv ({\bb 1}[\SVal{t} \le \BarSPrice{t}],\ {\bb 1}[\BPrice{t} \le \BarBVal{t}])$ on the same realization $\Val{t}$.
Given these, we can infer the second property from a simple coupling argument.
This finishes the proof of \Cref{lem:GBB-correlated:discretization}.
\end{proof}

In the rest of \Cref{sec:LB-GBB-Partial-Correlated}, we safely restrict actions to the discrete set $\ACTION$.
To emphasize this, we rewrite $\Mech^{\ACTION}$ for a mechanism $\Mech$.
To proceed, we divide $\ACTION$ into a ``good'' subset $\+G$ (cf.\ the \textit{green region} in \Cref{fig:GBB-correlated:1}) and a ``bad'' subset $\+B$ (cf.\ the \textit{red region} in \Cref{fig:GBB-correlated:1}), as follows.
\begin{align*}
    \textstyle
    \+G & ~\defeq~ \cup_{k \in [0 : K]} \+G_{k},\\
    \+G^{k} &\textstyle ~\defeq~ \{(0.4 + \frac{i}{5K},\ 0.4 + \frac{k}{5K})\}_{i \in [0 : K]},
    && \forall k \in [0 : K],\\
    \+B & ~\defeq~ \ACTION \setminus \+G.
\end{align*}

Let us denote by $\Regret_{\+D}(p, q)$ the regret incurred by each take of an action $(p, q) \in \ACTION = \+G \cup \+B$.
The following \Cref{prop:GBB-correlated:per-round-regret} lower-bounds $\Regret_{\+D}(p, q)$ in each possibility $\+D = \+D^{k}$ for $k \in [0 : K]$; its proof is simply elementary algebra and thus is omitted for brevity.

\begin{proposition}[{\GainsFromTrade}]
\label{prop:GBB-correlated:per-round-regret}
For the base instance $\+D^{0}$ and the hard instances $\+D^{k}$, $\forall k \in [K]$:
\begin{align*}
    \Regret_{\+D^{0}}(p, q)
    & ~\ge~
    \begin{cases}
        0.1, & \forall (p, q) \in \+B\\
        3\delta K \cdot (q - p), & \forall (p, q) \in \+G
    \end{cases},\\
    \Regret_{\+D^{k}}(p, q)
    & ~\ge~
    \begin{cases}
        0.1, & \forall (p, q) \in \+B\\
        3\delta K \cdot (q - p) + 0.2\delta \cdot {\bb 1}[(p, q) \notin \+G^{k}], & \forall (p, q) \in \+G
    \end{cases}.
\end{align*}
\end{proposition}

\subsection*{Step~2: Relaxation of {\GlobalBudgetBalance}}

We now show that, under our lower-bound construction, a {\GBB} mechanism $\Mech^{\ACTION}$ always satisfies the following ``\blackref{eq:GlobalPriceBalance}'' condition; thus, it serves as a relaxation of the {\GBB} constraint.

\begin{lemma}[Relaxation of {\GBB}]
\label{lem:GBB-correlated:PGBB}
For a {\GBB} fixed-price mechanism $\Mech^{\ACTION}$, in each possibility $k \in [0 : K]$:
\begin{align*}
    \textstyle
    {\bb E}^{k}[\sum_{t \in [T]} (\BPrice{t} - \SPrice{t})] ~\ge~ 0.
    \tag{\textsf{Global Price Balance}}
    \label{eq:GlobalPriceBalance}
\end{align*}
\end{lemma}

\begin{proof}
The {\GBB} constraint requires that, in each possibility $k \in [0 : K]$, the following holds almost surely:
\begin{align*}
    \textstyle
    \sum_{t \in [T]} (\BPrice{t} - \SPrice{t}) \cdot \Trade{t}
    ~\equiv~ \sum_{t \in [T]} (\BPrice{t} - \SPrice{t}) \cdot {\bb 1}[\SVal{t} \le \SPrice{t} \land \BPrice{t} \le \BVal{t}]
    ~\ge~ 0
    \tag{\GlobalBudgetBalance}
\end{align*}
For the sake of contradiction, we assume that ${\bb E}^{k}[\sum_{t \in [T]}(\BPrice{t} - \SPrice{t})] < 0$, for some base/hard instance $\+D^{k}$, $k \in [0 : K]$.
Recall that $\Val{t} \sim \+D^{k}$ for $t \in [T]$ are i.i.d. Our construction of $\+D^{k}$ guarantees that:\\
(i)~Under the particular action $(0.5, 0.5) \in \ACTION$, the trade in a single round succeeds with a nonzero probability $\alpha^{k} \defeq {\bb P}^{0}[\Trade{t} = 1 \;|\; \Price{t} = (0.5, 0.5)] > 0$.\\
(ii)~Under a generic action $(p, q) \in \ACTION$, the success probability admits the following bounds:\\
$p \le q \implies {\bb P}^{k}[\Trade{t} = 1 \;|\; \Price{t} = (p, q)] \le \alpha^{k}$
\hfill and \hfill
$p > q \implies {\bb P}^{k}[\Trade{t} = 1 \;|\; \Price{t} = (p, q)] \ge \alpha^{k}$.\\
For these reasons, we can deduce that
\begin{align*}
    & \textstyle
    {\bb E}^{k}\big[\sum_{t \in [T]} (\BPrice{t} - \SPrice{t}) \cdot {\bb 1}[\SVal{t} \le \SPrice{t} \land \BPrice{t} \le \BVal{t}]\big]\\
    \mr{linearity of expectation}
    &\textstyle ~=~ \sum_{t \in [T]} {\bb E}^{k}\big[(\BPrice{t} - \SPrice{t}) \cdot {\bb P}^{k}[\SVal{t} \le \SPrice{t} \land \BPrice{t} \le \BVal{t} \;|\; \Price{t}]\big]\\
    &\textstyle ~\le~ \sum_{t \in [T]} {\bb E}^{k}\big[(\BPrice{t} - \SPrice{t}) \cdot \alpha^{k}\big]\\
    \mr{linearity of expectation}
    &\textstyle ~=~ \alpha^{k} \cdot {\bb E}^{k}\big[\sum_{t \in [T]} (\BPrice{t} - \SPrice{t})\big]\\
    % \mr{$\alpha^{k} > 0$ and ${\bb E}^{k}[\sum_{t \in [T]}(\BPrice{t} - \SPrice{t})] < 0$}
    & ~<~ 0.
\end{align*}
This contradicts (even the relaxed ``in expectation'' version of) the {\GBB} constraint.
Therefore, refuting our assumption implies \Cref{lem:GBB-correlated:PGBB}.
\end{proof}

\subsection*{Step~3: Adaptation of Arguments in \Cref{sec:GBB-independent:LB}}

Similar to \Cref{sec:GBB-independent:LB}, for every $t\in [0\colon T]$, we consider two random sequences $\+R^{t}$ and $\+R^{t}_{+}$ and the associated (sub) $\sigma$-algebras $\+F^{t} = \sigma(\+R^{t})$ and $\+F^{t}_{+} = \sigma(\+R^{t}_{+})$.\textsuperscript{\ref{footnote:record}}
Note that $\+R^{0} = \emptyset$ and $\+F^{0} = \emptyset$.
\begin{align*}
    \+R^{t}
    &\textstyle ~\defeq~ (\SPrice{r}, \BPrice{r}, \SFeedback{r}, \BFeedback{r})_{r \in [t]},\\
    \+R^{t}_{+}
    & ~\defeq~ \+R^{t} \oplus \Price{t + 1}.
\end{align*}
Regarding each hard instance $\+D^{k}$ for $k \in [K]$, we define its \textit{informative action subset} $\+I^{k} \subseteq \+B$ (cf.\ \Cref{fig:GBB-correlated-informativeLine})---only taking such actions can help in distinguishing this $\+D^{k}$ from the other base/hard instances---including
one ``horizontal'' subset $\+I^{k}_{\mathrm{hor}}$,
one ``lower-left'' subset $\+I^{k}_{\mathrm{LL}}$,
one ``upper-left'' subset $\+I^{k}_{\mathrm{UL}}$,
one ``lower-right'' subset $\+I^{k}_{\mathrm{LR}}$, and
one ``upper-right'' subset $\+I^{k}_{\mathrm{UR}}$.
\begin{align*}
    \+I^{k}
    &\textstyle ~\defeq~ \+I^{k}_{\mathrm{hor}} \cup \+I^{k}_{\mathrm{LL}} \cup \+I^{k}_{\mathrm{UL}} \cup \+I^{k}_{\mathrm{LR}} \cup \+I^{k}_{\mathrm{UR}},\\
    \+I^{k}_{\mathrm{hor}}
    &\textstyle ~\defeq~ \{(\frac{k + i}{5K},\ 0.4 + \frac{k}{5K}) \;|\; i \in [0 : K - 1]\},\\
    \+I^{k}_{\mathrm{LL}}
    &\textstyle ~\defeq~ \{(\frac{k - 1}{5K},\ 0.4 + \frac{i}{5K}) \;|\; i\in [0 : k - 1]\},\\
    \+I^{k}_{\mathrm{UL}}
    &\textstyle ~\defeq~ \{(\frac{k - 1}{5K},\ 0.4 + \frac{i}{5K}) \;|\; i\in [k : 3K]\},\\
    \+I^{k}_{\mathrm{LR}}
    &\textstyle ~\defeq~ \{(0.2 + \frac{k - 1}{5K},\ 0.4 + \frac{i}{5K}) \;|\; i\in [0 : k - 1]\},\\
    \+I^{k}_{\mathrm{UR}}
    &\textstyle ~\defeq~ \{(0.2 + \frac{k - 1}{5K},\ 0.4 + \frac{i}{5K}) \;|\; i\in [k + 1 : 3K]\}.
\end{align*}

\begin{figure}[t]
    \centering
    \tikzset{every picture/.style={line width = 0.75pt}} %set default line width to 0.75pt        

\begin{tikzpicture}[x = 2pt, y = 2pt, scale = 2.5]
% corners
\draw (0, 40) node[anchor = 0] {$(0, 0.4)$};
\draw (0, 100) node[anchor = 0] {$(0, 1)$};
\draw (60, 40) node[anchor = 180] {$(0.6, 0.4)$};
\draw (60, 100) node[anchor = 180] {$(0.6, 1)$};

\draw (0, 40) -- (0, 100) -- (60, 100);
\draw (30, 40-1.8) node [below][inner sep=0.75pt] {seller};
\draw (0-1.8, 70) node [above][inner sep = 0.75pt,rotate=90] {buyer};

% diagonal
\draw (37, 37) -- (63, 63);

% boxed line
\foreach \x in {0, 2.5, 5, 7.5, 10, 12.5, 15, 17.5, 20, 22.5, 25, 27.5, 30, 32.5, 35, 37.5, 40, 42.5, 45, 47.5, 50, 52.5, 55, 57.5, 60} {
    \draw[dotted, draw opacity = 0.5]  (\x, 40) -- (\x, 100);
    \draw [dotted, draw opacity = 0.5]  (0, \x + 40) -- (60, \x + 40);
}

% good action
\fill[green, fill opacity = 0.2] (40, 40) -- (60, 40) -- (60, 60) -- (40, 60) -- cycle;

% bad action
\fill[red, fill opacity = 0.2] (0, 40) -- (40, 40) -- (40, 60) -- (60, 60) -- (60, 100) -- (0, 100) -- cycle;

% value'
\foreach \x in {20, 22.5, 25, 27.5, 30, 32.5, 35, 37.5, 42.5, 45, 47.5, 50, 52.5, 55, 57.5, 60} {
    \draw[black, fill = white] (\x, \x + 20) circle (1.2pt);
}
\draw[black, fill = white] (40 + 0.7071*0.6, 60 + 0.7071*0.6) arc (45:225:1.2pt);

% value''
\foreach \x in {0, 2.5, 5, 7.5, 10, 12.5, 15, 17.5, 20, 22.5, 25, 27.5, 30, 32.5, 35, 37.5, 40, 42.5, 45, 47.5, 50, 52.5, 55, 57.5, 60} {
    \draw[black, fill = black] (\x, \x + 40) circle (1.2pt);
}

% value majority
\draw[black, fill = SkyBlue] (40 - 0.7071*0.6, 60 - 0.7071*0.6) arc (-135:45:1.2pt);
\draw[black] (40 - 0.7071*0.6, 60 - 0.7071*0.6) -- (40 + 0.7071*0.6, 60 + 0.7071*0.6);

% value corner
% \draw[black, fill = BrickRed] (0, 0) circle (1.2pt);
\draw[black, fill = BrickRed] (0, 100) circle (1.2pt);
% \draw[black, fill = BrickRed] (100, 0) circle (1.2pt);
% \draw[black, fill = BrickRed] (100, 100) circle (1.2pt);

% informative lines
\draw[darkgray, line width = 5, draw opacity = 0.5] (12.5, 52.5) -- (30, 52.5);
\node[above, darkgray] at (21.25, 52.5) {$\+I^{k}_{\mathrm{hor}}$};
\draw[darkgray, line width = 5, draw opacity = 0.5] (10, 40) -- (10, 50);
\node[right, darkgray] at (10, 45) {$\+I^{k}_{\mathrm{LL}}$};
\draw[darkgray, line width = 5, draw opacity = 0.5] (10, 52.5) -- (10, 100);
\node[right, darkgray] at (10, 76.25) {$\+I^{k}_{\mathrm{UL}}$};
\draw[darkgray, line width = 5, draw opacity = 0.5] (30, 40) -- (30, 50);
\node[right, darkgray] at (30, 45) {$\+I^{k}_{\mathrm{LR}}$};
\draw[darkgray, line width = 5, draw opacity = 0.5] (30, 55) -- (30, 100);
\node[right, darkgray] at (30, 77.5) {$\+I^{k}_{\mathrm{UR}}$};

% good lines
\draw[teal, line width = 5, draw opacity = 0.75] (40, 52.5) -- (60, 52.5);
\node[above, teal] at (50, 52.5) {$\+G_{k}$};

% instance k
\node[left] at (10, 50) {$-\delta$};
\node[left] at (12.5, 52.5) {$+\delta$};
\node[right] at (30, 50) {$+\delta$};
\node[right] at (32.5, 52.5) {$-\delta$};
\draw[black, fill = orange] (10, 50) circle (1.2pt);
\draw[black, fill = blue] (12.5, 52.5) circle (1.2pt);
\draw[black, fill = blue] (30, 50) circle (1.2pt);
\draw[black, fill = orange] (32.5, 52.5) circle (1.2pt);

%legends
\draw[black, fill = black] (65, 85) circle (1.2pt) node[right] {$\VALupperleft$};
\draw[black, fill = white] (65, 79) circle (1.2pt) node[right] {$\VALlowerright$};
\draw[black, fill = BrickRed] (65, 73) circle (1.2pt) node[right] {$\VALcorner = \{0, 1\}^{2}$};
\draw[black, fill = SkyBlue] (65, 67) circle (1.2pt) node[right] {$\valmajority = (0.4, 0.6)$};
\draw[teal, line width = 5, draw opacity = 0.5] (64, 61) -- (69, 61);
\node[right, teal] at (70, 61) {$\+G_{k}$};
\draw[darkgray, line width = 5, draw opacity = 0.5] (64, 55) -- (69, 55);
\node[right, darkgray] at (70, 55) {$\+I^{k}$};

\end{tikzpicture}
    \caption{A diagram of the informative action subsets $\+I^{k} = \+I^{k}_{\mathrm{hor}} \cup \+I^{k}_{\mathrm{LL}} \cup \+I^{k}_{\mathrm{UL}} \cup \+I^{k}_{\mathrm{LR}} \cup \+I^{k}_{\mathrm{UR}}$ for $k \in [K]$.}
    \label{fig:GBB-correlated-informativeLine}
\end{figure}

The following \Cref{lem:GBB-correlated:similarity-in-T_G}---a counterpart of \Cref{lem:GBB-independent:similarity-in-T_S} in \Cref{sec:GBB-independent:LB}---is a direct combination of \Cref{lem:GBB-correlated:chain-rule,lem:GBB-correlated:bound-for-single-KL} below.

\begin{lemma}[Necessity of Taking Informative Actions]
\label{lem:GBB-correlated:similarity-in-T_G}
% \begin{flushleft}
In each possibility $k \in [K]$:
\begin{align*}
    {\bb E}^{k}[T_{\+G^{k}}] - {\bb E}^{0}[T_{\+G^{k}}]
    ~\le~ 4\delta T \cdot \sqrt{{\bb E}^{0}[T_{\+I^{k}}]}.
\end{align*}
% \end{flushleft}
\end{lemma}

Firstly, \Cref{lem:GBB-correlated:chain-rule} establishes a telescoping upper bound on ${\bb E}^{k}[T_{\+G^{k}}] - {\bb E}^{0}[T_{\+G^{k}}]$; its proof is identical to the proof of \Cref{lem:GBB-independent:chain-rule} and thus is omitted it for brevity. Note that \Cref{lem:GBB-correlated:chain-rule} holds both for the discrete lower-bound construction here and for the continuous lower-bound construction later in \Cref{sec:LB-GBB-Partial-Correlated:density}.

\begin{lemma}[Telescoping Upper Bounds]
\label{lem:GBB-correlated:chain-rule}
% \begin{flushleft}
In each possibility $k \in [K]$:
\begin{align*}
    &\textstyle {\bb E}^{k}[T_{\+G^{k}}] - {\bb E}^{0}[T_{\+G^{k}}]
    ~\le~ T \cdot \sqrt{\tfrac{1}{2} \sum_{r \in [T]} {\bb E}^{0}[\distKL({\bb P}^{0}_{\Feedback{r} \;|\; \+F^{r - 1}_{+}},\ {\bb P}^{k}_{\Feedback{r} \;|\; \+F^{r - 1}_{+}})]}.
\end{align*}
% \end{flushleft}
\end{lemma}

Secondly, \Cref{lem:GBB-correlated:bound-for-single-KL} upper bounds the above KL divergences as desired, through a careful analysis of our lower-bound construction; its proof follows the same line of reasoning as that of \Cref{lem:GBB-independent:bound-for-single-KL}.

\begin{lemma}[Upper Bounds on KL Divergences]
\label{lem:GBB-correlated:bound-for-single-KL}
\begin{flushleft}
In each possibility $k \in [K]$:
\begin{align*}
    \textstyle
    \sum_{r \in [T]} {\bb E}^{0}[\distKL({\bb P}^{0}_{\Feedback{r} \;|\; \+F^{r - 1}_{+}},\ {\bb P}^{k}_{\Feedback{r} \;|\; \+F^{r - 1}_{+}})]
    ~\le~ 30\delta^{2} \cdot {\bb E}^{0}[T_{\+I^{k}}].
\end{align*}
\end{flushleft}
\end{lemma}

\begin{proof}
We can reformulate each index-$(r \in [T])$ term as follows:\textsuperscript{\ref{footnote:abuse_notation}}
\begin{align*}
    & {\bb E}^{0}[\distKL({\bb P}^{0}_{\Feedback{r} \;|\; \+F^{r - 1}_{+}},\ {\bb P}^{k}_{\Feedback{r} \;|\; \+F^{r - 1}_{+}})]\\
    &\textstyle ~=~ \iint_{(p, q) \in \ACTION} \underbrace{\distKL({\bb P}^{0}_{\Feedback{r} \;|\; \Price{r} = (p, q)},\ {\bb P}^{k}_{\Feedback{r} \;|\; \Price{r} = (p, q)})}_{(\ddagger)}
    \cdot {\bb P}^{0}[\Price{r} = (\dd p, \dd q)].
\end{align*}
We can reason about the term $(\ddagger)$ through case analysis:
\begin{itemize}
    \item \textbf{Case~1: $(p, q) \notin \+I^{k} = \+I^{k}_{\mathrm{hor}} \cup \+I^{k}_{\mathrm{LL}} \cup \+I^{k}_{\mathrm{UL}} \cup \+I^{k}_{\mathrm{LR}} \cup \+I^{k}_{\mathrm{UR}}$.}
    Such a \textit{non-informative action} must give $(\ddagger) = 0$.
    
    \item \textbf{Case~2: $(p, q) \in \+I^{k}_{\mathrm{hor}}$.}
    For brevity, let $\alpha_{11} = \alpha_{11}(p, q) \defeq {\bb P}^{0}[\Feedback{r} = (1, 1) \;|\; \Price{r} = (p, q)]$; similarly for $\alpha_{10}, \alpha_{01}, \alpha_{00}$.
    Note that $\alpha_{11}, \alpha_{01}, \alpha_{10}, \alpha_{00} \ge 0.1$, given that the base instance $\+D^{0}$ assigns a probability mass of $0.1$ to each of the four corner points $\VALcorner = \{0, 1\}^{2}$. Based on \Cref{fig:GBB-correlated-informativeLine}, we can deduce that
    \begin{align*}
        \textstyle
        (\ddagger)
        &\textstyle ~=~ \alpha_{11} \cdot \ln(\frac{\alpha_{11}}{\alpha_{11} + \delta})
        + \alpha_{10} \cdot \ln(\frac{\alpha_{10}}{\alpha_{10} - \delta})
        + \alpha_{01} \cdot \ln(\frac{\alpha_{01}}{\alpha_{01} - \delta})
        + \alpha_{00} \cdot \ln(\frac{\alpha_{00}}{\alpha_{00} + \delta})\\
        &\textstyle ~\le~ -(\delta - \frac{\delta^{2}}{2\alpha_{11}})
        + (\delta + \frac{\delta^{2}}{\alpha_{10}})
        + (\delta + \frac{\delta^{2}}{\alpha_{01}})
        - (\delta - \frac{\delta^{2}}{2\alpha_{00}})\\
        &\textstyle ~=~ \frac{\delta^{2}}{2\alpha_{11}}
        + \frac{\delta^{2}}{\alpha_{10}}
        + \frac{\delta^{2}}{\alpha_{01}}
        + \frac{\delta^{2}}{2\alpha_{00}}\\
        &\textstyle ~\le~ 30\delta^{2}.
    \end{align*}
    Here the second step holds for $\alpha_{11}, \alpha_{01}, \alpha_{10}, \alpha_{00} \ge 0.1$ and $\delta \in [0, \frac{1}{20}]$ and can be verified through elementary algebra.
    And the last step also uses  $\alpha_{11}, \alpha_{01}, \alpha_{10}, \alpha_{00} \ge 0.1$.

    \item \textbf{Case~3: $(p, q) \in \+I^{k}_{\mathrm{LL}}$.}
    Similarly, we have $(\ddagger)
    = \alpha_{11} \cdot \ln(\frac{\alpha_{11}}{\alpha_{11} - \delta})
    + \alpha_{01} \cdot \ln(\frac{\alpha_{01}}{\alpha_{01} + \delta})
    % \le \frac{\delta^{2}}{\alpha_{11}}
    % + \frac{\delta^{2}}{2\alpha_{01}}
    \le 15\delta^{2}$.
    
    \item \textbf{Case~4: $(p, q) \in \+I^{k}_{\mathrm{UL}}$.}
    Similarly, we have $(\ddagger)
    = \alpha_{10} \cdot \ln(\frac{\alpha_{10}}{\alpha_{10} - \delta})
    + \alpha_{00} \cdot \ln(\frac{\alpha_{00}}{\alpha_{00} + \delta})
    % \le \frac{\delta^{2}}{\alpha_{10}}
    % + \frac{\delta^{2}}{2\alpha_{00}}
    \le 15\delta^{2}$.
    
    \item \textbf{Case~5: $(p, q) \in \+I^{k}_{\mathrm{LR}}$.}
    Similarly, we have $(\ddagger)
    = \alpha_{11} \cdot \ln(\frac{\alpha_{11}}{\alpha_{11} + \delta})
    + \alpha_{01} \cdot \ln(\frac{\alpha_{01}}{\alpha_{01} - \delta})
    % \le \frac{\delta^{2}}{2\alpha_{11}}
    % + \frac{\delta^{2}}{\alpha_{01}}
    \le 15\delta^{2}$.
    
    \item \textbf{Case~6: $(p, q) \in \+I^{k}_{\mathrm{UR}}$.}
    Similarly, we have $(\ddagger)
    = \alpha_{10} \cdot \ln(\frac{\alpha_{10}}{\alpha_{10} + \delta})
    + \alpha_{00} \cdot \ln(\frac{\alpha_{00}}{\alpha_{00} - \delta})
    % \le \frac{\delta^{2}}{2\alpha_{10}}
    % + \frac{\delta^{2}}{\alpha_{00}}
    \le 15\delta^{2}$.
\end{itemize}
Putting everything together gives
\begin{align*}
    \textstyle \sum_{r \in [T]} {\bb E}^{0}[\distKL({\bb P}^{0}_{\Feedback{r} \;|\; \+F^{r - 1}_{+}},\ {\bb P}^{k}_{\Feedback{r} \;|\; \+F^{r - 1}_{+}})]
    &\textstyle ~\le~ 30\delta^{2} \cdot \sum_{r \in [T]} \iint_{(p, q) \in \+I^{k}} {\bb P}^{0}[\Price{r} = (\dd p, \dd q)]\\
    &\textstyle ~=~ 30\delta^{2} \cdot {\bb E}^{0}[T_{\+I^{k}}].
\end{align*}
This completes the proof of \Cref{lem:GBB-correlated:bound-for-single-KL}.
\end{proof}

Finally, we are ready to accomplish \Cref{thm:GBB-correlated:LB}.

\begin{proof}[Proof of \Cref{thm:GBB-correlated:LB}]
Recall that we divide $\ACTION$ into two disjoint subsets $\ACTION = \+G \cup \+B$.
For the base instance $\+D^{0}$, each take of a bad action $\Price{t} \in \+B$ incurs at least $0.1$ regret, and each take of a good action $\Price{t} \in \+G$ incurs $3\delta K \cdot (\BPrice{t} - \SPrice{t})$ regret (\Cref{prop:GBB-correlated:per-round-regret}), so the total regret is at least
\begin{align*}
    \Regret_{\+D^{0}} 
    &\textstyle ~=~ {\bb E}^{0}\big[\sum_{t \in [T]} \big({\bb 1}[\Price{t} \in \+B] \cdot \Regret_{\+D^{0}}\Price{t}\\
    &\textstyle \phantom{~=~ {\bb E}^{0}\big[\sum_{t \in [T]} \big(} + {\bb 1}[\Price{t} \in \+G] \cdot \Regret_{\+D^{0}}\Price{t}\big)\big]\\
    \mr{\Cref{prop:GBB-correlated:per-round-regret}}
    &\textstyle ~\ge~ 0.1 \cdot {\bb E}^{0}[T_{\+B}] ~+~ 3\delta K \cdot {\bb E}^{0}\big[\sum_{t \in [T]} (\BPrice{t} - \SPrice{t}) \cdot {\bb 1}[\Price{t} \in \+G]\big]\\
    \mr{\Cref{lem:GBB-correlated:PGBB} and $\+B = \ACTION \setminus \+G$}
    &\textstyle ~\ge~ 0.1 \cdot {\bb E}^{0}[T_{\+B}] ~-~ 3\delta K \cdot {\bb E}^{0}\big[\sum_{t \in [T]} (\BPrice{t} - \SPrice{t}) \cdot {\bb 1}[\Price{t} \in \+B]\big]\\
    \mr{$\BPrice{t} - \SPrice{t}\le 1$}
    &\textstyle ~\ge~ 0.1 \cdot {\bb E}^{0}[T_{\+B}] ~-~ 3\delta K \cdot {\bb E}^{0}[T_{\+B}]\\
    \mr{$\delta K = \frac{0.1K}{5K + 2} \le 0.02$}
    &\textstyle ~\ge~ 0.04 \cdot {\bb E}^{0}[T_{\+B}].
\end{align*}

For each hard instance $\+D^{k}$, $\forall k \in [K]$, a moment's reflection will show that the deduction above for $\+D^{0}$ still holds, and each take of an action $\Price{t} \in \+G \setminus \+G^{k}$ will incur $0.2\delta$ more regret (\Cref{prop:GBB-correlated:per-round-regret}). Thus, we can obtain
\begin{align*}
    \Regret_{\+D^{k}} 
    &\textstyle ~\ge~ 0.04 \cdot {\bb E}^{k}[T_{\+B}] ~+~ {\bb E}^{k}\big[\sum_{t \in [T]} 0.2\delta \cdot {\bb 1}[\Price{t} \in \+G \setminus \+G^{k}]\big]\\
    & ~=~ 0.04 \cdot {\bb E}^{k}[T_{\+B}] ~+~ 0.2\delta \cdot \big({\bb E}^{k}[T_{\+G}] - {\bb E}^{k}[T_{\+G^{k}}]\big)\\
    \mr{\Cref{lem:GBB-correlated:similarity-in-T_G}}
    & ~\ge~ 0.04 \cdot {\bb E}^{k}[T_{\+B}] ~+~ 0.2\delta \cdot \big(T - {\bb E}^{k}[T_{\+B}] - {\bb E}^{0}[T_{\+G^{k}}] - 4\delta T \cdot \sqrt{{\bb E}^{0}[T_{\+I^{k}}]}\big)\\
    \mr{$\delta \in [0, \frac{1}{20}]$}
    & ~\ge~ 0.2\delta \cdot \big(T - {\bb E}^{0}[T_{\+G^{k}}] - 4\delta T \cdot \sqrt{{\bb E}^{0}[T_{\+I^{k}}]}\big).
\end{align*}
Taking the average of $\Regret_{\+D^{k}}$'s for $k \in [K]$ gives
\begin{align*}
    \textstyle
    \frac{1}{K} \cdot \sum_{k \in [K]} \Regret_{\+D^{k}}
    &\textstyle ~\ge~ \frac{1}{K} \cdot \sum_{k \in [K]} 0.2\delta \cdot \big(T - {\bb E}^{0}[T_{\+G^{k}}] - 4\delta T \cdot \sqrt{{\bb E}^{0}[T_{\+I^{k}}]}\big)\\
    \mr{$\sum_{k \in [K]} T_{\+G^{k}} \le T$ (a.s.)}
    &\textstyle ~\ge~ 0.2\delta \cdot \big(T - \frac{T}{K} - \frac{4\delta T}{K} \cdot \sum_{k \in [K]} \sqrt{{\bb E}^{0}[T_{\+I^{k}}]}\big)\\
    \mr{Cauchy-Schwarz inequality}
    &\textstyle ~\ge~ 0.2\delta \cdot \big(T - \frac{T}{K} - 4\delta T \cdot \sqrt{\frac{1}{K} \cdot \sum_{k \in [K]} {\bb E}^{0}[T_{\+I^{k}}]}\big)\\
    \mr{$\sum_{k \in [K]} T_{\+I^{k}}\le 2T_{\+B}$ (a.s.)}
    &\textstyle ~\ge~ 0.2\delta \cdot \big(T - \frac{T}{K} - 4\delta T \cdot \sqrt{\frac{2}{K} \cdot {\bb E}^{0}[T_{\+B}]}\big).
\end{align*}
Here the last step uses $\sum_{k \in [K]} T_{\+G^{k}} \le T$, as a consequence of that $\+G^{k}$'s for $k \in [K]$ are disjoint.
And the last step uses $\sum_{k \in [K]} T_{\+I^{k}}\le 2T_{\+B}$, as a consequence of that $\+I^{k}_{\mathrm{hor}}$'s for $k \in [K]$ are disjoint and that $\+I^{k}_{\mathrm{LL}}$'s, $\+I^{k}_{\mathrm{UL}}$'s, $\+I^{k}_{\mathrm{LR}}$'s, $\+I^{k}_{\mathrm{UR}}$'s for $k \in [K]$ all are disjoint.

Plugging in $K = T^{1 / 4}$ and $\delta = \frac{0.1}{5K + 2}$, it is easy to verify through elementary algebra that
\begin{align*}
    & \textstyle \max\big\{\Regret_{\+D^{0}},\ \frac{1}{K} \cdot \sum_{k \in [K]} \Regret_{\+D^{k}}\big\}\\
    &\textstyle ~\ge~ \max\big\{0.04 \cdot {\bb E}^{0}[T_{\+B}],\ 0.004 \cdot T^{3 / 4} \cdot \big(\frac{1 - T^{-1 / 4}}{1 + 0.4T^{-1 / 4}} - \frac{0.08}{(1 + 0.4T^{-1 / 4})^{2}} \cdot \sqrt{\frac{2 \cdot {\bb E}^{0}[T_{\+B}]}{T^{3 / 4}}}\big)\big\}\\
    % &\textstyle ~=~ \max\big\{0.04 \cdot {\bb E}^{0}[T_{\+B}],\ 0.004 \cdot T^{3 / 4} \cdot \big((1 \pm o_{T}(1)) - (0.08 \pm o_{T}(1)) \cdot \sqrt{\frac{2 \cdot {\bb E}^{0}[T_{\+B}]}{T^{3 / 4}}}\big)\big\}\\
    &\textstyle ~=~ \Omega(T^{3 / 4}).
\end{align*}

In sum, the considered mechanism $\Mech^{\ACTION}$ incurs $\Omega(T^{3 / 4})$ regret in at least one possibility $k \in [0 : K]$. Then, the arbitrariness of $\Mech^{\ACTION}$ implies \Cref{thm:GBB-correlated:LB}.
\end{proof}

\subsection{Modification for Density-Bounded Values}
\label{sec:LB-GBB-Partial-Correlated:density}

% \ctodo{A figure showing the modification is enough.}
\begin{figure}[t]
    \centering
    \tikzset{every picture/.style={line width = 0.75pt}} %set default line width to 0.75pt        

\begin{tikzpicture}[x = 2pt, y = 2pt, scale = 1.5]
% corners
\def\ply{2.2}
\def\length{10}

% good action
\fill[green, fill opacity = 0.2] (40, 45) -- (55, 45) -- (55, 60) -- (40, 60) -- cycle;

% bad action
\fill[red, fill opacity = 0.2] (0, 45) -- (40, 45) -- (40, 60) -- (55, 60) -- (55, 100) -- (0, 100) -- cycle;

\draw (0, 0) node[anchor = 0] {$(0, 0)$};
\draw (0, 100) node[anchor = 0] {$(0, 1)$};
\draw (100, 0) node[anchor = 180] {$(1, 0)$};
\draw (100, 100) node[anchor = 180] {$(1, 1)$};

\draw[black, fill = BrickRed, draw opacity =0, fill opacity = 0.5] (0, 0) rectangle (0+\length/2,0+\length/2);
\draw[black, fill = BrickRed, draw opacity =0, fill opacity = 0.5] (0, 100) rectangle (0+\length/2,100-\length/2);
\draw[black, fill = BrickRed, draw opacity =0, fill opacity = 0.5] (100, 0) rectangle (100-\length/2,0+\length/2);
\draw[black, fill = BrickRed, draw opacity =0, fill opacity = 0.5] (100, 100) rectangle (100-\length/2,100-\length/2);

\draw[black, fill = SkyBlue, draw opacity =0, fill opacity = 0.5] (40, 60) rectangle (40+\length/2,60-\length/2);

\draw (0,0) -- (0,100) -- (100,100) -- (100,0) -- cycle;
\draw (50, -3) node [below][inner sep=0.75pt] {seller};
\draw (-3, 50) node [above][inner sep = 0.75pt,rotate=90] {buyer};

% diagonal
\draw (0, 0) -- (100, 100);

% boxed line
\foreach \y in {5, 10, 15, 20, 25, 30, 35, 40, 45, 50, 55, 60, 65, 70, 75, 80, 85, 90, 95} {
    \draw[dotted, draw opacity = 0.5]  (\y, 0) -- (\y, 100);
    \draw [dotted, draw opacity = 0.5]  (0, \y) -- (100, \y);
}

% good action
% \fill[green, fill opacity = 0.2] (40, 40) -- (60, 40) -- (60, 60) -- (40, 60) -- cycle;

% bad action
% \fill[red, fill opacity = 0.2] (0, 40) -- (40, 40) -- (40, 60) -- (60, 60) -- (60, 100) -- (0, 100) -- cycle;

% value'
% \foreach \x in {20, 25, 30, 35, 45, 50, 55, 60} {
%     \draw[black, fill = white] (\x, \x + 20) circle (2pt);
% }
% \draw[black, fill = white] (40 + 0.7071, 60 + 0.7071) arc (45:225:2pt);

% value''
% \foreach \x in {0, 5, 10, 15, 20, 25, 30, 35, 40, 45, 50, 55, 60} {
%     \draw[black, fill = black] (\x, \x + 40) circle (2pt);
% }

\foreach \x in {0, 2.5, 5, 7.5} {
     \draw[draw=black, fill=orange,draw opacity =0.5, fill opacity = 0.8] (\x+5, \x + 45) -- (\x+\length+5,\x+45) -- (\x+\length+5+\ply,\x+45+\ply) -- (\x+5+\ply, \x + 45+\ply) -- cycle;
     \draw[draw=black, fill=blue,draw opacity =0.5, fill opacity = 0.8] (\x+15, \x + 45) -- (\x+15+\length,\x+45) -- (\x+15+\length+\ply,\x+45+\ply) -- (\x+15+\ply, \x + 45+\ply) -- cycle;
     
}

% \foreach \x in { 45, 50, 55, 60} {
%      \draw[black, fill = orange, draw opacity =0, fill opacity = 0.5] (\x-\ply, \x + 40) rectangle (\x,\x+40-\length);
%      \draw[black, fill = blue, draw opacity =0, fill opacity = 0.5] (\x-\ply, \x + 20) rectangle (\x,\x+20+\length);
% }

\fill[OliveGreen, fill opacity = 0.8] (45, 65) -- (55, 75) -- (55, 95) -- (45, 85) -- cycle;

% \foreach \x in { 25, 30, 35, 40} {
%     \draw[black, fill = orange, draw opacity =0, fill opacity = 0.5] (\x-\ply, \x + 40) rectangle (\x+\ply,\x+40+\length);
% }

% hard instances 
% \fill[yellow, fill opacity = 0.4] (5, 45) -- (20, 60) -- (40, 60) -- (25, 45) -- cycle;

%legends
\draw[black, fill = orange,fill opacity=0.8] (110,64) -- (114,64) -- (115,65) -- (111,65) --cycle;
\node[right] at (115,65) {$\set{\VALhorleft^k}_{k\in [K]}$};
\draw[black, fill = blue,fill opacity=0.8] (110,54) -- (114,54) -- (115,55) -- (111,55) --cycle;
\node[right] at (115,55) {$\set{\VALhorright^k}_{k\in [K]}$};

\draw[black, fill = OliveGreen,fill opacity=0.8] (110,44) -- (112,46) -- (112,49) -- (110,47) --cycle;
\node[right] at (115,45) {$\VALver$};

\draw[black, fill = BrickRed,fill opacity=0.8] (110,34) -- (113,34) -- (113,37) -- (110,37) --cycle;
\node[right] at (115,35) {$\VALcorner$};

\draw[black, fill = SkyBlue,fill opacity = 0.8] (110,24) -- (113,24) -- (113,27) -- (110,27) --cycle;
\node[right] at (115,25)  {$\VALmajority$};

% \fill[yellow, fill opacity = 0.5] (110, 24) -- (111, 26) -- (115, 26) -- (114, 24) -- cycle;
% \node [anchor=west] at (115,25) {tweak area};
\end{tikzpicture}
    \caption{The lower-bound construction for \Cref{thm:GBB-correlated:LB} under the density-boundedness assumption. Each strip has width $\Theta(\delta)$ and length $\Theta(1)$. We can still construct $\Theta(1 / \delta)$ hard instances by tweaking the lower-left parallelogram area.}
    \label{fig:GBB-correlated:smooth}
\end{figure}

To enforce a density-boundedness constraint (\Cref{asm:density}) on our lower-bound construction, we slightly modify the original discrete instances by spreading the mass at each discrete point.

% \red{Specifically, for every point in $\+D^{k}$ (for some $k \in [K]$), we distribute its mass over a small $0.1 \times \delta$ parallelogram adjacent to that point.
% If the point lies in $\VALupperleft$, its corresponding parallelogram extends upward; if it lies in $\VALlowerright$, the parallelogram extends to the left.
% For the corner points in $\VALcorner$ and the majority point, we instead spread their mass uniformly over an adjacent $0.05 \times 0.05$ square.
% All other minor points can either be assigned small non-overlapping parallelogram or simply removed without affecting the hardness of the constructed instances---here, for simplicity, we remove them.}

\subsection*{Lower-Bound Construction for Continuous Values}

Our construction utilizes three parameters $K = \Theta(T^{1 / 4})$, $\delta = \Theta(T^{-1 / 4})$, and $\Delta = \Theta(T^{-1 / 4})$, as follows.
\begin{align*}
    K &\textstyle ~\defeq~ T^{1 / 4},\\
    \delta &\textstyle ~\defeq~ \frac{0.02}{2(K + 1)},\\
    \Delta &\textstyle ~\defeq~ \frac{0.1}{K + 1}.
\end{align*}
% Note that \red{$\delta \in [0, \frac{1}{100}]$ whenever $K \ge 2$}.

% \vspace{.1in}
\noindent
\textbf{The Value Support.}
We construct $(K + 1)$ base/hard instances $\{\+D^{k}\}_{k \in [0 : K]}$ over a common \textit{continuous} support, which consists of $2(K + 1) + 3$ regions $\{\VALhorleft^{k}\}_{k \in [0 : K]} \cup \{\VALhorright^{k}\}_{k \in [0 : K]} \cup \VALver \cup \VALcorner \cup \VALmajority$, as follows.
\begin{align*}
    \VALhorleft^{k}
    &\textstyle ~\defeq~ \set{(p, q) : q - p \in [0.3,\ 0.4],\ q \in [0.45 + k \Delta,\ 0.45 + (k + 1) \Delta]},
    && \forall k \in [0 : K],\\
    \VALhorright^{k}
    &\textstyle ~\defeq~ \{(p, q) : q - p \in [0.2,\ 0.3],\ q \in [0.45 + k \Delta,\ 0.45 + (k + 1) \Delta]\},
    && \forall k \in [0 : K],\\
    \VALver 
    &\textstyle ~\defeq~ \{(p, q) : q - p \in [0.2,\ 0.4],\ p \in [0.45,\ 0.55]\}\\
    \VALcorner
    & ~\defeq~ \{[0,\ 0.05] \cup [0.95,\ 1]\}^{2}, \\
    \VALmajority
    & ~\defeq~ [0.4,\ 0.45] \times [0.55,\ 0.6].
\end{align*}
These five types of regions serve different purposes; see \Cref{fig:GBB-correlated:smooth} for a diagram.
Note that the ``horizontal-left'' strips $\VALhorleft^{k}$ and the ``horizontal-right'' strips $\VALhorright^{k}$ share boundaries without overlapping, and together
\begin{align*}
    \{\VALhorleft^{k}\}_{k \in [0 : K]} \cup \{\VALhorright^{k}\}_{k \in[0:K]}
    ~=~ \{(p, q) : q - p \in [0.2,\ 0.4], q \in [0.45,\ 0.55]\}.
\end{align*}

\vspace{.1in}
\noindent
\textbf{The Base/Hard Instances.}
To define a base/hard instance, it suffices to specify the total probability mass assigned to each of the $2(K + 1) + 3$ regions---the probability density is \textit{uniform} within every single region (but can change across different regions).

Among the $K + 1$ instances, $\+D^{0}$ is the base instance and is given as follows. 
\begin{align*}
    {\bb P}_{{\Val{} \sim \+D^{0}}}[\Val{} \in V]
    ~\defeq~
    \begin{cases}
        \delta, & V = \VALhorleft^{k},\ \forall k \in [0 : K]\\
        \delta, & V = \VALhorright^{k},\ \forall k \in [0 : K]\\
        0.02, & V = \VALver\\
        0.4, & V = \VALcorner\\
        0.56, & V = \VALmajority
    \end{cases}.
\end{align*}
This $\+D^{0}$ is a well-defined distribution since $\delta \cdot 2(K + 1) + 0.02 + 0.4 + 0.56 = 1$. Note that $\{\VALhorleft^{k}\}_{k \in [0 : K]} \cup \{\VALhorright^{k}\}_{k \in [0 : K]} \cup \VALver$ has uniform density $= 1$, $\VALcorner$ has uniform density $= 40$, and $\VALmajority$ has the \textit{maximum} uniform density $= 224$. (This meets the density-boundedness assumption $M = 224$.)

In contrast, each $\+D^{k}$ for $k \in [K]$ is a hard instance and tweaks the probability masses at two horizontal-left strips $\VALhorleft^{k - 1},\ \VALhorleft^{k}$ and two horizontal-right points $\VALhorright^{k - 1},\ \VALhorright^{k}$ (but otherwise is identical to the base instance $\+D^{0}$). Again, this $\+D^{k}$ is a well-defined distribution with maximum density $= 224$.
\begin{align*}
    {\bb P}_{{\Val{} \sim \+D^{k}}}[\Val{} \in V]
    ~\defeq~
    \begin{cases}
        {\bb P}_{{\Val{} \sim \+D^{0}}}[\Val{} \in V] + \delta,
        & V \in \{\VALhorleft^{k},\ \VALhorright^{k - 1}\}\\
        {\bb P}_{{\Val{} \sim \+D^{0}}}[\Val{} \in V] - \delta,
        & V \in \{\VALhorleft^{k - 1},\ \VALhorright^{k}\}\\
        {\bb P}_{{\Val{} \sim \+D^{0}}}[\Val{} \in V],
        & \text{otherwise}
    \end{cases}.
\end{align*}

\subsection*{Lower-Bound Analysis for Continuous Values}

Compared to \Cref{sec:LB-GBB-Partial-Correlated:analysis} (the discrete case), we can readopt the notation and adapt the overall approach---again, the proof takes three steps, which one-to-one correspond to their counterparts in \Cref{sec:LB-GBB-Partial-Correlated:analysis}.

\vspace{.1in}
\noindent
\textbf{Step~1.}
The action space can be restricted to the rectangle $[0,\ 0.55] \times [0.45,\ 1]$ (even though we no longer discretize the action space as in the counterpart step of \Cref{sec:LB-GBB-Partial-Correlated:analysis}), which can be divided into the following ``good'' subset $\+G$ and ``bad'' subset $\+B$.
\begin{align*}
    \+G
    & ~\defeq~ \cup_{k \in [K]} \+G^{k},\\
    \+G^{k}
    & ~\defeq~ [0.4,\ 0.55] \times [0.45 + (k - 1) \Delta,\ 0.45 + (k + 1) \Delta],
    && \forall k \in [K],\\
    \+B
    & ~\defeq~ [0,\ 0.6] \times [0.45,\ 1] \setminus \+G.
\end{align*}

% \vspace{.1in}
\noindent
\textbf{Step~2.}
For the counterpart in \Cref{sec:LB-GBB-Partial-Correlated:analysis}---essentially \Cref{lem:GBB-correlated:PGBB} and its proof---the key point is to make the trade probability \textit{invariant} along (a segment of) the diagonal and across different instances;
our current lower-bound construction (the continuous case) maintains this property:
\begin{align*}
    \textstyle
    & {\bb P}^{k}[\Trade{t} = 1 \;|\; \Price{t} = (p, p)]
    ~=~ 0.02 + 0.1 + 0.56
    ~=~ 0.68,
    && \forall k \in [K],\ \forall p \in [0.45,\ 0.55].
\end{align*}
By adapting the proof of \Cref{lem:GBB-correlated:PGBB}, we can still relax the {\GBB} constraint to \blackeqref{eq:GlobalPriceBalance}.

\vspace{.1in}
\noindent
\textbf{Step~3.}
This step is the main difference between the continuous and discrete cases. We need to consider a \textit{strongly informative region} $\+I_{\mathrm{S}}$ and a \textit{weakly informative region} $\+I_{\mathrm{W}}$. 
(It is easy to check that actions outside $\+I_{\mathrm{S}}$ and $\+I_{\mathrm{W}}$ cannot help in distinguishing instances.)
\begin{itemize}
    \item The \textit{strongly informative region} $\+I_{\mathrm{S}}$ is defined as follows.
    \begin{align*}
        \+I_{\mathrm{S}}
        & ~\defeq~ \cup_{k \in [K]} \+I_{\mathrm{S}}^{k},\\
        \+I_{\mathrm{S}}^{k}
        & ~\defeq~ \VALhorleft^{k} \cup \VALhorright^{k} \cup \VALhorleft^{k - 1} \cup \VALhorright^{k - 1},
        && \forall k \in [K].
    \end{align*}
    For each hard instance $\+D^{k}$, $\forall k \in [K]$, our construction guarantees that
    \begin{align*}
        \textstyle
        \max_{(p, q) \in \+I_{\mathrm{S}}^{k}} \big|{\bb P}^{0}[\Trade{t} = 1 \;|\; \Price{t} = (p, q)] - {\bb P}^{k}[\Trade{t} = 1 \;|\; \Price{t} = (p, q)]\big| ~\le~ \delta.
    \end{align*}
    
    \item The \textit{weakly informative region} $\+I_{\mathrm{W}}$ is defined as follows.
    \begin{align*}
        \+I_{\mathrm{W}}
        ~=~ [0.05,\ 0.35] \times [0.45,\ 1] \setminus \+I_{\mathrm{S}}.
        % \{(p, q) : p \in [0.25,\ 0.35],\ q \in [0.45,\ p + 0.2]\}.
    \end{align*}
    For each hard instance $\+D^{k}$, $\forall k \in [K]$, our construction guarantees that
    \begin{align*}
        \textstyle
        \max_{(p, q) \in \+I_{\mathrm{W}}} \big|{\bb P}^{0}[\Trade{t} = 1 \;|\; \Price{t} = (p, q)] - {\bb P}^{k}[\Trade{t} = 1 \;|\; \Price{t} = (p, q)]\big|
        ~=~ 1 \cdot \Delta^{2}
        ~=~ 100\delta^{2}.
    \end{align*}
\end{itemize}
In combination and by adapting the proof of \Cref{lem:GBB-correlated:bound-for-single-KL} (note that $\delta \in [0,\frac{1}{100}]$), we can deduce that
\begin{align*}
    \textstyle
    \sum_{r \in [T]} {\bb E}^{0}[\distKL({\bb P}^{0}_{\Feedback{r} \;|\; \+F^{r - 1}_{+}},\ {\bb P}^{k}_{\Feedback{r} \;|\; \+F^{r - 1}_{+}})]
    ~\le~ 30\delta^{2} \cdot {\bb E}^{0}[T_{\+I_{\mathrm{S}}^{k}}] + 15 \cdot (100\delta^{2})^{2} \cdot {\bb E}^{0}[T_{\+B}].
\end{align*}
Together with \Cref{{lem:GBB-correlated:chain-rule}} and since $\sqrt{a + b} \le \sqrt{a} + \sqrt{b}$ for $a, b \ge 0$, we derive the counterpart of \Cref{lem:GBB-correlated:similarity-in-T_G}:
\begin{align*}
    {\bb E}^{k}[T_{\+G^{k}}] - {\bb E}^{0}[T_{\+G^{k}}]
    & ~\le~ 4\delta T \cdot \sqrt{{\bb E}^{0}[T_{\+I_{\mathrm{S}}^{k}}]} + 274\delta^{2} T \cdot \sqrt{{\bb E}^{0}[T_{\+I_{\mathrm{W}}}]}\\
    \mr{$T_{\+I_{\mathrm{W}}} \le T_{\+B}$ (a.s.)}
    & ~\le~ 4\delta T \cdot \sqrt{{\bb E}^{0}[T_{\+I_{\mathrm{S}}^{k}}]} + 274\delta^{2} T \cdot \sqrt{{\bb E}^{0}[T_{\+B}]}.
\end{align*}

% The rest of the argument mirrors the proof of \Cref{thm:GBB-correlated:LB} for the discrete case.
In addition, the following is the counterpart of \Cref{prop:GBB-correlated:per-round-regret}.
For the base instance $\+D^{0}$ and the hard instances $\+D^{k}$, $\forall k \in [K]$:
\begin{align*}
    \Regret_{\+D^{0}}(p, q)
    & ~\ge~
    \begin{cases}
        0.084, & \forall (p, q) \in \+B\\
        0.06 \cdot (q - p), & \forall (p, q) \in \+G
    \end{cases},\\
    \Regret_{\+D^{k}}(p, q)
    & ~\ge~
    \begin{cases}
        0.084, & \forall (p, q) \in \+B\\
        0.06 \cdot (q - p) + 0.1\delta \cdot {\bb 1}[(p, q) \notin \+G^{k}], & \forall (p, q) \in \+G
    \end{cases}.
\end{align*}

With all preparations above, we area ready to establish 

\begin{proof}[Proof of \Cref{thm:GBB-correlated:LB}]
By adapting the proof in the discrete case, we can show that
\[
    \Regret_{\+D^{0}}
    ~\ge~ (0.084 - 0.06) \cdot {\bb E}^{0}[T_{\+B}]
    ~\ge~ 0.024 \cdot {\bb E}^{0}[T_{\+B}],
\]
and that, for each hard instance $\+D^{k}$, $\forall k \in [K]$,
\begin{align*}
    \Regret_{\+D^{k}}
    &\textstyle ~\ge~ 0.024 \cdot {\bb E}^{k}[T_{\+B}] ~+~ 0.1\delta \cdot \big(T - {\bb E}^{k}[T_{\+B}] - {\bb E}^{0}[T_{\+G^{k}}]\\
    &\textstyle \phantom{~\ge~ 0.024 \cdot {\bb E}^{k}[T_{\+B}] ~+~ 0.1\delta \cdot \big(T}
    - 4\delta T \cdot \sqrt{{\bb E}^{0}[T_{\+I_{\mathrm{S}}^{k}}]} - 274\delta^{2} T \cdot \sqrt{{\bb E}^{0}[T_{\+B}]}\big)\\
    \mr{$\delta \in [0, \frac{1}{100}]$}
    &\textstyle ~\ge~ 0.1\delta \cdot \big(T - {\bb E}^{0}[T_{\+G^{k}}] - 4\delta T \cdot \sqrt{{\bb E}^{0}[T_{\+I_{\mathrm{S}}^{k}}]} - 274\delta^{2} T \cdot \sqrt{{\bb E}^{0}[T_{\+B}]}\big).
\end{align*}
Taking the average of $\Regret_{\+D^{k}}$'s for $k \in [K]$ gives
\begin{align*}
    \textstyle
    \frac{1}{K} \cdot \sum_{k \in [K]} \Regret_{\+D^{k}}
    &\textstyle ~\ge~ \frac{1}{K} \cdot \sum_{k \in [K]} 0.1\delta \cdot \big(T - {\bb E}^{0}[T_{\+G^{k}}] - 4\delta T \cdot \sqrt{{\bb E}^{0}[T_{\+I_{\mathrm{S}}^{k}}]} - 274\delta^{2} T \cdot \sqrt{{\bb E}^{0}[T_{\+B}]}\big)\\
    \mr{$\sum_{k \in [K]} T_{\+G^{k}} \le 2T$ (a.s.)}
    &\textstyle ~\ge~ 0.1\delta \cdot \big(T - \frac{2T}{K} - \frac{4\delta T}{K} \cdot \sum_{k \in [K]} \sqrt{{\bb E}^{0}[T_{\+I_{\mathrm{S}}^{k}}]} - 274\delta^{2} T \cdot \sqrt{{\bb E}^{0}[T_{\+B}]}\big)\\
    \mr{Cauchy-Schwarz inequality}
    &\textstyle ~\ge~ 0.1\delta \cdot \big(T - \frac{2T}{K} - 4\delta T \cdot \sqrt{\frac{1}{K} \cdot \sum_{k \in [K]} {\bb E}^{0}[T_{\+I_{\mathrm{S}}^{k}}]} - 274\delta^{2} T \cdot \sqrt{{\bb E}^{0}[T_{\+B}]}\big)\\
    \mr{$\sum_{k \in [K]} T_{\+I_{\mathrm{S}}^{k}}\le 2T_{\+B}$ (a.s.)}
    &\textstyle ~\ge~ 0.1\delta \cdot \big(T - \frac{2T}{K} - 4\delta T \cdot \sqrt{\frac{2}{K} \cdot {\bb E}^{0}[T_{\+B}]} - 274\delta^{2} T \cdot \sqrt{{\bb E}^{0}[T_{\+B}]}\big)\\
    \mr{$\delta \in [0, \frac{1}{100}]$}
    &\textstyle ~\ge~ 0.1\delta \cdot \big(T - \frac{2T}{K} - 6\delta T \cdot \sqrt{\frac{2}{K} \cdot {\bb E}^{0}[T_{\+B}]}\big).
\end{align*}
Substituting $K = T^{1 / 4}$ and $\delta = \frac{0.02}{2(K + 1)}$ yields, via direct calculation,
\begin{align*}
    & \textstyle \max\big\{\Regret_{\+D^{0}},\ \frac{1}{K} \cdot \sum_{k \in [K]} \Regret_{\+D^{k}}\big\}\\
    &\textstyle ~\ge~ \max\big\{0.024 \cdot {\bb E}^{0}[T_{\+B}],\ 0.001 \cdot T^{3 / 4} \cdot \big(\frac{1 - 2T^{-1 / 4}}{1 + T^{-1 / 4}} - \frac{0.06}{(1 + T^{-1 / 4})^{2}} \cdot \sqrt{\frac{2 \cdot {\bb E}^{0}[T_{\+B}]}{T^{3 / 4}}}\big)\big\}\\
    &\textstyle ~=~ \Omega(T^{3 / 4}).
\end{align*}
This accomplishes the proof of \Cref{thm:GBB-correlated:LB} in the continuous case.
\end{proof}

\bibliography{main}
\bibliographystyle{alpha}
% \newpage

\appendix

\section{\texorpdfstring{$\Omega(T^{1/2})$}{} {\GBB} Full-Feedback Lower Bound for Independent Values}
\label{sec:appendix:GBB}

\begin{theorem}[{\GBB} Full-Feedback Lower Bound for Independent Values]
\label{thm:appendix:GBB}
\begin{flushleft}
In the ``independent values'' setting (with or without the density-boundedness assumption---\Cref{asm:density} with parameter $M = 11$),\\
every ``{\GBB} full-feedback fixed-price mechanism'' has worst-case regret $\Omega(T^{1 / 2})$.
\end{flushleft}
\end{theorem}

\begin{proof}[Proof (Sketch)]
We simly reuse the lower-bound construction in \Cref{subsec:GBB-independent-LB-instance}---now $\delta \defeq \frac{1}{2}T^{-1 / 2} \in [0, \frac{1}{2}]$---and adapt the lower-bound analysis in \Cref{subsec:GBB-independent-LB-analysis}---now full feedback in place of semi feedback.

We assert the following \Cref{eq:GBB-independent-full:similarity}, which is a full-feedback counterpart of \Cref{lem:GBB-independent:similarity-in-T_S}.
\begin{align}
\label{eq:GBB-independent-full:similarity}
    &\textstyle {\bb E}^{1}[T_{\+A}] - {\bb E}^{2}[T_{\+A}]
    ~\le~ \delta T \cdot \sqrt{T},
    && \forall \+A \subseteq [0, 1]^{2}.
\end{align}
The proof of \Cref{eq:GBB-independent-full:similarity} only differs from that of \Cref{lem:GBB-independent:similarity-in-T_S} in two places.
Firstly, full feedback makes the \textit{informative action subset} $\+I = [0, 1]^{2}$ include everything, which gives $T_{\+I} = T$ (almost surely).
Secondly, full feedback reveals the \textit{buyer's true value} $\BVal{t}$---rather than only his/her intention to trade $\BFeedback{t}$---at the end of every round $t \in [T]$, so the equation $(\dagger) \le \frac{1}{2}\delta^{2}$ established in \textbf{Case~2} in the proof \Cref{lem:GBB-independent:bound-for-single-KL} should be replaced by the following:
\begin{align*}
    \textstyle
    \distKL({\bb P}^{1}_{\BVal{r} \;|\; (\SVal{r}, \SPrice{r}, \BPrice{r}) = (s, p, q)},\ {\bb P}^{2}_{\BVal{r} \;|\; (\SVal{r}, \SPrice{r}, \BPrice{r}) = (s, p, q)})
    ~=~ \frac{1 + \delta}{4\theta} \cdot \theta \cdot \ln(\frac{1 + \delta}{1 - \delta})
    ~+~ \frac{1 - \delta}{4\theta} \cdot \theta \cdot \ln(\frac{1 - \delta}{1 + \delta})
    ~\le~ \delta^{2}.
\end{align*}
Here the last step holds for $\delta \in [0, \frac{1}{2}]$ and can be verified via elementary algebra.

Based on \Cref{eq:GBB-independent-full:similarity,prop:GBB-independent:regret-per-round}, we can deduce that
\begin{align*}
    \Regret_{\+D^{1}} + \Regret_{\+D^{2}}
    &\textstyle ~\ge~ \frac{\theta}{16}\delta \cdot \big(2T - {\bb E}^{1}[T_{\+G'_{1}}] - {\bb E}^{2}[T_{\+G'_{2}}]\big)\\
    \mr{$T_{\+G_1'} + T_{\+G_2'} \le T$ almost surely}
    &\textstyle ~\ge~ \frac{\theta}{16}\delta \cdot \big(T - \big({\bb E}^{1}[T_{\+G'_{1}}] - {\bb E}^{2}[T_{\+G'_{1}}]\big)\big)\\
    \mr{\Cref{eq:GBB-independent-full:similarity}}
    &\textstyle ~\ge~ \frac{\theta}{16}\delta \cdot (T - \delta T \cdot \sqrt{T})\\
    \mr{$\theta = \frac{1}{13}$ and $\delta = \frac{1}{2}T^{-1 / 2}$}
    &\textstyle ~=~ \frac{1}{832}T^{1 / 2}.
\end{align*}
Then it is easy to see the claimed $\Omega(T^{1 / 2})$ lower bound. This finishes the proof of \Cref{thm:appendix:GBB}.
\end{proof}

\section{\texorpdfstring{$\Omega(T)$}{} {\WBB} Partial-Feedback Lower Bound for Independent Values}
\label{sec:appendix}

As mentioned before, for two-bit feedback and ``correlated values'', the previous work \cite[Theorem~6]{CCCFL24mor} claimed a linear lower bound $\Omega(T)$ for {\SBB} fixed-price mechanisms. Indeed, it is straightforward to check their proof holds more generally for {\WBB} fixed-price mechanisms, as we sketch below.

\begin{theorem}[{{\WBB} Partial-Feedback Lower Bound for Independent Values \cite[Theorem~6]{CCCFL24mor}}]
\label{thm:appendix:WBB}
\begin{flushleft}
In the ``independent values'' setting, every ``{\WBB} fixed-price mechanism with two-bit feedback'' has worst-case regret $\Omega(T)$.
\end{flushleft}
\end{theorem}

\begin{proof}[Proof Sketch]
Given a fixed mechanism $\+M$,
let $\nu^t_{z_1,z_2,\dots,z_{t-1}}$ denote the conditional distribution over the algorithm's price pair $(P_t,Q_t)$ at time $t$ conditioned on the feedback history $Z_s=z_s$ for $s \le t-1$ (we use $Z_s$ here to denote two-bit feedback for convenience). Let $A_t=\cup_{z_{1},\dots,z_{t-1}\in \set{0,1}^{2(t-1)}} \set{(p,q): \nu_{z_{1},\dots,z_{t-1}}(p,q)>0}$. Since $A_t$ is countable for any $t\in [T]$, then the union $A=\cup_{t\in[T]}A_t$ remains countable. By the uncountability of $\set{(x,x):x\in [0,4,0.6]}$, there exists a point $(a^*,a^*)\in [0.4,0.5]^{2}$ such that $(a^*,a^*)\notin A$. Therefore we can use $(a^*,a^*)$ to define the seller and buyer distributions as:
\[
    \textstyle
    f_S= \frac{1}{2}\delta_0+\frac{1}{2}\delta_{a^*} \quad\text{ and } \quad f_B= \frac{1}{2}\delta_{a^*}+ \frac{1}{2}\delta_1,
\]
where $\delta_x$ denotes the Dirac delta distribution. Under this hard instance:
\begin{enumerate}
    \item The optimal mechanism  is to take action  $(a^*,a^*)$, yielding  zero instantaneous regret.
    \item Since $(a^*,a^*)\notin A$, mechanism $\+M$ never selects this points in any round $t\in [T]$. Consequently, a constant instantaneous regret is incurred.
\end{enumerate}
Therefore the regret of $\+M$ is $\Omega(T)$. Full technical details can be found in the proof of  \cite[Theorem~6]{CCCFL24mor}.
\end{proof}

\end{document}